\documentclass[sigplan,10pt]{acmart}
\usepackage{preamble}
\acmConference[SOSP '26]
  {The 32nd ACM Symposium on Operating Systems Principles}
  {September 29--October 2, 2026}
  {Prague, Czech Republic}
\acmYear{2026}
\copyrightyear{2026}
\acmDOI{}
\acmISBN{}
\title{\sys{}: Append-Only Ledgers for (Mostly) Trusted Execution Environments}
\date{}

\author{Shubham Mishra}
\authornote{Work partially  done while at Azure Research, Microsoft}
\affiliation{%
    \institution{UC Berkeley}
    \country{}
}
\author{Jo\~{a}o Gon\c{c}alves}
\authornotemark[1]
\affiliation{%
    \institution{IST Lisbon \& INESC-ID}
    \country{}
}
\author{Chawinphat Tankuranand}
\affiliation{%
    \institution{UC Berkeley}
    \country{}
}
\author{Neil Giridharan}
\affiliation{%
    \institution{UC Berkeley}
    \country{}
}
\author{Natacha Crooks}
\authornotemark[1]
\affiliation{%
    \institution{UC Berkeley}
    \country{}
}
\author{Heidi Howard}
\affiliation{%
    \institution{Azure Research, Microsoft}
    \country{}
}

\author{Chris Jensen}
\affiliation{%
    \institution{Azure Research, Microsoft}
    \country{}
}

\renewcommand\footnotetextcopyrightpermission[1]{}

\begin{document}
\begin{abstract}
Distributed ledgers are increasingly relied upon by industry to provide trustworthy accountability, strong integrity-protection, and high availability for critical data without centralizing trust. Recently, distributed append-only logs are opting for a layered approach, combining crash-fault-tolerant (CFT) consensus with hardware-based Trusted Execution Environments (TEEs) for greater resiliency. Unfortunately, hardware TEEs can be subject to (rare) attacks, undermining the very guarantees that distributed ledgers are carefully designed to achieve.
In response, we present \sys{}, a new distributed consensus protocol that \textit{cautiously} trusts the guarantees of TEEs.
\sys{} carefully embeds a Byzantine fault-tolerant (BFT) protocol inside of a CFT protocol with no additional messages. This is made possible through careful refactoring of both the CFT and BFT protocols such that their structure aligns. \sys{} achieves performance in line with regular TEE-enabled consensus protocols, while guaranteeing integrity in the face of TEE platform compromises.
\end{abstract}

\ccsdesc[500]{Security and privacy~Distributed systems security}
\ccsdesc[500]{Computer systems organization~Dependable and fault-tolerant systems and networks}
\keywords{Trusted Execution Environment, Byzantine Fault Tolerance, Consensus, Distributed Systems.}
\maketitle
\pagestyle{plain}
\section{Introduction}

Tamper-proof shared logs are increasingly being used in industry to guarantee strong integrity protection and high availability. They spread \textit{trust} across multiple administrative domains, and rely on the fact that it is unlikely that all parties will be simultaneously compromised. Such logs have been used for applications like code transparency~\cite{Delignat23,scitt}, key recovery for encrypted messaging~\cite{svr3}, decentralized identity management~\cite{did-ccf}, privacy-preserving advertising~\cite{privacy-sandbox}, multi-party data sharing~\cite{PDO}, and confidential machine learning~\cite{apple-pcc, azure-cai}. Often, these logs are utilized to record a small amount of highly sensitive state, such as secret keys, serving as the root of trust for a much larger application: Signal for instance, implements its key recovery service across three clouds~\cite{svr3}, while Microsoft offers a code transparency service~\cite{cts-announcement}.

To build these tamper-proof logs, industry~\cite{ccf,svr3,oasis} and academia~\cite{Wang22,hybster,gupta2023bft,Achilles, dinis2023rr} are increasingly adopting a layered strategy in which they overlay cryptographic and algorithmic security atop additional hardware security.
These systems leverage hardware-based \textit{Trusted Execution Environments} (TEEs) such as AMD SEV-SNP~\cite{SEV-SNP}, Intel TDX~\cite{TDX}, or ARM CCA~\cite{CCA}. TEEs ensure integrity, even in the presence of machine compromise: they offer the guarantee that the nodes will never deviate from a given protocol implementation. To avoid data loss in the presence of benign crashes, tamper-proof logs then internally deploy a crash-fault-tolerant (CFT) consensus protocol, replicating data across $2f+1$ TEE-augmented replicas.

These solutions work well \textit{under the assumption} that TEEs are correct. Unfortunately, while secure hardware undoubtedly raises the barrier to attack, it is far from foolproof. Most TEEs can (with sufficient effort and access) be compromised. In fact, attackers have extracted user secrets from all major TEEs~\cite{sgx-step, foreshadow, one-glitch, fault-in-bus, malicious-vc, severed, badramsp25} and new vulnerabilities are continuously found.
The time between a vulnerability being discovered and patched is a huge window of opportunity for attackers, especially since these vulnerabilities affect \textit{entire classes of TEEs}. An attacker can potentially compromise all TEEs of the same type within a particular cloud region.

In short, TEE compromises are a rare but catastrophic event.
The current approach of placing full trust in TEEs leaves users vulnerable. The alternative approach of placing \textit{no} trust in TEEs is similarly unsatisfactory. It requires deploying a Byzantine-fault-tolerant consensus protocol (BFT), which comes at a significant performance cost. It requires increasing the replication factor to $3f+1$ and significantly increases latency. This paper asks: \textit{how can one preserve the performance of current append-only ledgers while providing meaningful guarantees in the presence of rare-but-catastrophic TEE compromises}.

To answer this question, we adopt the ethos of \textit{cautious trust}. In the common case, the system should make progress at a speed that matches the performance of current systems. However, the system should \textit{also} offer the ability for nodes to detect and recover to a consistent state when correctness violations occur.
Importantly,  this new \textit{auditing} process should pose no additional burden on clients. Prior works that focus on detecting (but not preventing) TEE compromises~\cite{ccf, microsoft-ccf,brandenburger2017rollback,peerreview} place significant burden on clients (or third-party auditors) who now need to be intimately aware of the "dirty" details of the distributed log.

To achieve this, we make two simple observations. First, we recognize that BFT consensus protocols naturally implement this detect-and-reconcile process.
These protocols guarantee that all nodes will agree on the same state, even if machines (and TEEs) act maliciously. Second, we observe that with careful refactoring, it is possible to \textit{embed} a BFT protocol inside a CFT protocol such that the BFT protocol will, with bounded delay, audit and reconcile all transactions that are committed by the CFT protocol.
In the absence of TEE compromises, committing a transaction offers the same guarantees as existing TEE-based consensus solutions~\cite{Achilles,hybster,gupta2023bft,Wang22}. In the presence of TEE compromises, the BFT protocol commit can then be used to audit the behavior of the original CFT protocol: TEE compromises that lead to divergent state are simply treated as instances of equivocation, which the BFT protocol can detect and reconcile.

Embedding a BFT protocol inside a CFT protocol in this way requires addressing three challenges. First, the number of phases in both protocol types differs.  CFT consensus protocols commit operations in a single phase (ignoring leader election) while BFT requires at least two. Second, quorum sizes are different. CFT consensus requires a simple majority of votes for committing, whereas BFT protocols require a supermajority. Finally, CFT protocols do not typically require additional cryptography such as signatures as all nodes are trusted. Instead they are essential for BFT correctness. %

\sys{}, our new distributed ledger, instantiates this idea. \sys{} \textit{commits} operations assuming crash faults only, but in the background \textit{audits} committed state against Byzantine attacks \textit{without any additional messages}.
Clients can continue to view a transaction as confirmed after it has been committed, as they do today. They may, however, choose to wait for an audit notification to be received if the transaction is highly sensitive or there is a known unpatched TEE vulnerability. \sys{} further guarantees that the auditing process never trails behind the regular commit index by more than a constant bound. Such an approach makes sure that the "blast radius" of an attack is at most the tail of the log that is committed but not yet audited. In comparison, current TEE-based protocols provide no such feature: an adversary can fork the entire log and remain undetected without active client intervention.

\sys{} achieves these guarantees while capturing more realistic failure patterns: the system consists of $\Pi$ independent platforms or trust domains, each of which consists of a variable number of nodes. Nodes may crash independently, while TEE vulnerabilities can compromise entire platforms.

\sys{} carefully repurposes existing BFT techniques to achieve these strong properties. First, \sys{} revisits Upright~\cite{clement09upright} to capture the notion of platforms: nodes crash independently, but platforms are compromised as a group. We coin this term \textit{platform fault tolerance}. %
Second, \sys{} \textit{incrementalizes quorums}: it leverages hash chaining (each batch is linked to all of its predecessors) to ensure voting for a batch is not just a vote for this batch but for all batches that precede it. Late votes or votes from later batches can be used to \textit{complete} old quorums. Thirdly, \sys{} combines pipelining with incremental quorums to ensure that the BFT auditing process never trails behind CFT consensus by more than a constant bound, despite requiring twice as many phases.

Naturally, combining all these techniques requires care. \sys{} must carefully design a new view change that accounts for these dual notions of transaction confirmation. \sys{} introduces the concept of \textit{view stabilization}: a simple extra phase in the view change that allows newly elected leaders to, in the absence of malicious compromise, reliably detect when batches could have committed.

Our results are promising. \sys{} gets best of both worlds. \sys{}'s throughput and \textit{commit} latency matches current TEE-enabled ledgers (including a signed version of Raft~\cite{raft} and production systems like Microsoft's CCF~\cite{ccf}), despite the stronger guarantees. \sys{}'s \textit{audit} latency and throughput instead matches the performance of a state-of-the-art BFT protocol, Autobahn~\cite{autobahn}.
Moreover, we show, on five applications, that users can easily navigate the trade-off of when to leverage commit guarantees (as they do today) relative to when they should upgrade to the new audit feature. In summary, we make three contributions:
\begin{itemize}[leftmargin=*]
\item We introduce platform-fault-tolerance, a close cousin of CFT and BFT that better reflects distributed trust today. (\S\ref{sec:platform})
\item We describe \sys{}, a new append-only log that embeds a BFT protocol inside a CFT protocol to achieve strong integrity properties with good performance. (\S\ref{sec:overview}, \S\ref{sec:protocol})
\item We implement and evaluate \sys{} on representative deployments and applications. (\S\ref{sec:eval})
\end{itemize}

\section{Towards Platform-Fault-Tolerance}
\label{sec:platform}

Distributed trust is increasingly used by application provi\-ders as a tool for building systems with strong reliability and integrity properties~\cite{ccf,svr3,oasis,cts-announcement,Wang22,hybster,gupta2023bft,Achilles,dinis2023rr}. These deployments typically rely on the existence of a tamper-proof ledger that remains correct in the presence of misbehaving nodes, where each node is assumed to be deployed in an independent trust domain, with at most $f$ nodes misbehaving.

\subsection{Why a new failure model?}

Careful study of these deployments reveals two common patterns emerging.  First, \textit{developers are increasingly deploying their distributed trust system within a cloud or across clouds}. Users of distributed trust want increased security, but do not want the operational overhead that comes with it. The cloud offers a convenient always-available black box. Azure Confidential Ledger~\cite{ccf-ledger}, for instance, provides users with the abstraction of a tamper-proof ledger by spreading trust across multiple nodes in Azure. Signal's key recovery service (SVR3)~\cite{svr3} spreads trust across multiple clouds and enclaves.

Second, \textit{distributed trust systems combine algorithmic and hardware protection}. The systems mentioned above all run a (modified) CFT protocol inside a TEE, e.g., Intel SGX, AMD SEV-SNP, etc.
The TEE is \textit{trusted} to preclude Byzantine attacks such as equivocation. The resulting consensus protocol can be simplified and optimized~\cite{gupta2023bft,chun2007a2m,hbft,trinc,Wang22,ccf,messadi2022splitbft,Achilles}, resulting in better performance.
Traditional failure models fail to capture this specificity in three key ways.
\begin{figure}
    \centering
    \includegraphics[width=\linewidth]{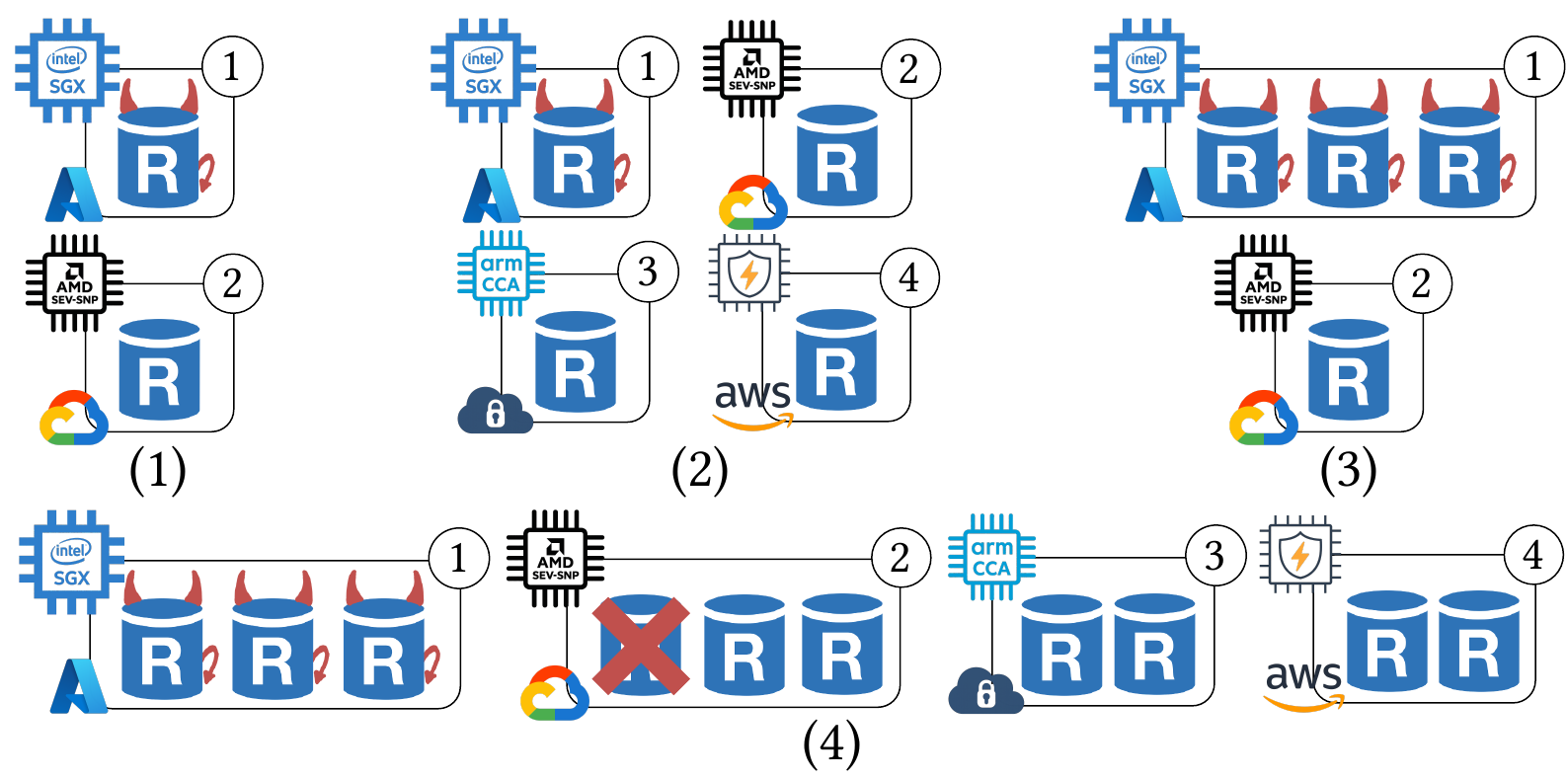}
    \caption{Platform Fault Tolerance. Independent crashes (X) but platform-wide TEE compromises (horns)}
    \label{fig:pft}
\end{figure}

\par \textit{1. They do not match the failure models of cloud-based trusted hardware.} Current failure models place either full trust or no trust in TEEs.  Both extremes are unrealistic. Trusted hardware significantly raises the barrier against Byzantine attacks, making them extremely unlikely and very costly to mount~\cite{badramsp25, batteringramsp26}. But TEEs still do not provide perfect guarantees. Even if the hardware manufacturer is not malicious, vulnerabilities are frequently found in all popular TEEs~\cite{googlesecurityreport,qiu2019voltjockey,chen2019sgxpectre,fei2021security,li2022systematic,li2021cipherleaks}.

\par \textit{2. They do not distinguish between nodes and trust domains.} Current failure models assume that node failures are independent. While this is mostly true for crash failures, TEE compromises are instead inherently correlated. A TEE vulnerability affects all nodes with the same hardware type~\cite{heckler}. %
Rather than considering nodes as independent fault domains, one should instead think of failures at the granularity of \textit{platforms}. Each platform consists of potentially many nodes, each describing a correlated failure domain: shared hardware, shared operators, shared physical location.
\par \textit{3. They conflate safety and liveness.} Cloud-based TEE deployments frequently choose to eschew liveness~\cite{svr3,angel2023nimble} in the presence of malicious compromise.  Malicious hosts can \textit{already} trivially violate liveness simply by refusing to disseminate client requests.
CSPs can also attack liveness by crashing nodes, manipulating hosts' notion of time, or even covertly reordering messages such that the network appears healthy but malicious attacks are easier to launch. The added replication factor cost to guarantee liveness is thus often not worth it. Unfortunately, current failure models offer no way to reason about safety and liveness independently.

\begin{table}[!t]
    \center
    \footnotesize
    \begin{tabular}{cl}
        \hline
        \multicolumn{1}{|c|}{\textbf{Notation}} & \multicolumn{1}{l|}{\textbf{Meaning}} \\ \hline \hline
        \multicolumn{1}{|c|}{$\Pi$}         & \multicolumn{1}{l|}{Total number of platforms}        \\ \hline
        \multicolumn{1}{|c|}{$N$}         & \multicolumn{1}{l|}{Total number of nodes}        \\ \hline
        \multicolumn{1}{|c|}{$N^i$}         & \multicolumn{1}{l|}{Total number of nodes in the $i^{th}$ platform}       \\ \hline
        \multicolumn{1}{|c|}{$\pi_{safe}$}         & \multicolumn{1}{l|}{Max number of malicious platforms (Safety) }        \\ \hline
        \multicolumn{1}{|c|}{$f_{safe}$}         & \multicolumn{1}{l|}{\makecell[l]{Max number of Byzantine failures (Safety) \\ Assume the sequence $\{ N^1, N^2, ..., N^\Pi \}$ is sorted in \\ descending order, then, $f_{safe} = \sum_{j = 1}^{\pi_{safe}} N^j$}}        \\ \hline
        \multicolumn{1}{|c|}{$\pi_{live}$}         & \multicolumn{1}{l|}{Max number of malicious platforms (Liveness) }        \\ \hline
        \multicolumn{1}{|c|}{$f_{live}$}         & \multicolumn{1}{l|}{\makecell[l]{Max number of Byzantine failures (Liveness) \\ Assume the sequence $\{ N^1, N^2, ..., N^\Pi \}$ is sorted in \\ descending order, then, $f_{live} = \sum_{j = 1}^{\pi_{live}} N^j$}}        \\ \hline
        \multicolumn{1}{|c|}{$c$}         & \multicolumn{1}{l|}{Max number of crash failures (Liveness)}        \\ \hline
         \multicolumn{1}{|c|}{$u$}         & \multicolumn{1}{l|}{$u=f_{live} + c$. Total number of node failures (Liveness).}        \\ \hline

    \end{tabular}
\caption{Notation summary PFT}
\label{tab:pft}
\end{table}
\subsection{Formalizing Platform Fault Tolerance}

To address these issues, we propose a new failure model,  platform-fault-tolerance (PFT). PFT distinguishes nodes from platforms, which correspond to correlated failure-domains, and better separates safety and liveness thresholds. A fault-model has two parts: language to model failure thresholds, and correctness guarantees that hold when these failure thresholds are met. We describe both in turn.

\par \textbf{Failure Modeling} In line with past work~\cite{lamport2006lower,Dutta2005BestCaseCO,malek2005fault,clement09upright}, we distinguish between crashed nodes and Byzantine nodes.
A \textit{crashed node} follows the protocol (TEE integrity guarantee is preserved) but stops sending/receiving messages. Nodes crash for a variety of reasons, including traditional hardware failures or software bugs. When a TEE is compromised, this voids the guarantee that the enclave will execute the attested code; the machine may then deviate \textit{arbitrarily} from the protocol. We refer to this as a \textit{Byzantine} failure.

\par \textbf{PFT Guarantees} A PFT system consists of N nodes, split into $\Pi$ platforms $\{p_1, p_2, ..., p_\Pi\}$, with platforms representing a distinct fault domain, be it a distinct cloud operator, cloud region, or hardware stack.
All nodes within a fault domain \textit{may} be compromised by an adversary.
Fundamentally, the notion of platform hinges on the operator's belief of what constitutes fault independence. Some system administrators may believe that different cloud platforms, even if they run the same TEE-type, will not both be compromised together, perhaps because the attacks require physical access. Others will consider two platforms independent only if they run on different clouds, different enclaves, and have a fully independent software stack. Failure tolerance against crash failures or reboots is then achieved by deploying multiple machines within each platform: a platform $p_i$ consists of $N^i$ nodes.

\begin{definition}
A PFT system of $N$ nodes and $\Pi$ platforms is \begin{itemize}
\item safe up to $\pi_{safe}$ platform compromises ($f_{safe}$ total compromised nodes)
\item  live up to $c$ crash failures and $\pi_{live}$ platform compromises ($f_{live}$ total compromised nodes)
\end{itemize}
\end{definition}

PFT systems make two separate \textit{and orthogonal} claims to users. First, a PFT system should be safe: it should expose a single totally ordered ledger. This guarantee holds up to $\pi_{safe}$ platform compromises. Second, the system should remain live: all transactions should eventually be committed. Liveness should be guaranteed up to $\pi_{live}$ platform compromises and $c$ node crashes. Crash failures are independent of platforms and can only impact liveness; PFT thus considers that the system \textit{as a whole} will suffer from at most $c$ crash failures.

Crucially, $\pi_{safe}$, $\pi_{live}$, and $c$ are orthogonal parameters that can be configured independently. A system configured with $\pi_{live}=1$ and $\pi_{safe}=2$ may fail to respond to users if it experiences two platform compromises ($\pi_{live}<2$), but any responses returned will be correct (safety is preserved). Conversely, if configured with $\pi_{live}=2$ and $\pi_{safe}=1$, the system will, with two platform compromises, continue responding but the responses returned may not be correct (liveness is preserved, safety is not) \footnote{In this paper, we always configure \sys{} to remain safe.}.

Allowing for safety and liveness to be configured independently allows for increased flexibility in choosing the number of platforms and replica count within each platform. This is reminiscent of the Upright model~\cite{clement09upright, malek2005fault,lamport2006lower}, which also considers safety and liveness as orthogonal as a way to reduce the replication factor. %
Specifically, solving consensus under PFT requires the total number of nodes $N$ to be:
\begin{theorem}
\label{theorem:node-count}
$N\geq 2(c+f_{live}) + f_{safe} + 1$

where $f_{safe}$ and $f_{live}$ represent, respectively, the maximum number of nodes that can be compromised as a result of $\pi_{safe}$ or $\pi_{live}$ platforms being compromised.
\end{theorem}
We defer the proof to \Cref{sec:safety}.
For simplicity, we write $u=f_{live} + c$.
Note that, in the traditional CFT setting, $f_{safe} = f_{live} = 0$, yielding $N \geq 2c + 1$, the CFT bound. Similarly, in the traditional BFT setting, $u = f_{safe} = f$, which yields $N \geq 3f + 1$, the same bound that BFT protocols expect. PFT never requires more nodes than that required by traditional BFT protocols.

This result provides two key benefits:
\\
\textit{1. It allows us to quantify the degree to which sacrificing liveness reduces replication overhead.}
Increasing tolerance to malicious compromise for providing liveness ($f_{live}$) requires increasing the total number of nodes by twice the amount necessary to improve safety ($f_{safe}$). Hence why many cloud-based ledgers choose to eschew liveness. The ability to strike this tradeoff is especially important in TEE-setups where the limited number of cloud providers and TEEs limit the number of platforms one can deploy. Assume, for simplicity, that all platforms consist of a single node so $\pi_{safe} = f_{safe}, ~\pi_{live} = f_{live}$ and $\Pi = N$. To guarantee safety in the presence of a single platform compromise, so, $\pi_{safe} = f_{safe} = 1$, two platforms suffice (Fig. \ref{fig:pft}(1)). Remaining live in the presence of a platform compromise ($\pi_{live}=f_{live}=1$) requires adding not one but two platforms $\Pi=N=4$ (Fig. \ref{fig:pft}(2)).

This asymmetry stems from needing to prevent conflicting operations from committing: commit quorums for operations $o$ and $o'$ must always intersect in one honest node (an honest node will never commit both). The challenge arises when one still wants to make progress when  $u=c+f_{live}$ replicas may not respond. This bounds the maximum size of a quorum for $o$ to $N-u$ replicas. $Q_o$, for instance, may only have included nodes $\{n_1,n_2,n_3\}$. The quorum $Q_{o'}$ (where $N-u\geq|Q_{o'}|$) formed for $o'$ may thus include up to $u$ nodes that do not know about $o$ (here $n_4$) as well as up to $f_{safe}$ malicious nodes that lie about $o'$ (here $n_1$).  $Q_{o'}$ is only guaranteed to see $o$ if it includes an honest node that definitely have seen $o$: $ |Q_{o'}| - (u + f_{safe})\geq 1$. Setting $|Q_{o'}|$ to its maximum size of $N-u$ yields $(N - u) - (u + f_{safe})\geq 1$. Rearranging the terms yields \cref{theorem:node-count}. Concretely, if $Q_{o} = \{n_1,n_2,n_3\}$ and $Q_{o'}=\{n_1,n_3,n_4\}$, $Q_{o'}$ is guaranteed to learn about $o$ through $n_3$ even though $n_1$ may lie and $n_4$ never knew about $o$.

\textit{2. It captures failure correlation without explicitly mentioning platforms.} The benefits of this decision are more subtle as they pertain to algorithm design. Omitting direct reference to platforms allows us to design a \textit{single} reconfiguration algorithm (see \Cref{sec:reconfiguration}) to reconfigure both nodes and platforms.  The equation instead captures failure correlation within a platform through parameters $f_{safe}$  and $f_{live}$: these are computed according to the maximum number of machines that become vulnerable when platforms get compromised. In Fig. \ref{fig:pft}(3) for instance, if one wants to tolerate one platform compromise for safety ($\pi_{safe}$), one should assume that, in the worst case, the largest platform will be compromised. As such $f_{safe}=3$ and $N=4$. Seeking to additionally tolerate $\pi_{live}=1$ (making $f_{live} = 3$) and $c=1$ brings up
the total number of nodes to $N=3+2*(3+1)+1 = 12$ (\cref{fig:pft}(4)).

 \par \textbf{Limitations} As the definition of what constitutes a platform depends on the operators belief about fault independence, unknown shared vulnerabilities across platforms may lead to more than $\pi_{safe}$ platform compromises. This is akin to exceeding the $f$ threshold in the BFT or CFT model, and is fundamental. Nonetheless, PFT does provide the ability to evolve platforms through reconfiguration (see \Cref{sec:reconfiguration}) as beliefs about fault independence change over time.
For example, an operator may have initially considered a (TEE, region) pair as an independent platform.
If they later discover a vulnerability in a shared library used in all of the platforms, they may consider reconfiguring to add platforms that do not share that dependency.

\section{\sys{} Overview}
\label{sec:overview}

\sys{}, our new PFT ledger, recognizes that platform compromise of TEEs is a rare but possible event. \sys{} assumes that, in the common case, TEEs are secure and runs a CFT protocol atop the TEEs to commit transactions. \sys{} additionally embeds a BFT protocol that audits transactions off the critical path. This allows \sys{} to detect malicious behavior and recover from rare platform compromises at no extra cost.

\subsection{API and Guarantees}
\sys{} supports two notions of transaction confirmation. \textit{Committed} transactions are guaranteed to be durable in the absence of TEE compromises. This is equivalent to the \textit{current} guarantees offered by TEE-augmented systems such as Signal's SVR3 or Microsoft's CCF. Formally:
\begin{definition}[Committed Transactions]
With no platform compromises, no two correct replicas consider transactions $C$ and $C'$ ($C\neq C')$
committed at the same position in the log.
\end{definition}

When TEE compromises do occur, committed transactions can be rolled back. \sys{} offers an additional \textit{audit} API that guarantees safety in the presence of such attacks. If a transaction has been audited, it can never be rolled back. %
\begin{definition}[Audited Transactions]
Up to $\pi_{safe}$ platform compromises, no two correct replicas consider transactions $C$ and $C'$ ($C\neq C'$)  audited at the same position in the log.
\end{definition}
Committed transactions are asynchronously audited.
This process bounds the degree of trust that one has to place in TEEs: a TEE compromise only affects transactions that have been committed but not yet audited, keeping the window of vulnerability small. In contrast, existing protocols that place full trust in TEEs would allow the entire log to be rolled back in the presence of even a single TEE compromise.

Choosing how to leverage this API is determined by the application. There are several options. As commit is the guarantee currently offered by major distributed trust deployments (Azure Confidential Ledger~\cite{ccf, ccf-ledger} and all prior TEE-assisted protocols), we assume that most clients will continue to consider commit guarantees as good enough and thus make no changes to their codebase. Pessimistic clients may instead take advantage of the added security and additionally wait for transactions to be audited, but at the cost of higher latency. This decision can evolve over time: clients can switch from commit to audit if a new vulnerability is discovered, and switch back once it has been patched. The decision can also be made on a per-transaction basis: sensitive transactions may, for instance, require auditing before they are considered final by clients (reconfiguration, revoking permissions, high-value bank transfers, etc.). We discuss several use cases in \S\ref{sec:eval}.%

\subsection{Protocol Overview}
\begin{figure}[!t]
    \centering
    \vspace{-0.4cm}
    \includegraphics[width=1.0\linewidth]{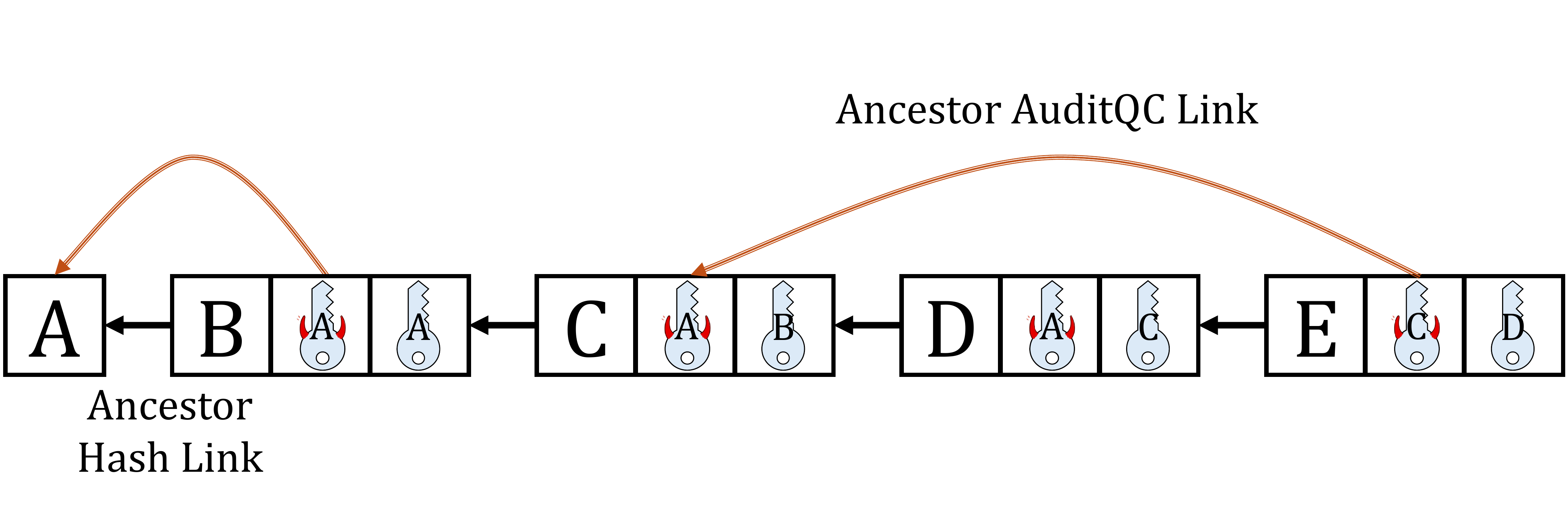}
    \vspace{-1cm}
    \caption{Log, Batches, and QCs}
    \label{fig:log}
\end{figure}
\sys{} shares the structure of other consensus protocols. It consists of three components: steady-state with an optional fast path, a view-change, and commit logic. %
\sys{} proceeds in a sequence of \textit{views}, with a designated leader per view. Each view consists of assembling one or more \textit{quorums} (set of votes) from sufficiently many replicas. As is standard, \sys{} assumes that the adversary cannot break standard cryptographic primitives. \sys{} similarly makes standard timing assumptions: it makes no assumption about message arrival time for safety but requires bounded delivery for liveness. \sys{} further leverages the lift-and-shift approach that new TEEs enable. The protocol code is agnostic to the TEE platform. %
\sys{} directly inherits the guarantee that nodes will always follow protocol-spec in the absence of platform compromise, but recognizes that they may deviate arbitrarily from it when attacked.

\sys{} carefully repurposes existing BFT techniques to embed a BFT protocol with no additional message overhead in a way that guarantees the same audit throughput as commit throughput. \sys{} specifically leverages three techniques: 1) hash chaining 2) pipelining 3) fast audits.

First, \sys{} leverages hash chains to \textit{incrementalize} quorum formation.
It cryptographically links each batch of transactions to all its predecessors in the chain and thus associates each new proposal with a unique voting history. A vote on a batch is also a vote for all batches (and thus all quorums) that precede it in the chain. Late votes, as well as votes received from a \textit{different} set of replicas for later operations, can be used to "complete" prior commit quorums and asynchronously harden them into audit quorums.
Second,
\sys{} combines pipelining with incremental quorums to ensure that the BFT auditing process never trails behind CFT consensus by more than a constant factor, despite requiring twice as many phases. Pipelining enables message sharing: the second phase of the BFT protocol can be layered atop the next batch of the CFT algorithm.
Finally, \sys{} supports an audit fast path that allows it to optimistically audit commit decisions in a single round-trip during gracious intervals~\cite{aardvark}.

Naturally, combining all these techniques requires care. \sys{} must revisit the traditional view change logic. BFT view changes only preserve transactions that have been audited. In \sys{}, branches that have been committed but not audited must also be preserved across views (in the absence of malicious actors). The asymmetry between the smaller, unsigned CFT quorum and the larger signed BFT quorum makes this challenging. To this effect, \sys{} introduces a new concept: \textit{view stabilization}. View stabilization adds an extra round to the view change to ensure that, in the absence of malicious nodes, if a leader commits a branch, future leaders will always reject older conflicting branches.

We formally prove \sys{} correct and include a formal TLA+ specification in \Cref{subsec:tlaplus}.

\section{Protocol Description}
\label{sec:protocol}

We now describe \sys{} in more detail. Note that, for ease of exposition, we omit writing the standard checks performed by all replicas (are the messages well-formed, signatures valid, and from the right view?) but assume that all replicas check that messages are valid prior to processing.

\begin{figure*}[t]
    \centering
\includegraphics[width=\linewidth]{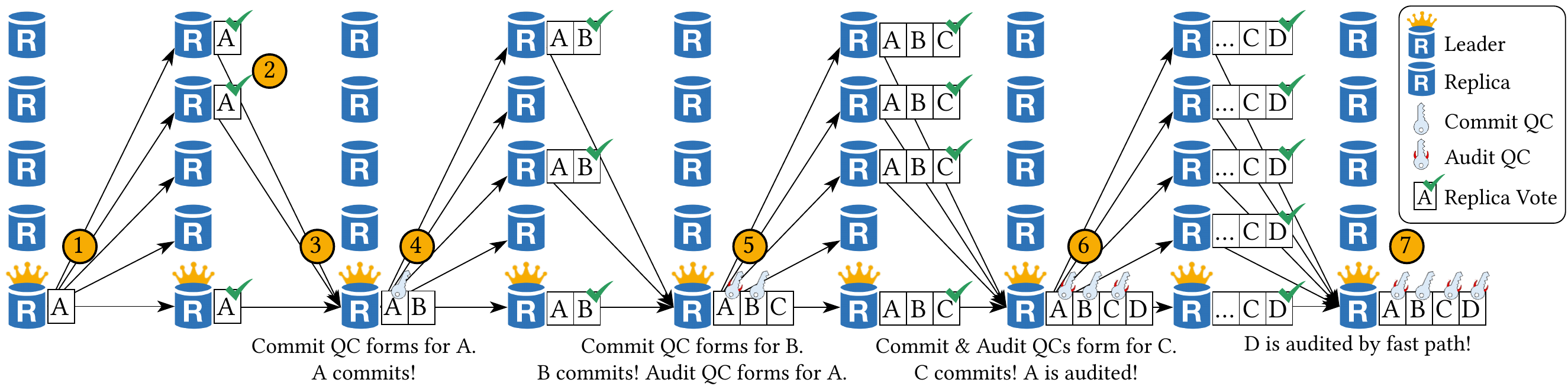}
    \caption{Steady-state with five replicas ($u = f_{safe} = 1$) and four batches (A,B,C,D). The example uses more nodes than necessary to make the differences in commit and audit QC apparent ($N=5\geq2u+f_{safe}+1$). Quorum sizes for commit, slow and fast audits are $3, 4,$ and $5$ respectively.  }
    \label{fig:steadystate}
\vspace{-0.2cm}
\end{figure*}

\subsection{Notation}
\begin{algorithm}
    \begin{algorithmic}[1]
    \State \textbf{$v_{curr}$}: Node's current view
    \State{\textit{$b_{curr}$}: Local branch seen by the node.}
    \State{\textit{highAuditQC}: Highest auditQC known to node} %
    \State{\textit{highCommitQC}: Highest commitQC known to node} %
    \State{\textit{auditIndex}: Audit index}
    \State{\textit{commitIndex}: Commit index}
    \State{\textit{pendingQCs[][]}: Mapping of (index,replica id) to votes received}
    \end{algorithmic}
    \caption{Node State Cheat Sheet}
    \label{alg:state}
\end{algorithm}
\par \textbf{Batch Format.} As is standard, \sys{} batches client requests into \textit{batches}, with each batch containing
a hash pointer to its parent batch. %
A batch $B$ contains the following information
$B=\langle v,i, p, commitQC_{anc},
auditQC_{anc}, b \rangle_{Opt<sig>}$. $v$ denotes the view the batch was
proposed in, $i$, the batch index number (its position in the ledger),
$p$, a pointer to the hash of its parent batch, the commit/audit quorum certificates certifying an ancestor
batch of $B$: $commitQC_{anc}$
and $auditQC_{anc}$; and $b$, the sequence of client transactions. Batches may optionally be signed by the leader ($Opt<sig>$).

\par \textbf{Batch extension and conflicts}. The sequence of parent pointers links batches into a chain. We formally define an ancestor (or predecessor) of $B$ to be any batch $B_{anc}$ for which a sequence of parent pointers links $B$ to $B_{anc}$.
$B'$ extends (or is a child of) a batch $B$ %
if $B$ is an ancestor of $B'$.  In contrast, $B'$ conflicts with $B$ if neither extends the other. %
\par \textbf{Quorum Certificates}. A Quorum Certificate, $QC=\langle Q, B \rangle$ consists of a set $Q$ of $k$ \textsc{Resp} messages from distinct nodes, and a digest of the batch $B$ being voted on or for one of its descendants.
\sys{} considers two types of quorum certificates, \textit{commit quorum certificates (commitQCs)} and \textit{audit quorum certificates (auditQCs)}. Audit quorum certificates must necessarily contain at least $k =~$\supermajority \textit{signed} \textsc{Resp} messages.  Commit quorum certificates relax this requirement. They form after receiving $k=$\majority \textsc{Resp} messages, with or without signatures.
The lack of signatures creates an interesting asymmetry that \sys{} must carefully reason about in the protocol: unsigned messages can be spoofed during Byzantine attacks. %

\par \textbf{Branches}  Each branch individually forms a chain of batch\-es that extend each other. \sys{}'s chaining allows us to uniquely represent a branch with a single batch, as each batch contains a pointer to its parent as well as the highest QCs known to the leader when the batch was formed.
Conflicting branches, as the name suggests, contain batches that conflict with batches of the other branch.
Fig. \ref{fig:log} depicts a branch with multiple hash-chained batches and QC links among them.

\subsection{Steady-state} We first describe the steady-state. We include redundant information in messages for clarity, but optimize away redundancy as part of our concrete implementation (\S\ref{subsec:impl-notes}). Node state is summarized in \Cref{alg:state}.

\protocol{\textbf{1: L} $\rightarrow$ \textbf{R}: Leader $L$ broadcasts \msg{Append-Entry}{B}.}
\par The leader $L$ for view $v$ forms a new batch $B=\langle v_{curr}, i, p,\\
highCommitQC, highAuditQC, b \rangle_{Opt<sig>}$.
  $v_{curr}$
 is the leader's current view, $p$ is the hash of the last batch in the chain, and $i$ is the sequence number of the batch.  $highCommitQC$ and $highAuditQC$ are the highest commit and audit QCs known to the leader. $b$ is the set of client transactions. For efficiency, \sys{} has the leader sign every $s$ batches only. Signatures are only necessary for auditing. Thanks to hash-chaining, auditing a single batch transitively audits all of its ancestors. $L$ then updates $b_{curr}$ to $B$, before broadcasting the \textsc{Append-Entry} message. Finally, $L$ inserts $B.i$ in $pendingQCs$ to asynchronously track the votes $B$ receives.

\protocol{\textbf{2: R} $\rightarrow$ \textbf{L}: $R$ processes \msg{Append-Entry}{B} and votes.}

$R$ checks that \one $L$ is the current view's leader, \two $R$ has not yet voted for this sequence number in this view. If $B$ is the direct child of $R$'s current branch ($B.parent=b_{curr}$), $R$ then moves directly to voting. If $R$ is \textit{not} the direct child ($B.parent.i >b_{curr}.i$), $R$ retrieves all batches between $b_{curr}$ and $B$ from the leader using the hash of $B.parent$ embedded with $B$ ($B.p$). Thus $R$ can retrieve $B.parent$ from the leader (and recursively, any ancestors).
In this process, if $R$ discovers that $B$ conflicts with $b_{curr}$, it will call for a \textsc{view-change} as $L$ may be malicious. Otherwise, $R$ appends all retrieved batches to $b_{curr}$ and moves to voting.

$R$ updates its local branch ($R.b_{curr} \gets B$) and both its high commit QC ($R.highCommitQC \gets B.highCommitQC$) and high audit QC ($R.highAuditQC \gets B.highAuditQC$). It inserts $B.i$ in $pendingQCs$ and returns a vote $\langle \textsc{Resp}, B, \allowbreak~Opt<sigs>\rangle$, which optionally contains, if $B$ is signed, an array $sigs$ of signatures on all \textit{signed} batches in $pendingQCs$, up to and including $B$.

 \protocol{\textbf{3: L}: $L$ assembles quorum of \textsc{Resp}{} messages.}
\par The logic of how the leader $L$ assembles a quorum is key to \sys{}'s ability to embed a Byzantine commit protocol (referred to as the audit) inside of a CFT protocol without additional messages.
In \sys{}, the leader tracks the set of \textsc{Resp} messages received for every batch $B$ in $pendingQCs$ and stores all new \textsc{Resp} messages according to their log entry. As batches in \sys{} are \textit{chained} (a batch $B$ uniquely references all of its parents), a vote for $B$ is also a vote for any ancestors that $B$ descends from. $L$ updates $pendingQCs$ accordingly: for all sequence numbers $j$ in $pendingQCs$ with $j \leq \textsc{Resp}.B.i,~ pendingQCs[j][R.id] \gets \textsc{Resp}$.

Next, the leader must determine when it has received sufficiently many votes to form either a commit or an audit quorum.
The leader constructs a commit quorum for $B$ when $pendingQCs$ contains \majority \textsc{Resp} messages for $B$. %
Similarly, the leader can construct an audit quorum when $pendingQCs$ contains at least \supermajority  \textit{signed} \textsc{Resp} messages.
Each time the leader detects that a new commit or audit quorum $QC =\langle Q, B \rangle$  has formed for a higher batch $B$, it updates (respectively) the $highCommitQC$ or $highAuditQC$.

This cumulative, asynchronous quorum formulation logic is powerful. It allows \sys{} to harden commit quorums into audit quorums using either 1) late votes from other replicas that were not needed for the commit quorum or 2) votes from distinct replicas in subsequent batches! In either case, no additional messages are needed. Moreover, this process is asynchronous. The leader never waits on a commit or audit quorum to send the next \textsc{Append-Entry}.

\par \textbf{Confirmation Rules} Recall that \sys{} is a pipelined protocol: \textsc{Append-Entry} messages
contain information about both the new batch and also prior batches that may now have committed.
The commitQCs and auditQCs sent as part of the batch inform replicas of which prefix of the log has formed a QC. Upon updating their local $highCommitQC$ and $highAuditQC$, each replica locally applies \textit{confirmation rules} to  determine the new status of a batch. The logic differs between commit and audit, as they uphold different safety guarantees: committed operations must persist across crashes but may be forgotten when nodes are compromised. Audited operations must instead persist up to $f_{safe}$ machine compromises. This stricter requirement places additional constraints on the Byzantine audit as it must remain robust to 1) spoofed messages 2) equivocation 3) state fabrication.

\par \textit{1. Commit} A batch is committed if it is replicated at a majority of nodes \textit{in the same view}. This follows from traditional crash fault tolerant logic: in the absence of Byzantine faults, any two quorums of $\lfloor\frac{N}{2}\rfloor + 1$ nodes intersect in at least one non-faulty node.
Nodes can thus mark the batch committed as soon as they learn that a commit QC has formed for $B$ or a descendant of $B$. Note that, since commit QCs do not need signed votes, a malicious leader can spoof a commit QC and pretend it was replicated to a majority. This is fine as committed transactions need not remain durable in the presence of such malicious actors.
\par \textit{2. Slow Audit}
Auditing a batch is a two-hop process (\cref{fig:log}). A node considers a batch $B$ audited if, \textit{in the same view}, an audit QC $Q$ forms for $B$ or descendants of $B$ and a second QC forms for the batch that first included Q (or its descendants).  The first QC ensures that no conflicting batch can be proposed as part of the same view (\textit{non-equivocation}): two QCs of \supermajority signed votes must necessarily intersect in at least one correct replica, and no correct replica would vote for conflicting batches in the same view. The second QC instead guarantees that this decision persists across views (\textit{durability}): a correct node with knowledge of $Q$ will be included in all subsequent view changes, allowing the leader to repropose the batch.
\par \textit{3. Fast Audit} In general, considering a transaction audited after a single audit QC is insufficient to ensure that the batch will persist across views. However, we observe that the leader usually, barring node crashes, does \textit{eventually} receive the full set of $N$ votes from distinct replicas. This can happen either directly through signed votes for a given batch $B$, or indirectly, for a batch $B'$ that extends $B$. This single, larger audit quorum of $N$ votes directly guarantees durability across views (any subsequent leader will observe at least $N - u - f_{safe}$ votes for this proposal). Upon seeing that $N$ signed votes have been received for $B$ or a descendant of $B$, the replica can thus immediately consider the transaction audited. Note that, if \sys{} has been configured with $N - u - f_{safe} \leq f_{safe}$  (a configuration that optimizes liveness over safety), fast audit is not possible. Byzantine replicas could spoof a branch, making it look like it was audited by fast path (as $f_{safe}$ Byzantine replicas can manufacture $f_{safe} \geq N - u - f_{safe}$ matching votes).
\par \textbf{Worked Example} Consider the execution described in  \cref{fig:steadystate}.
For simplicity, each replica is its own platform, and all replies are signed.
We set $u = 1$ and $f_{safe} = 1$. This requires $N \geq 4$ nodes. We set $N = 5$ to bring out the differences in quorum sizes. For commits, the quorum size is $\lfloor \frac{N}{2} \rfloor + 1  = 3$, and for slow audits, quorum size is $N - u = 4$. Fast audits require all $N = 5$ votes.
Leader $R_5$ broadcasts batch $b_A$ (\tc{1}).
$R_1$, $R_2$ and $R_5$ vote to accept it (\tc{2}).
Having received three votes (\tc{3}), $R_5$ commits $b_A$ and propagates the commitQC with the next batch $b_B$ (\tc{4}), allowing other replicas to also commit. $R_5$ similarly commits $b_B$ with votes from $R_1$, $R_3$, and $R_5$.
Receiving a vote from $R_3$ on $b_B$ and \textit{its entire history} allows $R_5$ to upgrade $b_A$'s quorum from three to four votes, thus creating an audit QC for $b_A$ (\tc{5}).
$R_5$ includes this auditQC for $b_A$ (and the commitQC for $b_B$) with the next batch $b_C$ .
$R_5$ receives four votes for batch $b_C$ and immediately forms an auditQC.
Note that the auditQC on C is also a QC on the auditQC for A. This suffices for $R_5$ to audit $b_A$ on the slow path (\tc{6}). This information is propagated with the next batch $b_D$.
Finally, $R_5$ receives votes from all five replicas, and can directly fast audit $b_D$ (\tc{7}).

\subsection{View-Change}
A faulty leader may stall progress by failing to commit or audit transactions. As is standard, \sys{} relies on timeouts to detect these situations and elects a new leader through a \textit{view change}. Traditional BFT view changes ~\cite{giridharan2023beegees, gueta2019sbft, hotstuff, jolteon-ditto, PBFT} usually have three steps: \one the view-change trigger, where sufficiently many replicas revolt against the current leader, \two the branch selection logic, where the new leader selects the branch to extend, and \three the dissemination step, where the new leader broadcasts the new view.
\sys{} must, however, handle an additional challenge: the presence of both committed and audited transactions.
Without care, committed-but-not-yet audited transactions can be lost due to a combination of their smaller quorum size and legitimate network asynchrony.
To address this issue, \sys{} adds a fourth step called \textit{View Stabilization}. View stabilization ensures that, in the absence of malicious actors,  new views will always extend a branch containing all committed transactions.

\par \textbf{Protocol Logic}  A replica starts a timer for view $v$ upon receiving a \textsc{new-view} message, announcing the start of the new view. It resets the timer every time an \textsc{append-entry} message includes a new auditQC.

\par \protocol{\textbf{1: R} $\rightarrow$ \textbf{R}: Replica $R$ broadcasts \textsc{view-change}.}

A replica $R$ whose timer expires, broadcasts a message \msg{View-Change}{v_{curr}, b_{curr}}$_R$.
The message contains the highest auditQC seen by the replica as a part of its current branch $b_{curr}$.  $R$ will ignore future messages in view $v_{curr}$.

\par \protocol{\textbf{2: R} $\rightarrow$ \textbf{R}:  $R$ receives \textsc{view-change} and advances view}

Upon receiving $f_{safe} + 1$ \textsc{view-change} messages in view $v$, $R$ concludes that this view-change request is legitimate
(at least one correct node requests it). It stops responding to messages in view $v_{curr}$, and broadcasts its own \textsc{view-change} message. This amplifying step ensures that if one correct replica triggers a view-change, all correct replicas will eventually join the new view.
Upon accepting \supermajority distinct \textsc{view-change} messages for $v$, $R$ advances its local view $R.v_{curr} = v+1$ and starts a timer for $v+1$.
If $L$ is the leader for view $v+1$, $L$ must additionally select a branch to extend. There might, however, be several conflicting branches in the set $\mathcal{V}$ of \textsc{view-change} messages received. $L$ uses a \textit{branch selection rule} to decide on which branch to extend (more detail soon).  The leader then broadcasts a \textsc{new-view} message with the chosen branch and $N - u$ \textsc{view-change} messages as proof. In traditional BFT consensus protocols, this step normally marks the end of the view-change.  The asymmetry between commit and audit quorums in \sys{} requires additional care in \sys{}. If we were to stop there, a new leader could commit blocks on a branch that does not possess a recent auditQC. In future views, another leader, upon seeing a conflicting branch with an older auditQC would have to extend it (that branch may have been audited), causing committed transactions to be lost (\Cref{sec:vs}).

\par \protocol{\textbf{3: L} $\rightarrow$ \textbf{R}: Leader $L$ broadcasts \textsc{New-View} and waits.}

To initiate a view, the leader first replaces its own local branch, if conflicting, with $B_{chosen}$. It then constructs a \textsc{New-View} message, which is a special \textsc{Append-Entry} message that contains a blank signed batch $S$ in view $L.v_{curr}$ that extends $B_{chosen}$. The message additionally contains the set $\mathcal{V}$ of \supermajority \textsc{View-Change} messages that the leader used for branch selection.
The leader then broadcasts \msg{New-View}{S,~\mathcal{V}}  to all other replicas.
Unlike other pipelined \textsc{Append-Entry} messages, this \textsc{New-View} message is blocking: $L$ waits for \supermajority votes for $S$ before proposing any new batch.

\protocol{\textbf{4: R} $\rightarrow$ \textbf{L}: Replica $R$ processes \textsc{new-view} and votes.}

$R$ first checks that \one the leader is for the current view, \two that $S$ extends the correct $B$ given the set $\mathcal{V}$ of \textsc{view-change} messages and the branch selection rule.
 $R$ then updates its local branch $R.b_{curr} \gets S$ and both its high commit and audit QC ($R.highCommitQC \gets S.highCommitQC$, $R.highAuditQC \gets S.highAuditQC$). Finally, it returns a signed \textsc{Resp} message to $L$. While the processing of \textsc{new-view} and \textsc{Append-Entry} are similar, there is one key difference: $L$'s chosen branch may conflict with $b_{curr}$ (due to equivocation or asynchrony).
In that case, $R$ \textit{rolls back} batches from $b_{curr}$ until $b_{curr}$ does not conflict with $B_{chosen}$.
This forced removal of batches is key to \sys{}'s reconciliation capabilities.

\protocol{\textbf{5: L}: Leader $L$ waits for \supermajority \textsc{Resp}{} messages.}

The leader synchronously waits for \supermajority \textsc{Resp} messages and forms an audit quorum on $S$. It updates $L.highCommitQC$ and $L.highAuditQC$ and starts processing batches as normal.
Only at this point does the leader consider the view \textit{stabilized}.
Effectively, \textit{view stabilization phase} blocks a new leader from adding new batches until it has formed an auditQC for the chosen branch. Committed entries will thus always extend the latest auditQC. As the view-change QC  intersects the commit QC in at least one node, the leader is guaranteed to preserve all committed branches by selecting the branch with the highest audit QC (in the absence of malicious replicas).

\par \textbf{Branch Selection Rule} Our discussion so far has not described  how $B_{chosen}$ is chosen. View stabilization and the branch selection rule work closely together to ensure that audited branches are preserved always, and committed branch\-es are preserved in the absence of malicious replicas.  Branch selection consists of four filtering steps, applied in order.

\protocol{\textbf{Rule 1.} Select  branch(es) that could have been audited}

First, the leader selects the branch(es) with $highAuditQC$ in the highest view. Any audited branch will be part of that set. Auditing a branch requires two auditQCs in the same view. If a branch was audited, the view change will thus necessarily see at least one audit QC for it (two sets of $N - u$ nodes intersect in at least one honest node). The view stabilization logic further guarantees that any committed branch will belong to this set: committed branches always extend the latest audit QC.

\protocol{\textbf{Rule 2.} Select branch(es) that could have gone fast path}

Next, the leader must ensure that it extends any branch that was audited on the fast audit path.  As the fast path requires $N$ votes to proceed, any batch that appears $k=N - (u + f_{safe})$ times could have gone fast path ($u$ nodes may not reply, and $f_{safe}$ may lie). The leader thus selects any branch $B$ that appears $k$ times.

\protocol{\textbf{Rule 3.} Select branch with highest view}
\protocol{\textbf{Rule 4.} Select branch with highest sequence number}
~
The first two steps ensure that audited branches are preserved. The last two steps ensures that committed branches are preserved \textit{in the absence of malicious replicas}. The view stabilization process guarantees that, in the absence of malicious nodes, the branch with 1) the highest view 2) the highest sequence number in the view change will always correspond to the latest committed branch. A commit QC (\majority) intersects the view change QC $(N-u)$ in $X \geq \lfloor \frac{f_{safe} + 1}{2} \rfloor + 1$ nodes, where $X$ may be smaller than $f_{safe}$. If the leader sees $X$ votes for committed operations in the view change, it must assume $B$ could have committed and select that branch. Unfortunately, if $f_{safe} \geq X$, there may be equally many conflicting votes for a branch $B'$ that conflicts with $B$, as malicious replicas can simply spoof them. This does not violate safety as committed branches are only guaranteed to persist in the absence of any malicious node compromises. A leader can select any of these conflicting branches.

\subsection{Receipts}\label{sec:receipts}
In addition to providing execution results via the usual $f_{safe}+1$ matching responses, \sys{}  provides cryptographic evidence of committing (and auditing) of transactions to clients through receipts~\cite{ccf, ethereum}.
Receipts act as \textit{transferable} evidence: clients can use these to convince other clients that their transaction was committed or audited.
As soon as a client's transaction is committed (resp. audited), the leader responds back to the client with a \textit{signed} commit (resp. audit) receipt. Commit receipts simply include a Merkle inclusion proof of the transaction in a batch.
Audit receipts additionally contain proof that the batch in which the transaction was included is audited.
This has two components:
(1) one or more \qcs (depending on whether fast or slow path audits have been used) certifying that a batch extending from the transaction's batch has been audited, and
(2) a short segment of the hash-chain connecting the transaction's batch to the audited batch.
Receipts have bounded size: Merkle inclusion proofs scale logarithmically with the (bounded) batch size  and the size of the hash-chain is at most the batch signing interval.

\subsection{Implementation Notes}
\label{subsec:impl-notes}
The protocol in the previous section included inefficiencies for simplicity of exposition. We briefly summarize them here.

\par \textbf{Digest} \sys{} sends the full batch only once as part of the first \textsc{Append-Entry} message. In other cases, it includes a digest of the message only. %

\par \textbf{Materializing CommitQC} We described each batch has having both a commitQC and an auditQC. In reality, it is unnecessary to send the $commitQC_{anc}$ over the network as
(a) it can't be used to trace the origin of votes since the votes may or may not have signatures,
(b) for commits, the leader is assumed to be non-malicious.
Thus, sending only the commit index, i.e., $commitQC_{anc}.B.seq$ with each batch suffices. %

\section{Evaluation}
\label{sec:eval}
In this evaluation, we answer four questions:
\begin{enumerate}[leftmargin=*]
    \item  What is the cost of background auditing in \sys{}?
    \item How does \sys{} perform in a multi-platform setup?
    \item How quickly does \sys{} detect and reconcile logs in the presence of platform compromises?
    \item How can applications make use of asynchronous auditing?
\end{enumerate}

\par \textbf{Implementation}
We implement \sys{} in $\sim$8k lines of Rust.
We use Protobuf for serialization, and Tokio for event-driven programming and networking. %
We use the Ed25519 signatures (dalek-cryptography library) and SHA512 hashes.

\begin{figure*}[!t]
\begin{subfigure}[b]{0.60\linewidth}
    \centering
    \includegraphics[width=0.9\linewidth]{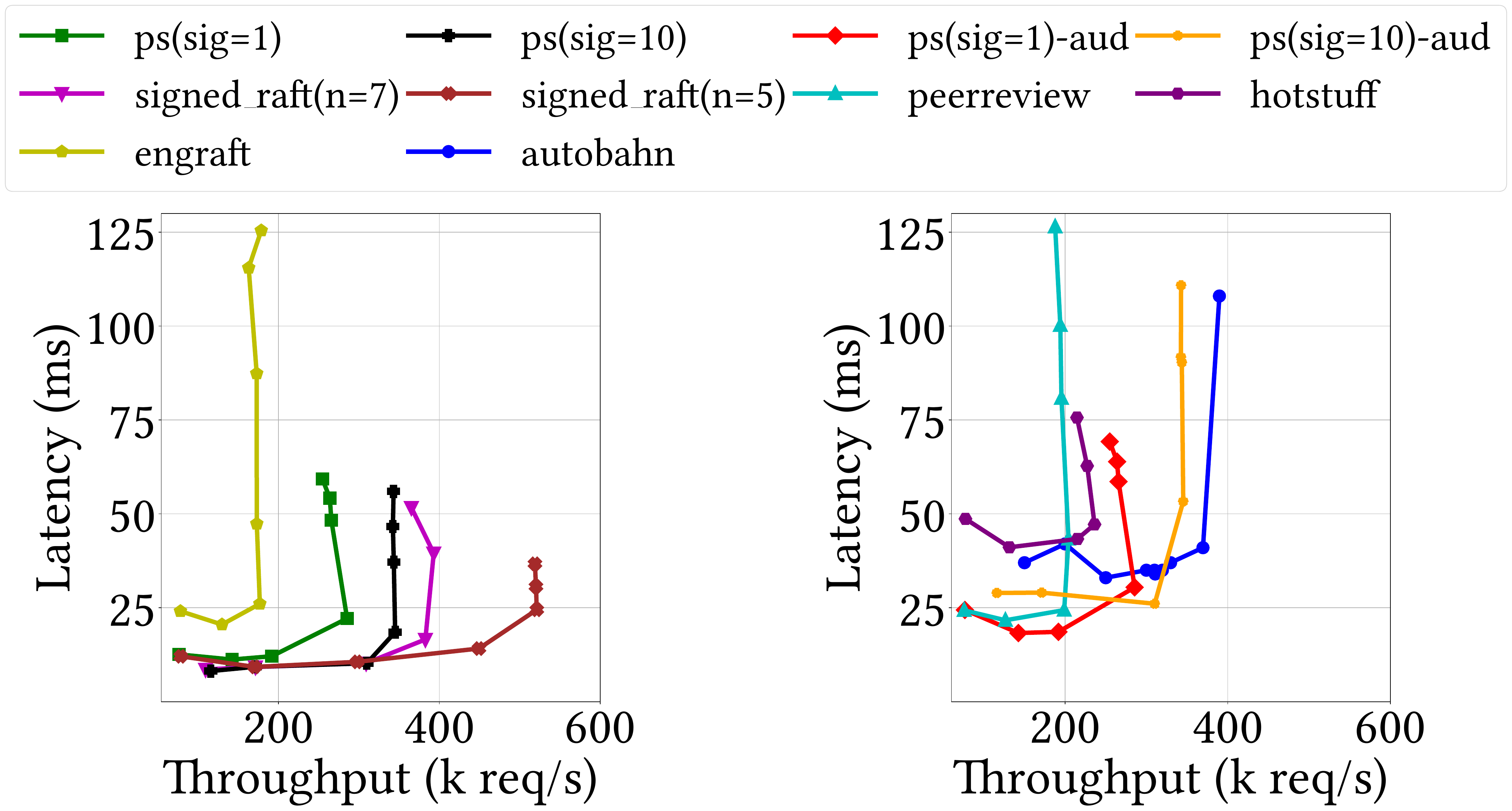}
    \caption{Transaction commit (left) and audit (right) performance }
    \label{fig:lan-experiment}
\end{subfigure}
\begin{subfigure}[b]{0.39\linewidth}
    \centering
    \includegraphics[width=0.9\linewidth]{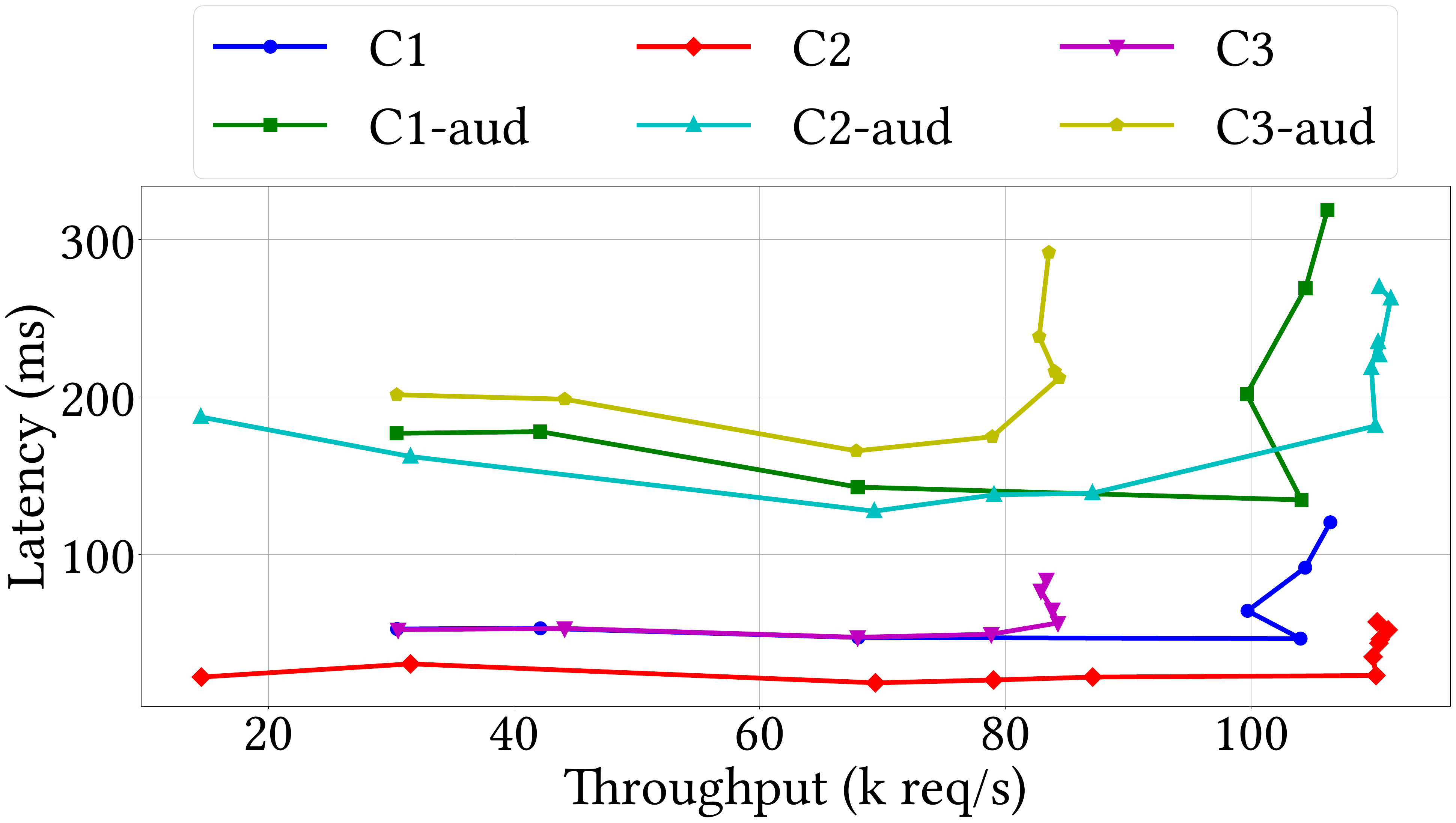}
    \caption{Multi-platform performance}
    \label{fig:wan-experiment}
\end{subfigure}
\caption{Common Case Performance}
\vspace{-0.3cm}
\end{figure*}

\par \textbf{Evaluation Setup} \label{sec:eval_setup}
We evaluate \sys{} on Azure using
16 core DC16ads\_v5 with AMD SEV-SNP, DC16eds\_v5 with Intel TDX, and D16ads\_v5 without any TEE capabilities (all with 64 GB of RAM, P40 data disks with 1000 MB/s burst throughput, 10 Gbps bandwidth for TEE nodes, 12.5 Gbps for non-TEE nodes). We use up to 5 D8ds\_v5 VMs as clients.
We report the average commit and audit latencies recorded by the clients. Unless otherwise stated, experiments use $N=7$ ($u=f_{safe}=2$) in a LAN setup (East US region) with AMD SEV-SNP machines (ping time 0.6ms). The batch size is 1000, the payload is 512B and the default signature interval is 10.
We implement the following baselines:
\begin{itemize}[leftmargin=*]
    \item {\sf signed\_raft}: Raft~\cite{raft} with hash-chaining and periodic signatures, inspired by the steady-state operation of \cite{ccf}.
    \item {\sf engraft}: Raft with an extra 2-phase-commit to store the latest metadata to prevent disk rollback attacks~\cite{Wang22}.
    \item {\sf hotstuff}: State-of-the art chained BFT protocol \cite{hotstuff}.   %
    \item {\sf autobahn}: State-of-the-art DAG BFT protocol \cite{autobahn} that optimizes for low latency and high throughput.
    \item {\sf peer\_review}: Accountability protocol~\cite{peerreview}, which uses signed witnesses to asynchronously detect Byzantine faults.
\end{itemize}
We implement {\sf signed\_raft}, {\sf engraft}, {\sf peer\_review} and {\sf hotstuff} on the \sys{} codebase. We use {\sf autobahn}'s open source code%
\footnote{We validated our {\sf autobahn} and {\sf peer\_review} results with the authors}.
BFT protocols use $3f+1$ nodes, the remaining TEE-augmented protocols use $2f+1$ nodes.

\subsection{Common Case Performance}
\par \textbf{Commit Performance}
Consider first the behavior of committing transactions in \sys{} relative to \sr and \engr. Both offer equivalent guarantees to \sys{}'s commit and ignore TEE compromises. We ask: can \sys{} match their performance despite its additional audit functionality? We summarize results in \cref{fig:lan-experiment}, where {\sf ps(sig=k)} denotes \sys{} with signing interval of $k$ batches.

\psten achieves a peak throughput of $345k$ txn/s and is bandwidth-bottlenecked at the leader.  \sr, running with 5 nodes instead of 7, achieves $520k$ txn/s (also signing every 10 batches). This 35\% degradation is caused by \psten's higher replication factor (necessary to enable auditing). \psten achieves 87\% of the throughput of \sr configured with seven nodes (\srseven). \psten must send auditQCs, which consume additional bandwidth.  %
\psten instead has $\approx 1.9x$ higher throughput than \engr. \engr is a TEE-based Raft modification that tolerates disk rollbacks. In \engr, any write to the local disk is replaced by a 2PC phase. Both the leader and the followers perform a 2PC to store Raft metadata and hash of the last batch before they vote. These two additional round-trips significantly hurt performance.

Next, we investigate how \sys{}'s audit performance compares against state-of-the-art BFT protocols, who make no assumptions about TEEs. We consider the same execution of \sys{}, but this time look at audited transactions. Our audit performance should, ideally, remain competitive with such systems.
\psten achieves 88.6\% of the throughput of \abn but 15\% lower latency. \abn allows parallel proposers to submit hash chains of batches, with a leader performing consensus on
\textit{cuts} of these chains. In our experiments, \abn is bottlenecked on hashing.  While \sys{} does not allow parallel proposers, its pipelining logic allows for the leader to have multiple outstanding batches at once. Unlike \abn, \sys{} carefully incrementalizes hashing such that the majority of the hashing work can be done in parallel. %
When \sys{} signs every batch, its throughput falls to 73\% of \abn. %

Both \psten and \psone achieve higher throughput than \hs (46\% and 20\% respectively).
\hs uses one extra phase to commit and, by design, must sign all batches. \hs's blocking pipelining logic can only have a single outstanding batch in the network. This well-known shortcoming~\cite{autobahn} fundamentally limits its maximum throughput.
\psone has 39\% higher peak throughput than \prv.
\prv requires all messages of the underlying BFT protocol to be signed. Its detection relies on an $O(n^2)$ communication phase to transmit signed witnesses of these messages to respective witness verifier nodes. Transmission of extra messages along with the extra load of auditing witnesses interferes with the common case performance, leading to the throughput drop.

In summary, \sys{} is able to maintain the performance of current TEE-enabled systems while providing additional audit guarantees. \sys{}'s audit performance is only minimally lower than state-of-the-art BFT protocols.
If one places no trust on TEEs, \sys{} can be used interchangeably with state-of-the-art BFT protocols. But, placing cautious trust in TEEs yields up to 2.6x lower commit latencies.

\par \textbf{Platforms} Next, we investigate
how \sys{} scales with the number and type of platforms. We deploy \sys{} on up to five platforms with the following (TEE, region) pairs:
(a) (Intel TDX, East US 2),
(b) (AMD SEV-SNP, East US),
(c) (AMD SEV-SNP, West US),
(d) (Intel TDX, Central US),
(e) (AMD SEV-SNP, West Europe).
Here, we consider that different regions with the same TEE technology are sufficiently unlikely to experience correlated failures to form different platforms (many TEE attacks require physical access).

We fix $\pi_{safe} = 1, \pi_{live} = 0$, optimizing for safety, not liveness~\cite{angel2023nimble,svr3}
and create three different configurations:
\begin{enumerate}[label={$\mathbf{C_{\arabic*}:}$}]
    \item $u = f_{safe} = 1$ with a $(1, 1, 1, 1)$ replica distribution.
    \item $u = f_{safe} = 2$ with a $(2, 2, 2, 1)$ replica distribution.
    \item $u = 3, ~f_{safe} = 2$ with a $(2, 2, 2, 2, 1)$ replica distribution.
\end{enumerate}
In all cases, the leader and the clients are located in the East US2 region.
Fig. \ref{fig:wan-experiment} reports throughput and latency numbers for \sys{}.
\sys{} in configuration $\mathbf{C_2}$ has 51\% and 59\% lower commit latency than respectively $\mathbf{C_1}$ and $\mathbf{C_3}$. The leader needs votes from only two regions to commit in $\mathbf{C_2}$ but needs votes from three regions in $\mathbf{C_1}$ and $\mathbf{C_3}$.
For auditing, votes are required from at least three platforms in all configurations. $\mathbf{C_3}$ only audits on the slow path as the latency of platform (e) is more than twice the latency of all the other platforms. As a consequence, it has the highest latency.

The difference in throughput between configurations stems from bandwidth heterogeneity, from approximately 1.17 Gbps to 9.5 Gbps. \sys{} chooses to rate limit its sending such that sending queues at the leader do not become excessively large. In effect, this ensures that each node receives batches only as fast as the slowest node does. %
$\mathbf{C_1}$ and $\mathbf{C_2}$ use the same set of platforms and thus the same network configuration. They have identical throughput.
$\mathbf{C_3}$ instead includes a more constrained link to West Europe with only 1.17 Gbps of available bandwidth (the slowest link in $\mathbf{C_1}$, is instead 1.5 Gbps). Peak throughput thus drops by 24\% compared to $\mathbf{C_1}$.

\begin{figure}[hbtp]
\vspace{-0.1cm}
\centering
\begin{subfigure}[b]{0.4\linewidth}
    \includegraphics[width=1.0\linewidth]{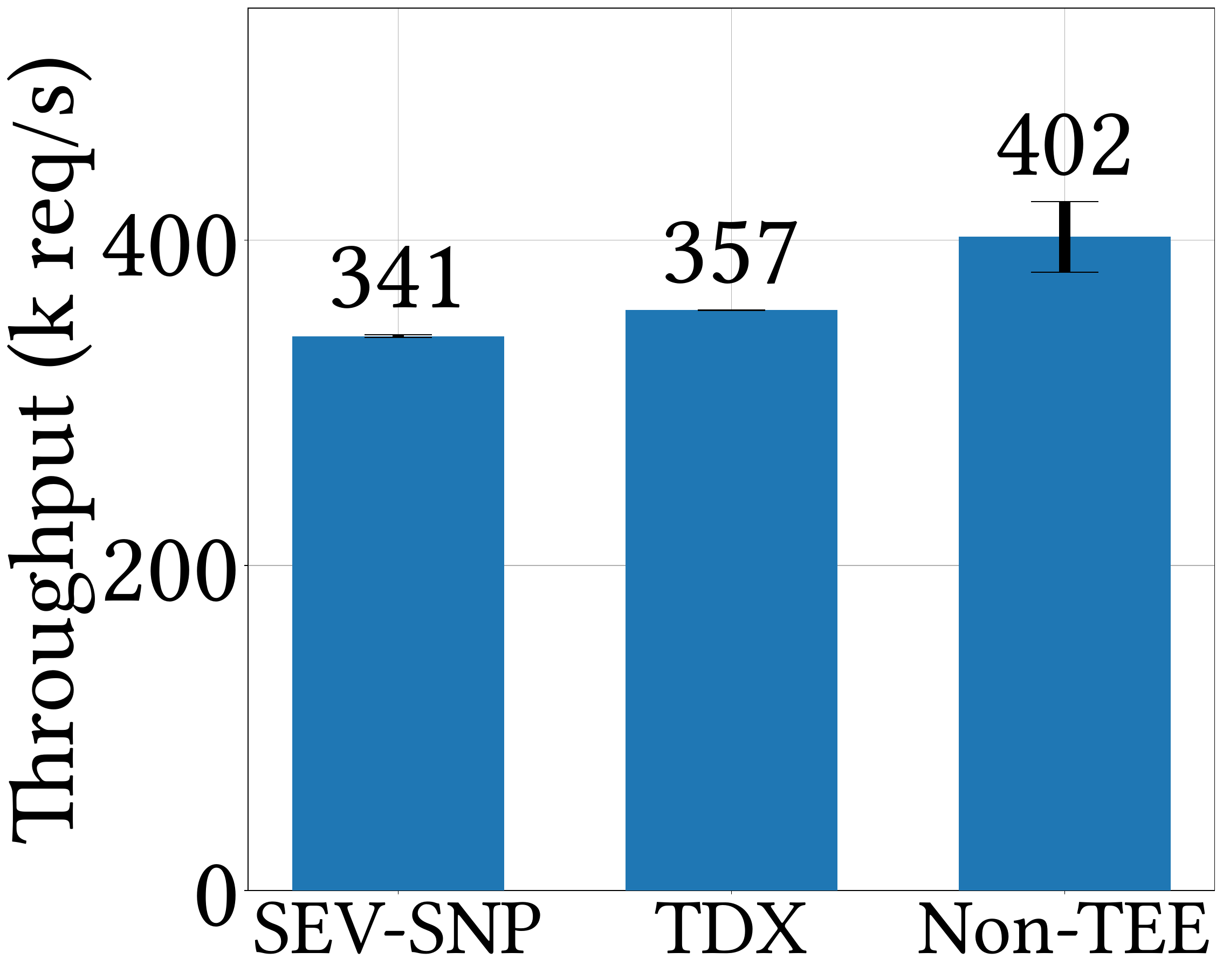}
    \centering
    \caption{TEE Type}%
    \label{fig:storage-ablation}
\end{subfigure}
~~~~~~
\begin{subfigure}[b]{0.4\linewidth}
    \centering
    \includegraphics[width=1.0\linewidth]{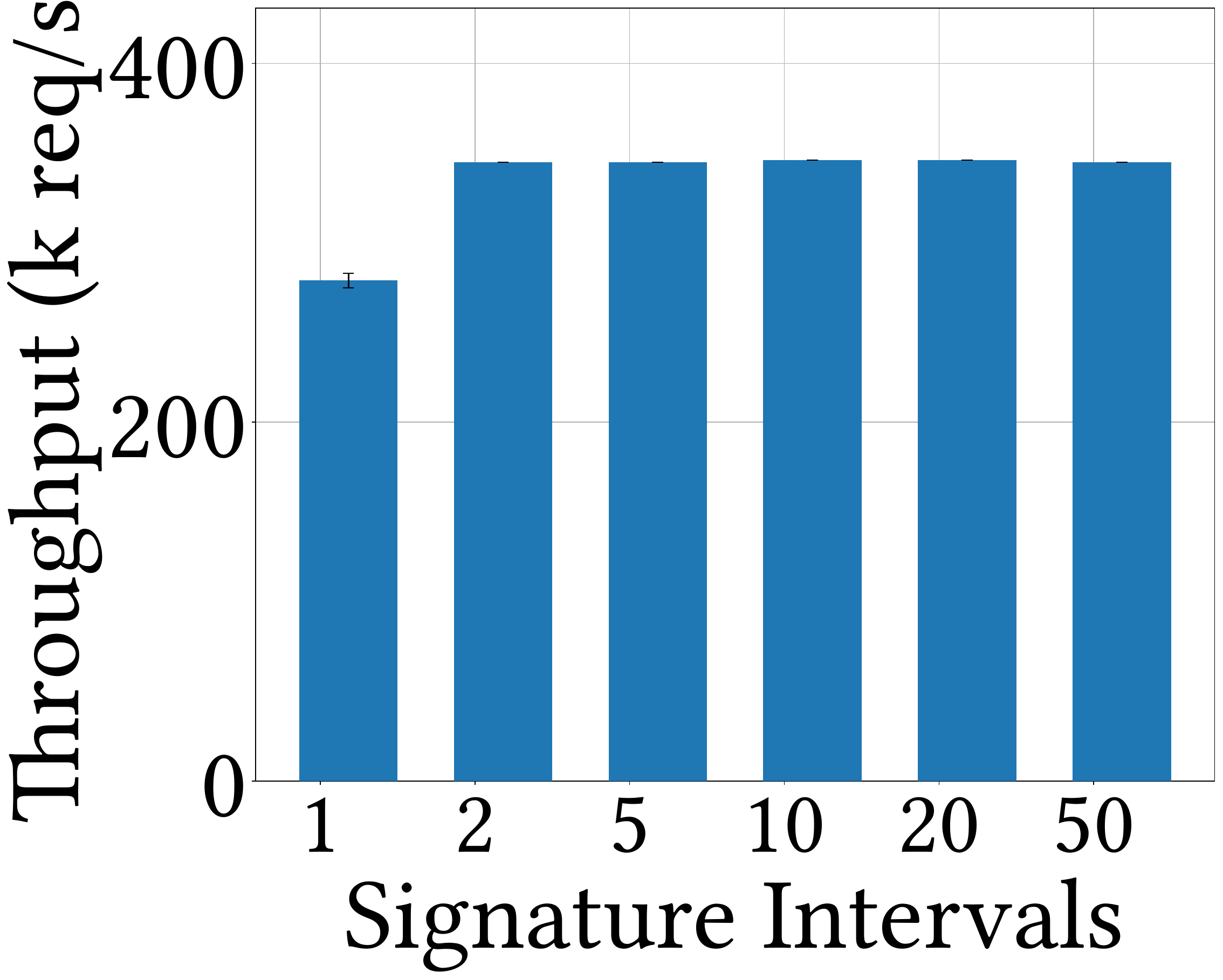}
    \caption{Signing Interval}
    \label{fig:crypto-thread-sensisitivity}
\end{subfigure}
\caption{Microbenchmarks}
\label{fig:microbenchmarks}
\end{figure}

\par \textbf{Overheads of TEEs}
Fig. \ref{fig:storage-ablation} shows
the effect of running \sys{} in different trusted execution environments and without a TEE. We find that the overhead of running inside of Intel TDX and AMD SEV-SNP is similar, both cause an 11--15\% reduction in throughput.

\par \textbf{Periodic Signing}
Fig. \ref{fig:crypto-thread-sensisitivity} quantifies sensitivity to periodic signing. As expected, increasing the signing interval from every batch to ten batches improves throughput by a factor of 1.24  as it amortizes signature cost. This comes with an increase in auditing latency by a factor of 1.84x. Increasing the signing frequency beyond that has no further throughput benefit as the system becomes network-bound.

\subsection{Failure Case Performance}
Next, we consider situations in which a platform has been compromised. We categorize the space of possible Byzantine behaviors into two: (i) failure to respond to a vote or view-change request (\textit{omission attacks}), and (ii) more subtle attacks involving equivocation (\textit{commission attacks})~\cite{clement09upright}.

\begin{figure*}[hbtp]
\begin{subfigure}[b]{\columnwidth}
    \centering
    \includegraphics[width=0.7\linewidth]{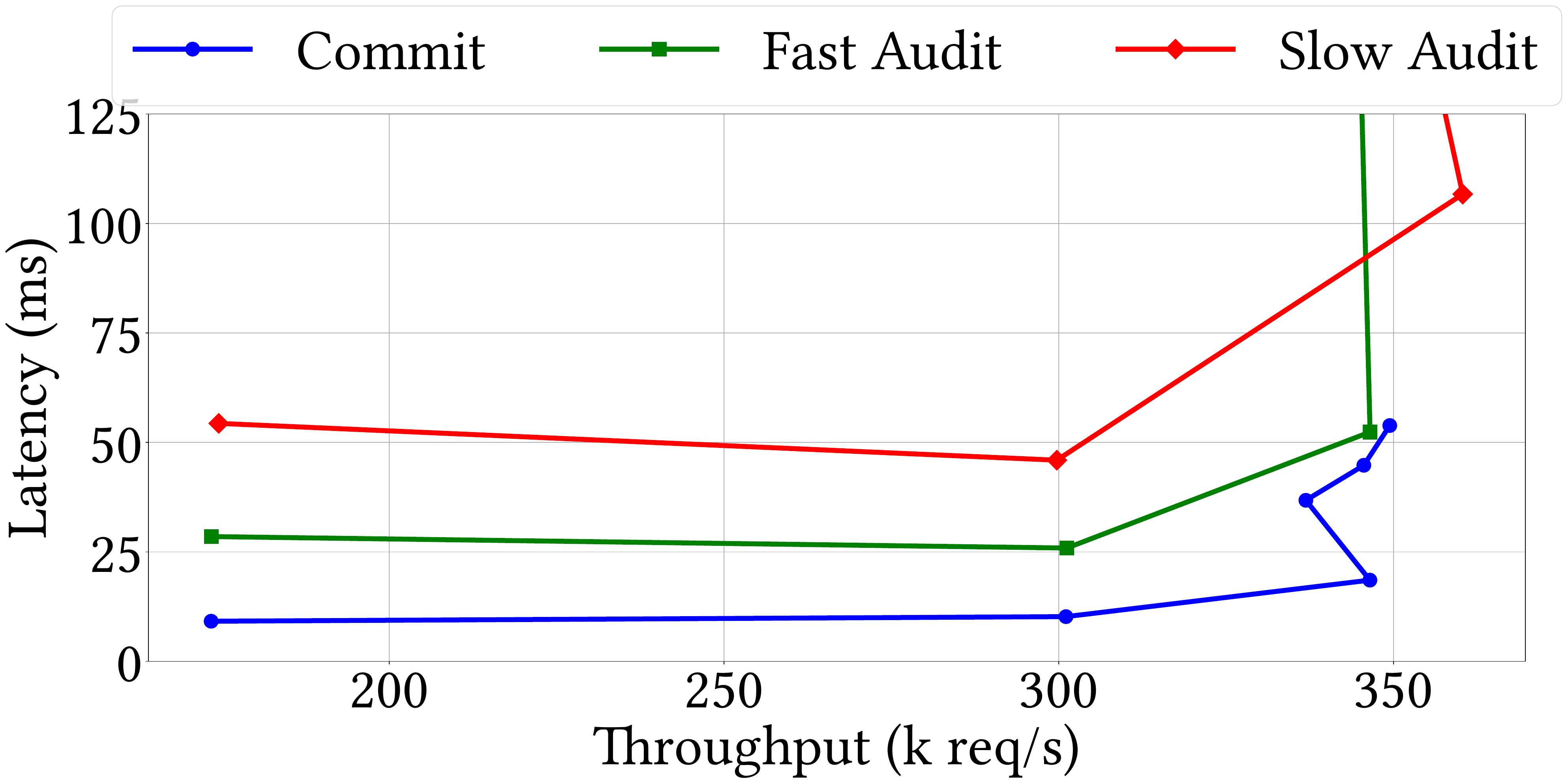}
    \caption{Slow path  vs fast path audit}
    \vspace{-0.2cm}
    \label{fig:fast-path}
\end{subfigure}
\begin{subfigure}[b]{\columnwidth}
    \centering
    \includegraphics[width=0.7\linewidth]{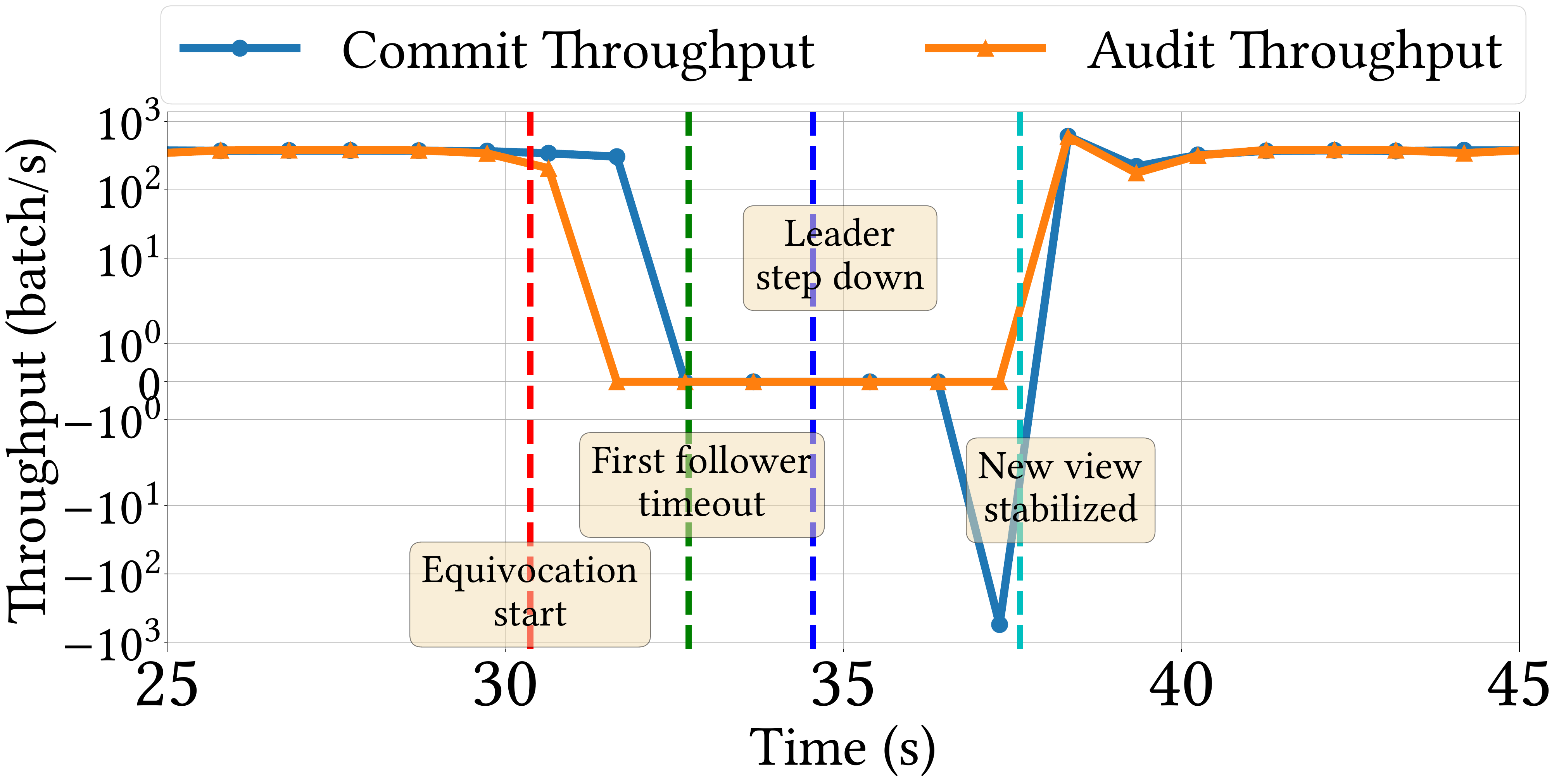}
    \caption{Equivocation}
    \vspace{-0.2cm}
    \label{fig:view-change}
\end{subfigure}
\caption{Effects of failures}
\label{fig:failures}
\vspace{-0.3cm}
\end{figure*}

\par \textbf{Omission Attacks} Quorums in \sys{} are able to tolerate up to $u$ nodes that do not respond. Malicious nodes will thus not prevent quorums from forming, they will at best increase tail latency. They can, however, prevent fast audit quorums from forming, as fast audit quorums require responses from all nodes. We quantify the effect of malicious nodes effectively deactivating audits via fast path in \cref{fig:fast-path}. As expected, auditing transactions on the slow path has 1.9x worse latency than doing so via the fast path.
Fast path audit requires waiting for a single QC, while the slow path audit waits to form a QC over a QC, thus taking twice as long.  Fast path audits occasionally have to wait for straggler nodes, explaining why latency reduces by slightly less than a factor of two. There is no negative effect on throughput as pipelining ensures that no additional messages are being sent, regardless of the path taken in the protocol.

\par \textbf{Commission Attacks}  Next, we consider situations in which a leader in a compromised platform equivocates. We quantify how quickly \sys{}'s audit process can detect malicious leaders that broadcast two conflicting batches to two different partitions in the cluster. We plot the observed commit and audit throughput for a replica in the system in Fig. \ref{fig:view-change}.
In the first phase after equivocation (30.3 -- 32.7s), transactions commit as normal. Commit quorums do not necessarily intersect in an honest node. It is thus possible for two conflicting batches to commit, and for new batches to extend conflicting branches. This is not the case for the larger audit quorums and the auditing throughput thus quickly drops to zero. Followers time out in 4s waiting for an audit  and the system transitions to a new view with a new leader (35s).
The next leader detects equivocation, selects its own branch, and applies it to all replicas. Replicas that had extended conflicting branches (like the one plotted in this graph), thus have to rollback their logs (hence the negative throughput) while the new leader stabilizes the view (37.61s).
The new leader then quickly makes the rolled back replica audit the new branch (the positive spike in audit throughput). The system returns to steady state around 40s.

\subsection{Application Use Cases}
\label{subsec:apps}

\begin{figure*}[hbtp]
\captionsetup[sub]{font=scriptsize,labelfont={bf,sf}}
\begin{subfigure}{0.19\linewidth}
    \centering
    \includegraphics[width=1.0\linewidth]{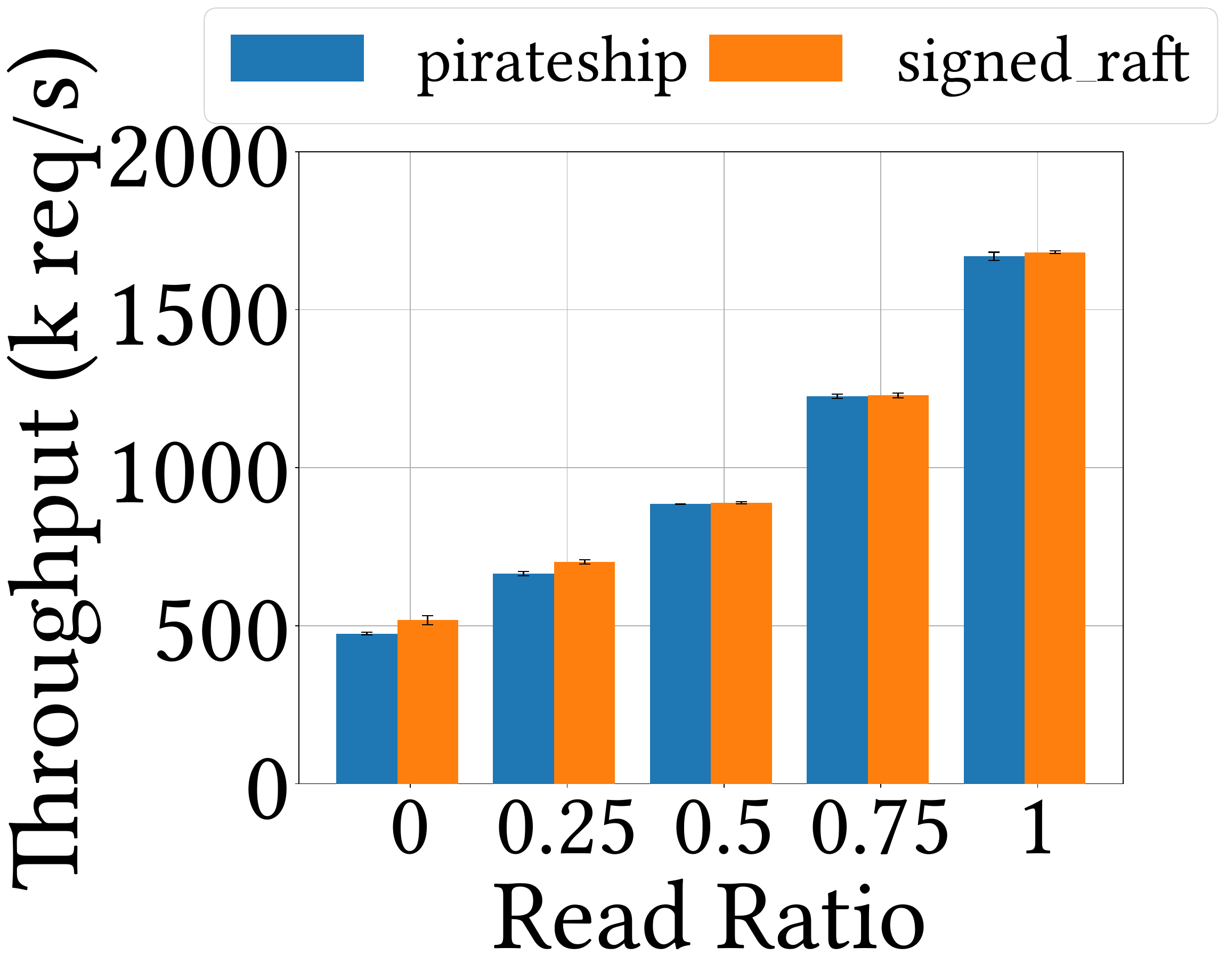}
    \caption{Key-Value store}
    \label{fig:ycsb}
\end{subfigure}
~
\begin{subfigure}{0.19\linewidth}
    \centering
    \includegraphics[width=1.0\linewidth]{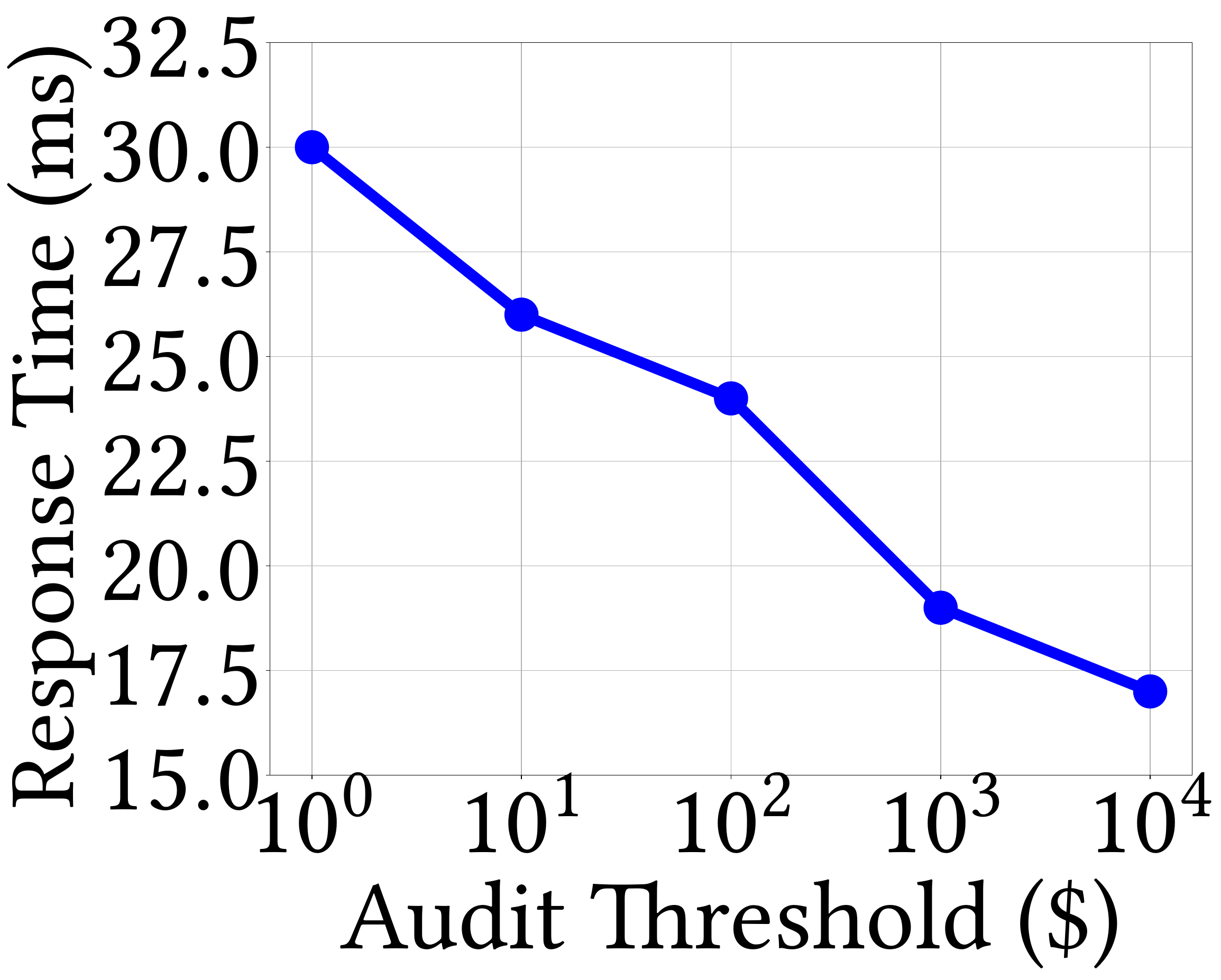}
    \caption{Banking}
    \label{fig:banking}
\end{subfigure}
~
\begin{subfigure}{0.19\linewidth}
    \centering
    \includegraphics[width=1.0\linewidth]{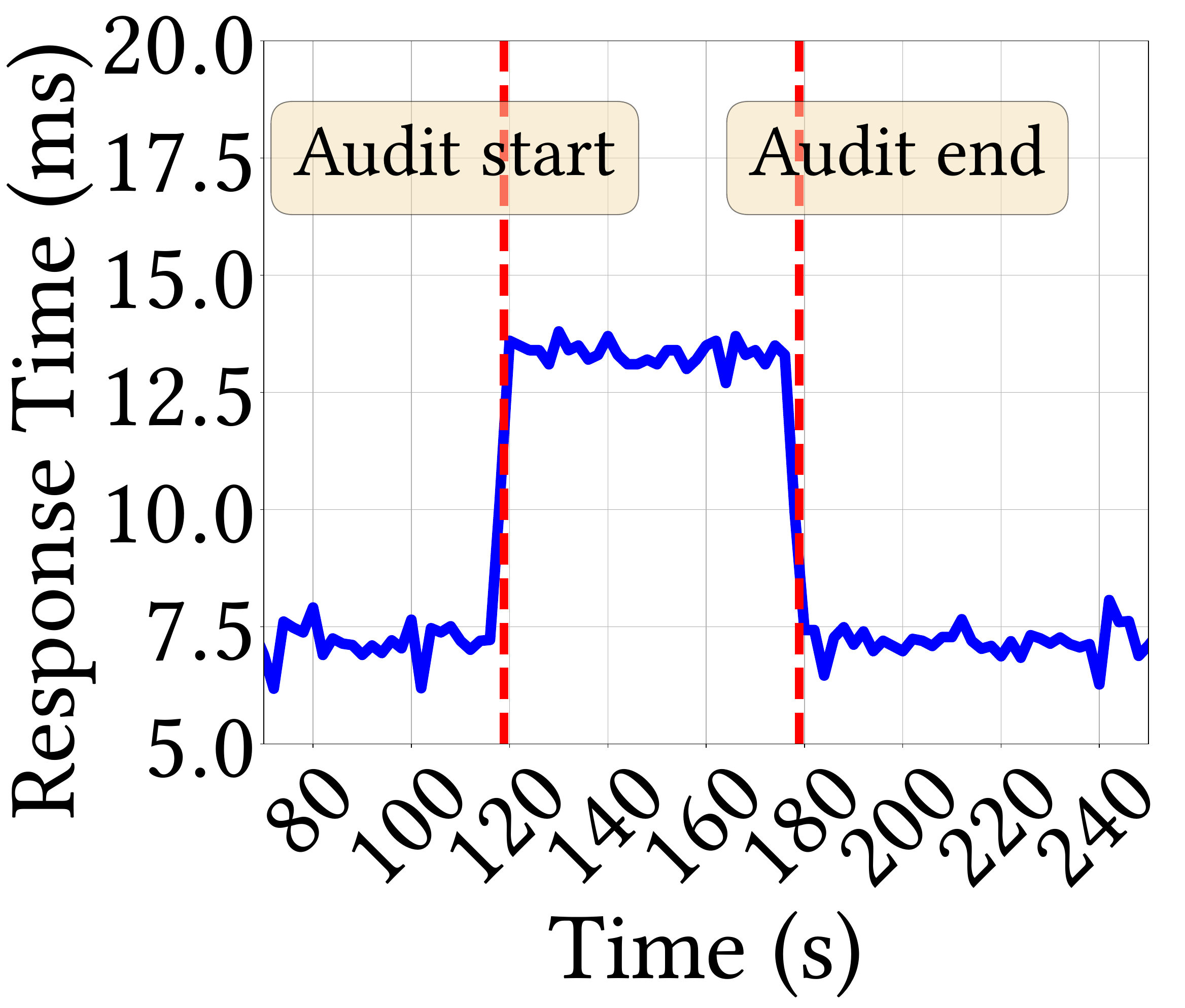}
    \caption{Key Transparency Service}
    \label{fig:kms}
\end{subfigure}
~
\begin{subfigure}{0.19\linewidth}
    \centering
    \includegraphics[width=1.0\linewidth]{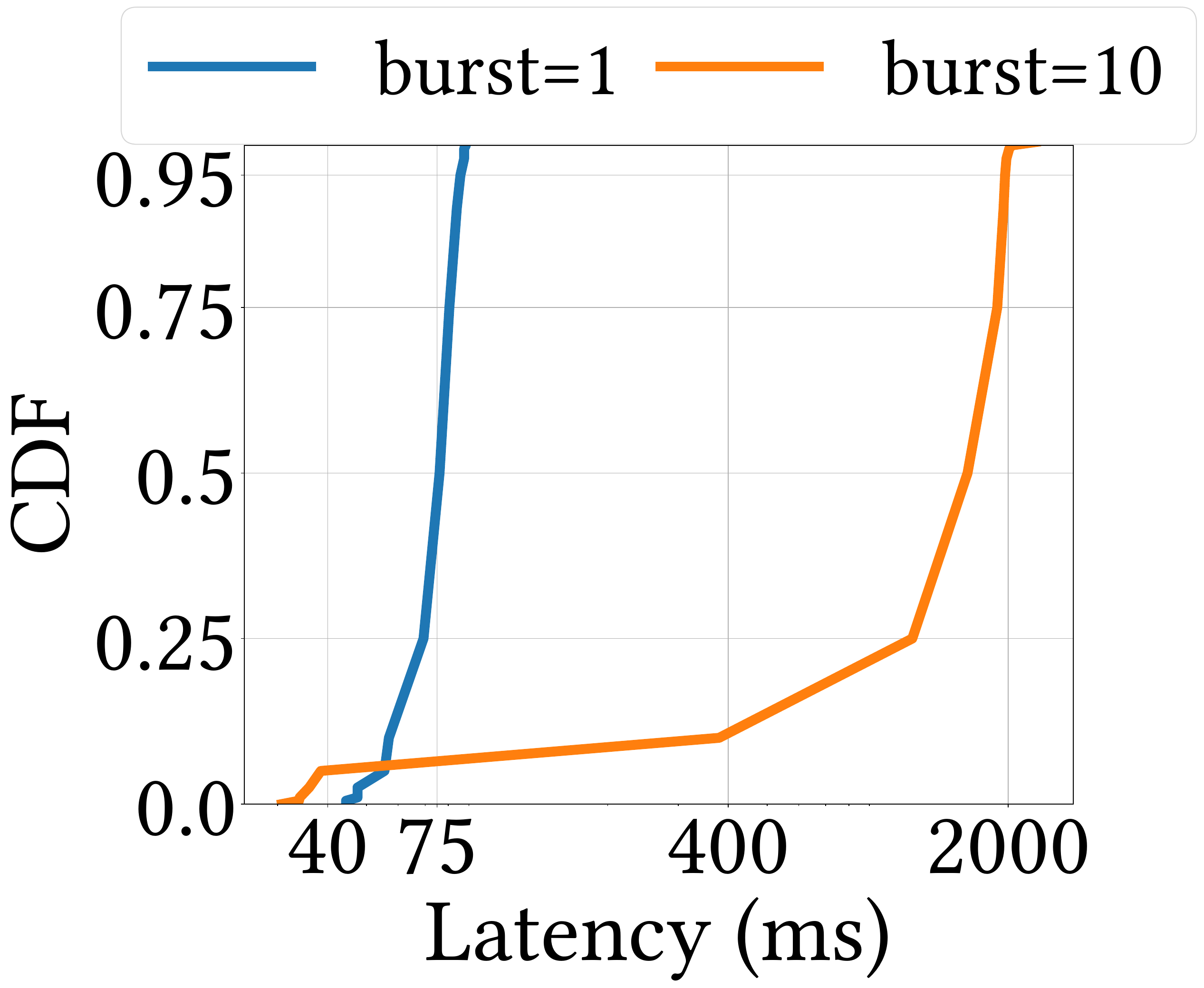}
    \caption{Secret Key Recovery}
    \label{fig:svr3}
\end{subfigure}
~
\begin{subfigure}{0.19\linewidth}
    \centering
    \includegraphics[width=1.0\linewidth]{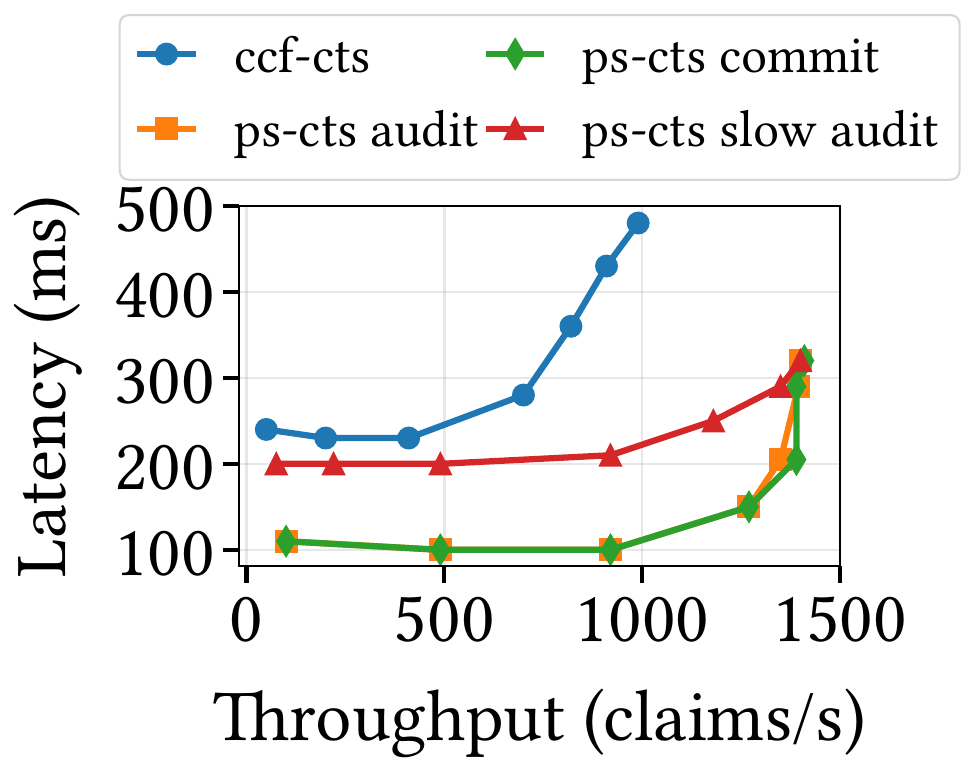}
    \caption{Code Transparency Service}
    \label{fig:cts-experiment}
\end{subfigure}

\caption{Application Performance }
\label{fig:applications}
\vspace{-0.4cm}
\end{figure*}

We next quantify how applications currently using TEE-enabled ledgers can make use of \sys{}'s additional audit guarantees. Applications can choose to wait for audit based on either \one type of transaction, \two suspicion of an attack, \three suspicion of an adversarial client; or \four it can expose the details of commit and audit to clients using receipts.

\par \textbf{Key-Value Store.}
First, we show how to handle rollbacks in a key-value store.
We build the kv-store in two layers.
The first layer is versioned and stores, for each key, write versions along with the batch sequence number.
On a rollback, all operations with version numbers higher than the new rolled back commit index are deleted. Once a batch is audited, the batch will no longer be rolled back. It can be safely moved to the non-versioned second layer.
We measure the performance of this KV store using YCSB \cite{ycsb} ($300k$ objects, 100B each) (Fig~\ref {fig:ycsb}).
\sys{} and \srfive perform similarly as they are both bottlenecked on execution.

\par \textbf{Banking.}
We use the Key-Value store to develop a simple banking application. %
Large transfers above a threshold wait to be audited, whereas small transfers proceed on commit. Transfer values range \$1 and \$10K and follow a  Zipfian distribution.
Fig. \ref{fig:banking} shows latency according to the set threshold.
As expected, average response time grows as an increased fraction of transfers must be audited. Latency for audited transfers is 1.7$\times$ higher than for committed transfers.

\par \textbf{Key Transparency Service.}
A key transparency service stores public keys for users~\cite{optiks, CONIKS}.
We test it using a workload of 50\% queries to a user's public key (uniformly distributed across all users) and 50\% key updates.
We consider (\cref{fig:kms}) a scenario in which a new vulnerability in a TEE is discovered, and the application switches to returning to the client on confirmation of audit rather than commit.
When waiting for audits, the response time immediately jumps by $\approx 1.7x$ and returns to normal once the application resumes to replying on commit (180s).

\par \textbf{Secret Key Recovery.}
We implement Signal’s Secret Key Recovery service~\cite{svr3} by designing an $(N\!-\!u)$-out-of-$N$ password-protected secret-sharing scheme.
A user recovers their key by sending a password guess to $N\!-\!u$ servers to obtain sufficient shares. Each user is allowed a maximum number of guesses, which we  maintain through a per-user global guess counter. We leverage the bound that \sys{} enforces between audit and commit indices: \sys{} forces clients to wait for audit once they have too many pending (committed) password guesses such that they may exceed password limits during a rollback. \cref{fig:svr3} considers two clients: one sending one password guess per second, the other performing a burst of ten guesses at once.
If the application discovers two consecutive committed-but-not-audited requests from the same user, the server throttles the user until all of these requests can be audited.
The client performing one request per second never gets throttled, whereas the bursty client gets throttled often, and as a result has 20x higher median latency.

\par \textbf{Code Transparency Service.}
A Code Transparency Service~\cite{Delignat23, scitt, scitt-ietf} (CTS) provides a root of trust for software supply chains by maintaining a ledger of signed claims describing releases, their provenance, and the policies they satisfy.
Clients and auditors can verify that a release was issued by an authorized publisher, and each accepted claim is recorded together with a receipt that can be checked offline.
Our CTS implementation, \syscts, as is standard, only admits \textit{valid} claims in the ledger.
To validate claims \textit{before} they could be committed, \sys{CTS} adds a validation step across all replicas, that satisfies the following invariants: \one a correct leader validates incoming claims before adding them to a batch, \two correct replicas independently validate every claim in a proposed batch before voting to accept it.
If correct replicas detect policy-violating claims, they simply withhold their vote, preventing the batch from gathering a quorum.
Clients can wait for a commit receipt, which provides the same guarantees as the receipts granted by current TEE-based ledgers, e.g., CCF~\cite{ccf}.
With \sys{}, clients can additionally wait for audit receipts.
More details on implementation can be found in \Cref{subsec:details-cts}.

We evaluated \syscts against a CCF-based CTS~\cite{ccf, scitt} under a workload consisting of continuous claim submissions and receipt retrievals.
Because CCF only supports confidential container instances (C-ACI), we deployed the service on 7 C-ACI containers and used 3 regular containers as clients (16 CPU cores, 64 GB of RAM)
(Fig. \ref{fig:cts-experiment}).
\syscts' stronger guarantees do not come at the cost of performance.
In fact, audit receipts on the fast path achieve latency comparable to commit receipts, as expected, since they require only a single unanimous QC.

\section{Related Work}

Consensus protocols using TEEs for better resilience either \one minimize their TCB, or \two put their entire consensus logic in the TCB. The former places a small amount of logic in the TEE to detect malicious behavior.
A2M-PBFT~\cite{chun2007a2m} stores message digests in trusted append-only logs.
MinBFT~\cite{minbft} and Hybster~\cite{hybster} use trusted counters~\cite{trinc} to detect equivocation.
These protocols can then safely run with $n=2f+1$ or reduce the number of phases~\cite{gupta2023bft}.
While others~\cite{brandenburger2017rollback,angel2023nimble} focus on rollback detection with untrusted storage, by using special client protocols or TEE-endorsed replication.

With the advent of VM-based TEEs~\cite{SEV-SNP,TDX,CCA} that place the whole VM inside their TCBs, the second style of consensus protocols is gaining prominence and \sys{} falls into this category.
This line of work augments CFT protocols with mechanisms to prevent or detect specific Byzantine behaviors.
CCF~\cite{ccf, ccf-ledger} modifies Raft~\cite{raft} to include periodic signature transactions that prove the integrity of all previous entries in the log.
Engraft~\cite{Wang22} adds rollback resistance to the Raft log using a separate 2-phase replication protocol.
Signal's Key Recovery (SVR3)~\cite{svr3} guards against rollbacks of physical memory using extra replicas and one extra phase.
IA-CCF~\cite{shamis2022iaaccf} runs a standard BFT inside enclaves and adds individual accountability for misbehaving nodes using an auditor, even with more than $f$ Byzantine faults.
SplitBFT~\cite{messadi2022splitbft} runs a compartmentalized version of PBFT where each compartment is in a separate enclave.
While conducive to very high performance, it does not capture the correlated nature of TEE failures.

Prior work seeks to consider more fine-grained failure types for better performance~\cite{clement09upright, porto2015visigoth,xft}, distinguishing between omission and commission faults as well as slow and crashed nodes. Other research efforts decouple safety from liveness~\cite{malek2005fault,clement09upright}, while some projects argue that Byzantine replicas are never incentivized to break safety, only liveness~\cite{flexible-bft}.  \sys{} similarly distinguishes between Byzantine and crashed nodes. Correlated TEE failures were previously informally described in \cite{glacier, totalrecall, signal-recovery}.   \sys{} formalizes the concept using platform fault tolerance.

While \sys{} detects and recovers from log divergence caused by TEE compromises, it does not try to blame faulty nodes.  Prior works that characterize detectable faults~\cite{haeberlen09opodis} or assign blame, like
PeerReview~\cite{peerreview} or BFT Protocol Forensics~\cite{bft-forensics}, are orthogonal but complementary to \sys{}.

Finally, \sys{} draws from many existing concepts. Hash-chaining, commit rules and linearizing communication have been used in blockchains~\cite{buchman2016tendermint}, SBFT~\cite{gueta2019sbft}, Hotstuff~\cite{hotstuff}, Jolteon~\cite{jolteon-ditto} and Beegees~\cite{giridharan2023beegees}.
\sys{}'s fast path is inspired by Zyzzyva~\cite{zyzzyva}, SBFT~\cite{gueta2019sbft} and Autobahn~\cite{autobahn}.
Recently, DAG-based BFT protocols~\cite{dag-rider, danezis2022narwhal, Bullshark, autobahn, babel2023mysticeti} have achieved very high throughput by decoupling data dissemination from consensus. \sys{}'s embedding of BFT logic within a CFT protocol achieves a similar effect.

\section{Conclusion}
\label{sec:conclusion}

This paper presents \sys{}, a PFT append-only ledger that uses TEEs to efficiently process requests when all is well. When TEE platform compromises arise, \sys{} efficiently detects and reconciles
conflicts by embedding a BFT audit protocol inside a CFT commit protocol.

\begin{acks}
We thank our shepherd and the reviewers for their helpful feedback in shaping the paper; Eddy Ashton and Amaury Chamayou for their help during earlier versions of the project, Markus Kuppe for his support with its TLA+ modeling and C\'{e}dric Fournet for his feedback on receipts.
This work was supported in part by Azure Research, Microsoft, NSF CAREER 2442542, a Sui Research Award, and gifts from Accenture, Algorithmic SuperIntelligence Labs, Amazon, AMD, Anyscale, Broadcom, cmpnd, Google, IBM, Intel, Intesa Sanpaolo, Lambda, Lightspeed, Mirendil, NVIDIA, Samsung SDS, and VESSL.
\end{acks}

\bibliographystyle{ACM-Reference-Format}
\bibliography{references,reference2,refs3}


\begin{thebibliography}{77}


\ifx \showCODEN    \undefined \def \showCODEN     #1{\unskip}     \fi
\ifx \showDOI      \undefined \def \showDOI       #1{#1}\fi
\ifx \showISBNx    \undefined \def \showISBNx     #1{\unskip}     \fi
\ifx \showISBNxiii \undefined \def \showISBNxiii  #1{\unskip}     \fi
\ifx \showISSN     \undefined \def \showISSN      #1{\unskip}     \fi
\ifx \showLCCN     \undefined \def \showLCCN      #1{\unskip}     \fi
\ifx \shownote     \undefined \def \shownote      #1{#1}          \fi
\ifx \showarticletitle \undefined \def \showarticletitle #1{#1}   \fi
\ifx \showURL      \undefined \def \showURL       {\relax}        \fi
\providecommand\bibfield[2]{#2}
\providecommand\bibinfo[2]{#2}
\providecommand\natexlab[1]{#1}
\providecommand\showeprint[2][]{arXiv:#2}

\bibitem[Abd-El-Malek et~al\mbox{.}(2005)]%
        {malek2005fault}
\bibfield{author}{\bibinfo{person}{Michael Abd-El-Malek},
  \bibinfo{person}{Gregory~R. Ganger}, \bibinfo{person}{Garth~R. Goodson},
  \bibinfo{person}{Michael~K. Reiter}, {and} \bibinfo{person}{Jay~J. Wylie}.}
  \bibinfo{year}{2005}\natexlab{}.
\newblock \showarticletitle{Fault-scalable Byzantine fault-tolerant services}.
  In \bibinfo{booktitle}{\emph{Proceedings of the Twentieth ACM Symposium on
  Operating Systems Principles}} (Brighton, United Kingdom)
  \emph{(\bibinfo{series}{SOSP '05})}. \bibinfo{publisher}{Association for
  Computing Machinery}, \bibinfo{address}{New York, NY, USA},
  \bibinfo{pages}{59–74}.
\newblock
\showISBNx{1595930795}
\urldef\tempurl%
\url{https://doi.org/10.1145/1095810.1095817}
\showDOI{\tempurl}


\bibitem[AMD(2020)]%
        {SEV-SNP}
\bibfield{author}{\bibinfo{person}{AMD}.} \bibinfo{year}{2020}\natexlab{}.
\newblock \bibinfo{title}{{AMD} {SEV}-{SNP}: Strengthening {VM} isolation with
  integrity protection and more}.
\newblock
\newblock
\newblock
\shownote{\url{https://www.amd.com/system/files/TechDocs/SEV-SNP-strengthening-vm-isolation-with-integrity-protection-and-more.pdf}}.


\bibitem[Angel et~al\mbox{.}(2023)]%
        {angel2023nimble}
\bibfield{author}{\bibinfo{person}{Sebastian Angel}, \bibinfo{person}{Aditya
  Basu}, \bibinfo{person}{Weidong Cui}, \bibinfo{person}{Trent Jaeger},
  \bibinfo{person}{Stella Lau}, \bibinfo{person}{Srinath Setty}, {and}
  \bibinfo{person}{Sudheesh Singanamalla}.} \bibinfo{year}{2023}\natexlab{}.
\newblock \showarticletitle{Nimble: Rollback Protection for Confidential Cloud
  Services}. In \bibinfo{booktitle}{\emph{17th USENIX Symposium on Operating
  Systems Design and Implementation (OSDI 23)}}. \bibinfo{publisher}{USENIX
  Association}, \bibinfo{address}{Boston, MA}, \bibinfo{pages}{193--208}.
\newblock
\showISBNx{978-1-939133-34-2}
\urldef\tempurl%
\url{https://www.usenix.org/conference/osdi23/presentation/angel}
\showURL{%
\tempurl}


\bibitem[Apple(2024)]%
        {apple-pcc}
\bibfield{author}{\bibinfo{person}{Apple}.} \bibinfo{year}{2024}\natexlab{}.
\newblock \bibinfo{title}{Private Cloud Compute: A new frontier for AI privacy
  in the cloud}.
\newblock
\newblock
\newblock
\shownote{\url{https://security.apple.com/blog/private-cloud-compute/}}.


\bibitem[ARM(2023)]%
        {CCA}
\bibfield{author}{\bibinfo{person}{ARM}.} \bibinfo{year}{2023}\natexlab{}.
\newblock \bibinfo{title}{Learn the architecture - Introducing Arm Confidential
  Compute Architecture - Issue 2.0}.
\newblock
\newblock
\urldef\tempurl%
\url{https://developer.arm.com/documentation/den0125/latest}
\showURL{%
\tempurl}
\newblock
\shownote{last accessed: 06/10/2023}.


\bibitem[Babel et~al\mbox{.}(2023)]%
        {babel2023mysticeti}
\bibfield{author}{\bibinfo{person}{Kushal Babel}, \bibinfo{person}{Andrey
  Chursin}, \bibinfo{person}{George Danezis}, \bibinfo{person}{Lefteris
  Kokoris-Kogias}, {and} \bibinfo{person}{Alberto Sonnino}.}
  \bibinfo{year}{2023}\natexlab{}.
\newblock \bibinfo{title}{Mysticeti: Low-Latency DAG Consensus with Fast Commit
  Path}.
\newblock
\newblock
\showeprint[arxiv]{2310.14821}~[cs.DC]
\urldef\tempurl%
\url{https://arxiv.org/abs/2310.14821v1}
\showURL{%
\tempurl}


\bibitem[Behl et~al\mbox{.}(2017)]%
        {hybster}
\bibfield{author}{\bibinfo{person}{Johannes Behl}, \bibinfo{person}{Tobias
  Distler}, {and} \bibinfo{person}{R\"{u}diger Kapitza}.}
  \bibinfo{year}{2017}\natexlab{}.
\newblock \showarticletitle{Hybrids on Steroids: SGX-Based High Performance
  BFT}. In \bibinfo{booktitle}{\emph{Proceedings of the Twelfth European
  Conference on Computer Systems}} (Belgrade, Serbia)
  \emph{(\bibinfo{series}{EuroSys '17})}. \bibinfo{publisher}{Association for
  Computing Machinery}, \bibinfo{address}{New York, NY, USA},
  \bibinfo{pages}{222–237}.
\newblock
\showISBNx{9781450349383}
\urldef\tempurl%
\url{https://doi.org/10.1145/3064176.3064213}
\showDOI{\tempurl}


\bibitem[Bhagwan et~al\mbox{.}(2004)]%
        {totalrecall}
\bibfield{author}{\bibinfo{person}{Ranjita Bhagwan}, \bibinfo{person}{Kiran
  Tati}, \bibinfo{person}{Yu-Chung Cheng}, \bibinfo{person}{Stefan Savage},
  {and} \bibinfo{person}{Geoffrey~M. Voelker}.}
  \bibinfo{year}{2004}\natexlab{}.
\newblock \showarticletitle{Total recall: system support for automated
  availability management}. In \bibinfo{booktitle}{\emph{Proceedings of the 1st
  Conference on Symposium on Networked Systems Design and Implementation -
  Volume 1}} (San Francisco, California) \emph{(\bibinfo{series}{NSDI'04})}.
  \bibinfo{publisher}{USENIX Association}, \bibinfo{address}{USA},
  \bibinfo{pages}{25}.
\newblock


\bibitem[Birkholz et~al\mbox{.}(2026)]%
        {scitt-ietf}
\bibfield{author}{\bibinfo{person}{H. Birkholz}, \bibinfo{person}{A.
  Delignat-Lavaud}, \bibinfo{person}{C. Fournet}, \bibinfo{person}{Y.
  Deshpande}, {and} \bibinfo{person}{S. Lasker}.}
  \bibinfo{year}{2026}\natexlab{}.
\newblock \bibinfo{title}{{An Architecture for Trustworthy and Transparent
  Digital Supply Chains}}.
\newblock \bibinfo{howpublished}{RFC 9943}.
\newblock
\urldef\tempurl%
\url{https://doi.org/10.17487/RFC9943}
\showDOI{\tempurl}


\bibitem[Brandenburger et~al\mbox{.}(2017)]%
        {brandenburger2017rollback}
\bibfield{author}{\bibinfo{person}{Marcus Brandenburger},
  \bibinfo{person}{Christian Cachin}, \bibinfo{person}{Matthias Lorenz}, {and}
  \bibinfo{person}{Rudiger Kapitza}.} \bibinfo{year}{2017}\natexlab{}.
\newblock \showarticletitle{{ Rollback and Forking Detection for Trusted
  Execution Environments Using Lightweight Collective Memory }}. In
  \bibinfo{booktitle}{\emph{2017 47th Annual IEEE/IFIP International Conference
  on Dependable Systems and Networks (DSN)}}. \bibinfo{publisher}{IEEE Computer
  Society}, \bibinfo{address}{Los Alamitos, CA, USA},
  \bibinfo{pages}{157--168}.
\newblock
\showISSN{2158-3927}
\urldef\tempurl%
\url{https://doi.org/10.1109/DSN.2017.45}
\showDOI{\tempurl}


\bibitem[Bryant(2023)]%
        {googlesecurityreport}
\bibfield{author}{\bibinfo{person}{Jerry Bryant}.}
  \bibinfo{year}{2023}\natexlab{}.
\newblock \bibinfo{title}{Intel – Google TDX Security Review
  \url{https://community.intel.com/t5/Blogs/Products-and-Solutions/Security/Intel-Google-TDX-Security-Review/post/1471177}}.
\newblock
\newblock


\bibitem[Buchman(2016)]%
        {buchman2016tendermint}
\bibfield{author}{\bibinfo{person}{Ethan Buchman}.}
  \bibinfo{year}{2016}\natexlab{}.
\newblock \emph{\bibinfo{title}{Tendermint: Byzantine fault tolerance in the
  age of blockchains}}.
\newblock \bibinfo{thesistype}{Ph.\,D. Dissertation}.
  \bibinfo{school}{University of Guelph}.
\newblock
\urldef\tempurl%
\url{http://hdl.handle.net/10214/9769}
\showURL{%
\tempurl}


\bibitem[Buhren et~al\mbox{.}(2021)]%
        {one-glitch}
\bibfield{author}{\bibinfo{person}{Robert Buhren}, \bibinfo{person}{Hans-Niklas
  Jacob}, \bibinfo{person}{Thilo Krachenfels}, {and}
  \bibinfo{person}{Jean-Pierre Seifert}.} \bibinfo{year}{2021}\natexlab{}.
\newblock \showarticletitle{One Glitch to Rule Them All: Fault Injection
  Attacks Against AMD's Secure Encrypted Virtualization}. In
  \bibinfo{booktitle}{\emph{Proceedings of the 2021 ACM SIGSAC Conference on
  Computer and Communications Security}} (Virtual Event, Republic of Korea)
  \emph{(\bibinfo{series}{CCS '21})}. \bibinfo{publisher}{Association for
  Computing Machinery}, \bibinfo{address}{New York, NY, USA},
  \bibinfo{pages}{2875–2889}.
\newblock
\showISBNx{9781450384544}
\urldef\tempurl%
\url{https://doi.org/10.1145/3460120.3484779}
\showDOI{\tempurl}


\bibitem[Bulck et~al\mbox{.}(2018)]%
        {foreshadow}
\bibfield{author}{\bibinfo{person}{Jo~Van Bulck}, \bibinfo{person}{Marina
  Minkin}, \bibinfo{person}{Ofir Weisse}, \bibinfo{person}{Daniel Genkin},
  \bibinfo{person}{Baris Kasikci}, \bibinfo{person}{Frank Piessens},
  \bibinfo{person}{Mark Silberstein}, \bibinfo{person}{Thomas~F. Wenisch},
  \bibinfo{person}{Yuval Yarom}, {and} \bibinfo{person}{Raoul Strackx}.}
  \bibinfo{year}{2018}\natexlab{}.
\newblock \showarticletitle{Foreshadow: Extracting the Keys to the Intel {SGX}
  Kingdom with Transient {Out-of-Order} Execution}. In
  \bibinfo{booktitle}{\emph{27th USENIX Security Symposium (USENIX Security
  18)}}. \bibinfo{publisher}{USENIX Association}, \bibinfo{address}{Baltimore,
  MD}, \bibinfo{pages}{991{\textendash}1008}.
\newblock
\showISBNx{978-1-939133-04-5}
\urldef\tempurl%
\url{https://www.usenix.org/conference/usenixsecurity18/presentation/bulck}
\showURL{%
\tempurl}


\bibitem[Castro and Liskov(1999)]%
        {PBFT}
\bibfield{author}{\bibinfo{person}{Miguel Castro} {and}
  \bibinfo{person}{Barbara Liskov}.} \bibinfo{year}{1999}\natexlab{}.
\newblock \showarticletitle{Practical Byzantine Fault Tolerance}. In
  \bibinfo{booktitle}{\emph{3rd Symposium on Operating Systems Design and
  Implementation (OSDI 99)}}. \bibinfo{publisher}{USENIX Association},
  \bibinfo{address}{New Orleans, LA}, \bibinfo{pages}{173--186}.
\newblock
\urldef\tempurl%
\url{https://www.usenix.org/conference/osdi-99/practical-byzantine-fault-tolerance}
\showURL{%
\tempurl}


\bibitem[Chen et~al\mbox{.}(2019)]%
        {chen2019sgxpectre}
\bibfield{author}{\bibinfo{person}{Guoxing Chen}, \bibinfo{person}{Sanchuan
  Chen}, \bibinfo{person}{Yuan Xiao}, \bibinfo{person}{Yinqian Zhang},
  \bibinfo{person}{Zhiqiang Lin}, {and} \bibinfo{person}{Ten~H Lai}.}
  \bibinfo{year}{2019}\natexlab{}.
\newblock \showarticletitle{Sgxpectre: Stealing intel secrets from sgx enclaves
  via speculative execution}. In \bibinfo{booktitle}{\emph{2019 IEEE European
  Symposium on Security and Privacy (EuroS\&P)}}. IEEE,
  \bibinfo{pages}{142--157}.
\newblock


\bibitem[Chun et~al\mbox{.}(2007)]%
        {chun2007a2m}
\bibfield{author}{\bibinfo{person}{Byung-Gon Chun}, \bibinfo{person}{Petros
  Maniatis}, \bibinfo{person}{Scott Shenker}, {and} \bibinfo{person}{John
  Kubiatowicz}.} \bibinfo{year}{2007}\natexlab{}.
\newblock \showarticletitle{Attested append-only memory: making adversaries
  stick to their word}. In \bibinfo{booktitle}{\emph{Proceedings of
  Twenty-First ACM SIGOPS Symposium on Operating Systems Principles}}
  (Stevenson, Washington, USA) \emph{(\bibinfo{series}{SOSP '07})}.
  \bibinfo{publisher}{Association for Computing Machinery},
  \bibinfo{address}{New York, NY, USA}, \bibinfo{pages}{189–204}.
\newblock
\showISBNx{9781595935915}
\urldef\tempurl%
\url{https://doi.org/10.1145/1294261.1294280}
\showDOI{\tempurl}


\bibitem[Clement et~al\mbox{.}(2009a)]%
        {clement09upright}
\bibfield{author}{\bibinfo{person}{Allen Clement}, \bibinfo{person}{Manos
  Kapritsos}, \bibinfo{person}{Sangmin Lee}, \bibinfo{person}{Yang Wang},
  \bibinfo{person}{Lorenzo Alvisi}, \bibinfo{person}{Mike Dahlin}, {and}
  \bibinfo{person}{Taylor Riche}.} \bibinfo{year}{2009}\natexlab{a}.
\newblock \showarticletitle{Upright cluster services}. In
  \bibinfo{booktitle}{\emph{Proceedings of the ACM SIGOPS 22nd Symposium on
  Operating Systems Principles}} (Big Sky, Montana, USA)
  \emph{(\bibinfo{series}{SOSP '09})}. \bibinfo{publisher}{Association for
  Computing Machinery}, \bibinfo{address}{New York, NY, USA},
  \bibinfo{pages}{277–290}.
\newblock
\showISBNx{9781605587523}
\urldef\tempurl%
\url{https://doi.org/10.1145/1629575.1629602}
\showDOI{\tempurl}


\bibitem[Clement et~al\mbox{.}(2009b)]%
        {aardvark}
\bibfield{author}{\bibinfo{person}{Allen Clement}, \bibinfo{person}{Edmund
  Wong}, \bibinfo{person}{Lorenzo Alvisi}, \bibinfo{person}{Mike Dahlin}, {and}
  \bibinfo{person}{Mirco Marchetti}.} \bibinfo{year}{2009}\natexlab{b}.
\newblock \showarticletitle{Making Byzantine Fault Tolerant Systems Tolerate
  Byzantine Faults}. In \bibinfo{booktitle}{\emph{Proceedings of the 6th USENIX
  Symposium on Networked Systems Design and Implementation}} (Boston,
  Massachusetts) \emph{(\bibinfo{series}{NSDI'09})}. \bibinfo{publisher}{USENIX
  Association}, \bibinfo{address}{USA}, \bibinfo{pages}{153–168}.
\newblock


\bibitem[Connell et~al\mbox{.}(2024)]%
        {svr3}
\bibfield{author}{\bibinfo{person}{Graeme Connell}, \bibinfo{person}{Vivian
  Fang}, \bibinfo{person}{Rolfe Schmidt}, \bibinfo{person}{Emma Dauterman},
  {and} \bibinfo{person}{Raluca~Ada Popa}.} \bibinfo{year}{2024}\natexlab{}.
\newblock \showarticletitle{Secret Key Recovery in a {Global-Scale}
  {End-to-End} Encryption System}. In \bibinfo{booktitle}{\emph{18th USENIX
  Symposium on Operating Systems Design and Implementation (OSDI 24)}}.
  \bibinfo{publisher}{USENIX Association}, \bibinfo{address}{Santa Clara, CA},
  \bibinfo{pages}{703--719}.
\newblock
\showISBNx{978-1-939133-40-3}
\urldef\tempurl%
\url{https://www.usenix.org/conference/osdi24/presentation/connell}
\showURL{%
\tempurl}


\bibitem[Cooper et~al\mbox{.}(2010)]%
        {ycsb}
\bibfield{author}{\bibinfo{person}{Brian~F. Cooper}, \bibinfo{person}{Adam
  Silberstein}, \bibinfo{person}{Erwin Tam}, \bibinfo{person}{Raghu
  Ramakrishnan}, {and} \bibinfo{person}{Russell Sears}.}
  \bibinfo{year}{2010}\natexlab{}.
\newblock \showarticletitle{Benchmarking cloud serving systems with YCSB}. In
  \bibinfo{booktitle}{\emph{Proceedings of the 1st ACM Symposium on Cloud
  Computing}} (Indianapolis, Indiana, USA) \emph{(\bibinfo{series}{SoCC '10})}.
  \bibinfo{publisher}{Association for Computing Machinery},
  \bibinfo{address}{New York, NY, USA}, \bibinfo{pages}{143–154}.
\newblock
\showISBNx{9781450300360}
\urldef\tempurl%
\url{https://doi.org/10.1145/1807128.1807152}
\showDOI{\tempurl}


\bibitem[Danezis et~al\mbox{.}(2022)]%
        {danezis2022narwhal}
\bibfield{author}{\bibinfo{person}{George Danezis}, \bibinfo{person}{Lefteris
  Kokoris-Kogias}, \bibinfo{person}{Alberto Sonnino}, {and}
  \bibinfo{person}{Alexander Spiegelman}.} \bibinfo{year}{2022}\natexlab{}.
\newblock \showarticletitle{Narwhal and Tusk: a DAG-based mempool and efficient
  BFT consensus}. In \bibinfo{booktitle}{\emph{Proceedings of the Seventeenth
  European Conference on Computer Systems}}. \bibinfo{pages}{34--50}.
\newblock


\bibitem[De~Meulemeester et~al\mbox{.}(2026)]%
        {batteringramsp26}
\bibfield{author}{\bibinfo{person}{Jesse De~Meulemeester},
  \bibinfo{person}{David Oswald}, \bibinfo{person}{Ingrid Verbauwhede}, {and}
  \bibinfo{person}{Jo Van~Bulck}.} \bibinfo{year}{2026}\natexlab{}.
\newblock \showarticletitle{{Battering RAM}: Low-Cost Interposer Attacks on
  Confidential Computing via Dynamic Memory Aliasing}. In
  \bibinfo{booktitle}{\emph{47th {IEEE} Symposium on Security and Privacy
  ({S\&P})}}.
\newblock


\bibitem[De~Meulemeester et~al\mbox{.}(2025)]%
        {badramsp25}
\bibfield{author}{\bibinfo{person}{Jesse De~Meulemeester},
  \bibinfo{person}{Luca Wilke}, \bibinfo{person}{David Oswald},
  \bibinfo{person}{Thomas Eisenbarth}, \bibinfo{person}{Ingrid Verbauwhede},
  {and} \bibinfo{person}{Jo Van~Bulck}.} \bibinfo{year}{2025}\natexlab{}.
\newblock \showarticletitle{{BadRAM}: Practical Memory Aliasing Attacks on
  Trusted Execution Environments}. In \bibinfo{booktitle}{\emph{46th {IEEE}
  Symposium on Security and Privacy ({S\&P})}}.
\newblock


\bibitem[Delignat-Lavaud et~al\mbox{.}(2023)]%
        {Delignat23}
\bibfield{author}{\bibinfo{person}{Antoine Delignat-Lavaud},
  \bibinfo{person}{C\'{e}dric Fournet}, \bibinfo{person}{Kapil Vaswani},
  \bibinfo{person}{Sylvan Clebsch}, \bibinfo{person}{Maik Riechert},
  \bibinfo{person}{Manuel Costa}, {and} \bibinfo{person}{Mark Russinovich}.}
  \bibinfo{year}{2023}\natexlab{}.
\newblock \showarticletitle{Why Should {I} Trust Your Code? {C}onfidential
  {C}omputing Enables Users to Authenticate Code Running in {TEE}s, but Users
  Also Need Evidence This Code is Trustworthy.}
\newblock \bibinfo{journal}{\emph{Queue}} \bibinfo{volume}{21},
  \bibinfo{number}{4} (\bibinfo{date}{sep} \bibinfo{year}{2023}),
  \bibinfo{pages}{94--122}.
\newblock
\showISSN{1542-7730}
\urldef\tempurl%
\url{https://doi.org/10.1145/3623460}
\showDOI{\tempurl}


\bibitem[Dinis et~al\mbox{.}(2023)]%
        {dinis2023rr}
\bibfield{author}{\bibinfo{person}{Baltasar Dinis}, \bibinfo{person}{Peter
  Druschel}, {and} \bibinfo{person}{Rodrigo Rodrigues}.}
  \bibinfo{year}{2023}\natexlab{}.
\newblock \showarticletitle{{RR}: A fault model for efficient {TEE}
  replication}. In \bibinfo{booktitle}{\emph{The Network and Distributed System
  Security Symposium}}. Internet Society, \bibinfo{publisher}{Internet
  Society}, \bibinfo{address}{San Diego, CA, USA},
  \bibinfo{numpages}{18}~pages.
\newblock
\urldef\tempurl%
\url{https://doi.org/10.14722/ndss.2023.24001}
\showDOI{\tempurl}


\bibitem[Duan et~al\mbox{.}(2015)]%
        {hbft}
\bibfield{author}{\bibinfo{person}{Sisi Duan}, \bibinfo{person}{Sean Peisert},
  {and} \bibinfo{person}{Karl~N Levitt}.} \bibinfo{year}{2015}\natexlab{}.
\newblock \showarticletitle{hBFT: speculative Byzantine fault tolerance with
  minimum cost}.
\newblock \bibinfo{journal}{\emph{IEEE Transactions on Dependable and Secure
  Computing}} \bibinfo{volume}{12}, \bibinfo{number}{1} (\bibinfo{year}{2015}),
  \bibinfo{pages}{58--70}.
\newblock
\urldef\tempurl%
\url{https://doi.org/10.1109/TDSC.2014.2312331}
\showDOI{\tempurl}


\bibitem[Duan and Zhang(2022)]%
        {duan2022dyno}
\bibfield{author}{\bibinfo{person}{Sisi Duan} {and} \bibinfo{person}{Haibin
  Zhang}.} \bibinfo{year}{2022}\natexlab{}.
\newblock \showarticletitle{Foundations of Dynamic BFT}. In
  \bibinfo{booktitle}{\emph{2022 IEEE Symposium on Security and Privacy (SP)}}.
  \bibinfo{pages}{1317--1334}.
\newblock
\urldef\tempurl%
\url{https://doi.org/10.1109/SP46214.2022.9833787}
\showDOI{\tempurl}


\bibitem[Dutta et~al\mbox{.}(2005)]%
        {Dutta2005BestCaseCO}
\bibfield{author}{\bibinfo{person}{Partha Dutta}, \bibinfo{person}{Rachid
  Guerraoui}, {and} \bibinfo{person}{Marko Vukolic}.}
  \bibinfo{year}{2005}\natexlab{}.
\newblock \bibinfo{booktitle}{\emph{Best-Case Complexity of Asynchronous
  Byzantine Consensus}}.
\newblock \bibinfo{type}{{T}echnical {R}eport} EPFL/IC/200499.
  \bibinfo{institution}{EPFL}, \bibinfo{address}{Lausanne, Switzerland}.
\newblock
\urldef\tempurl%
\url{https://lpdwww.epfl.ch/upload/documents/publications/567931850DGV-feb-05.pdf}
\showURL{%
\tempurl}


\bibitem[Fei et~al\mbox{.}(2021)]%
        {fei2021security}
\bibfield{author}{\bibinfo{person}{Shufan Fei}, \bibinfo{person}{Zheng Yan},
  \bibinfo{person}{Wenxiu Ding}, {and} \bibinfo{person}{Haomeng Xie}.}
  \bibinfo{year}{2021}\natexlab{}.
\newblock \showarticletitle{Security vulnerabilities of SGX and
  countermeasures: A survey}.
\newblock \bibinfo{journal}{\emph{ACM Computing Surveys (CSUR)}}
  \bibinfo{volume}{54}, \bibinfo{number}{6} (\bibinfo{year}{2021}),
  \bibinfo{pages}{1--36}.
\newblock


\bibitem[Gelashvili et~al\mbox{.}(2022)]%
        {jolteon-ditto}
\bibfield{author}{\bibinfo{person}{Rati Gelashvili}, \bibinfo{person}{Lefteris
  Kokoris-Kogias}, \bibinfo{person}{Alberto Sonnino},
  \bibinfo{person}{Alexander Spiegelman}, {and} \bibinfo{person}{Zhuolun
  Xiang}.} \bibinfo{year}{2022}\natexlab{}.
\newblock \showarticletitle{Jolteon and Ditto: Network-Adaptive Efficient
  Consensus with Asynchronous Fallback}. In \bibinfo{booktitle}{\emph{Financial
  Cryptography and Data Security: 26th International Conference, FC 2022,
  Grenada, May 2–6, 2022, Revised Selected Papers}} (Grenada, Grenada).
  \bibinfo{publisher}{Springer-Verlag}, \bibinfo{address}{Berlin, Heidelberg},
  \bibinfo{pages}{296–315}.
\newblock
\showISBNx{978-3-031-18282-2}
\urldef\tempurl%
\url{https://doi.org/10.1007/978-3-031-18283-9_14}
\showDOI{\tempurl}


\bibitem[Giridharan et~al\mbox{.}(2024)]%
        {autobahn}
\bibfield{author}{\bibinfo{person}{Neil Giridharan}, \bibinfo{person}{Florian
  Suri-Payer}, \bibinfo{person}{Ittai Abraham}, \bibinfo{person}{Lorenzo
  Alvisi}, {and} \bibinfo{person}{Natacha Crooks}.}
  \bibinfo{year}{2024}\natexlab{}.
\newblock \showarticletitle{Autobahn: Seamless high speed BFT}. In
  \bibinfo{booktitle}{\emph{Proceedings of the ACM SIGOPS 30th Symposium on
  Operating Systems Principles}} (Austin, TX, USA) \emph{(\bibinfo{series}{SOSP
  '24})}. \bibinfo{publisher}{Association for Computing Machinery},
  \bibinfo{address}{New York, NY, USA}, \bibinfo{pages}{1–23}.
\newblock
\showISBNx{9798400712517}
\urldef\tempurl%
\url{https://doi.org/10.1145/3694715.3695942}
\showDOI{\tempurl}


\bibitem[Giridharan et~al\mbox{.}(2023)]%
        {giridharan2023beegees}
\bibfield{author}{\bibinfo{person}{Neil Giridharan}, \bibinfo{person}{Florian
  Suri-Payer}, \bibinfo{person}{Matthew Ding}, \bibinfo{person}{Heidi Howard},
  \bibinfo{person}{Ittai Abraham}, {and} \bibinfo{person}{Natacha Crooks}.}
  \bibinfo{year}{2023}\natexlab{}.
\newblock \showarticletitle{BeeGees: stayin'alive in chained BFT}. In
  \bibinfo{booktitle}{\emph{Proceedings of the 2023 ACM Symposium on Principles
  of Distributed Computing}}. \bibinfo{pages}{233--243}.
\newblock


\bibitem[Golan~Gueta et~al\mbox{.}(2019)]%
        {gueta2019sbft}
\bibfield{author}{\bibinfo{person}{Guy Golan~Gueta}, \bibinfo{person}{Ittai
  Abraham}, \bibinfo{person}{Shelly Grossman}, \bibinfo{person}{Dahlia Malkhi},
  \bibinfo{person}{Benny Pinkas}, \bibinfo{person}{Michael Reiter},
  \bibinfo{person}{Dragos-Adrian Seredinschi}, \bibinfo{person}{Orr Tamir},
  {and} \bibinfo{person}{Alin Tomescu}.} \bibinfo{year}{2019}\natexlab{}.
\newblock \showarticletitle{SBFT: A Scalable and Decentralized Trust
  Infrastructure}. In \bibinfo{booktitle}{\emph{2019 49th Annual IEEE/IFIP
  International Conference on Dependable Systems and Networks (DSN)}}.
  \bibinfo{publisher}{IEEE}, \bibinfo{address}{USA}, \bibinfo{pages}{568--580}.
\newblock
\urldef\tempurl%
\url{https://doi.org/10.1109/DSN.2019.00063}
\showDOI{\tempurl}


\bibitem[Gupta et~al\mbox{.}(2023)]%
        {gupta2023bft}
\bibfield{author}{\bibinfo{person}{Suyash Gupta}, \bibinfo{person}{Sajjad
  Rahnama}, \bibinfo{person}{Shubham Pandey}, \bibinfo{person}{Natacha Crooks},
  {and} \bibinfo{person}{Mohammad Sadoghi}.} \bibinfo{year}{2023}\natexlab{}.
\newblock \showarticletitle{Dissecting BFT Consensus: In Trusted Components we
  Trust!}. In \bibinfo{booktitle}{\emph{Proceedings of the Eighteenth European
  Conference on Computer Systems}} (Rome, Italy)
  \emph{(\bibinfo{series}{EuroSys '23})}. \bibinfo{publisher}{Association for
  Computing Machinery}, \bibinfo{address}{New York, NY, USA},
  \bibinfo{pages}{521–539}.
\newblock
\showISBNx{9781450394871}
\urldef\tempurl%
\url{https://doi.org/10.1145/3552326.3587455}
\showDOI{\tempurl}


\bibitem[Haeberlen et~al\mbox{.}(2007)]%
        {peerreview}
\bibfield{author}{\bibinfo{person}{Andreas Haeberlen}, \bibinfo{person}{Petr
  Kouznetsov}, {and} \bibinfo{person}{Peter Druschel}.}
  \bibinfo{year}{2007}\natexlab{}.
\newblock \showarticletitle{PeerReview: practical accountability for
  distributed systems}. In \bibinfo{booktitle}{\emph{Proceedings of
  Twenty-First ACM SIGOPS Symposium on Operating Systems Principles}}
  (Stevenson, Washington, USA) \emph{(\bibinfo{series}{SOSP '07})}.
  \bibinfo{publisher}{Association for Computing Machinery},
  \bibinfo{address}{New York, NY, USA}, \bibinfo{pages}{175–188}.
\newblock
\showISBNx{9781595935915}
\urldef\tempurl%
\url{https://doi.org/10.1145/1294261.1294279}
\showDOI{\tempurl}


\bibitem[Haeberlen and Kuznetsov(2009)]%
        {haeberlen09opodis}
\bibfield{author}{\bibinfo{person}{Andreas Haeberlen} {and}
  \bibinfo{person}{Petr Kuznetsov}.} \bibinfo{year}{2009}\natexlab{}.
\newblock \showarticletitle{The Fault Detection Problem}. In
  \bibinfo{booktitle}{\emph{Proceedings of the 13th International Conference on
  Principles of Distributed Systems}} (N\^{\i}mes, France)
  \emph{(\bibinfo{series}{OPODIS '09})}. \bibinfo{publisher}{Springer-Verlag},
  \bibinfo{address}{Berlin, Heidelberg}, \bibinfo{pages}{99–114}.
\newblock
\showISBNx{9783642108761}
\urldef\tempurl%
\url{https://doi.org/10.1007/978-3-642-10877-8_10}
\showDOI{\tempurl}


\bibitem[Haeberlen et~al\mbox{.}(2005)]%
        {glacier}
\bibfield{author}{\bibinfo{person}{Andreas Haeberlen}, \bibinfo{person}{Alan
  Mislove}, {and} \bibinfo{person}{Peter Druschel}.}
  \bibinfo{year}{2005}\natexlab{}.
\newblock \showarticletitle{Glacier: highly durable, decentralized storage
  despite massive correlated failures}. In
  \bibinfo{booktitle}{\emph{Proceedings of the 2nd Conference on Symposium on
  Networked Systems Design \& Implementation - Volume 2}}
  \emph{(\bibinfo{series}{NSDI'05})}. \bibinfo{publisher}{USENIX Association},
  \bibinfo{address}{USA}, \bibinfo{pages}{143–158}.
\newblock


\bibitem[Howard et~al\mbox{.}(2023)]%
        {ccf}
\bibfield{author}{\bibinfo{person}{Heidi Howard}, \bibinfo{person}{Fritz
  Alder}, \bibinfo{person}{Edward Ashton}, \bibinfo{person}{Amaury Chamayou},
  \bibinfo{person}{Sylvan Clebsch}, \bibinfo{person}{Manuel Costa},
  \bibinfo{person}{Antoine Delignat-Lavaud}, \bibinfo{person}{C\'{e}dric
  Fournet}, \bibinfo{person}{Andrew Jeffery}, \bibinfo{person}{Matthew Kerner},
  \bibinfo{person}{Fotios Kounelis}, \bibinfo{person}{Markus~A. Kuppe},
  \bibinfo{person}{Julien Maffre}, \bibinfo{person}{Mark Russinovich}, {and}
  \bibinfo{person}{Christoph~M. Wintersteiger}.}
  \bibinfo{year}{2023}\natexlab{}.
\newblock \showarticletitle{Confidential Consortium Framework: Secure
  Multiparty Applications with Confidentiality, Integrity, and High
  Availability}.
\newblock \bibinfo{journal}{\emph{Proc. VLDB Endow.}} \bibinfo{volume}{17},
  \bibinfo{number}{2} (\bibinfo{date}{oct} \bibinfo{year}{2023}),
  \bibinfo{pages}{225–240}.
\newblock
\showISSN{2150-8097}
\urldef\tempurl%
\url{https://doi.org/10.14778/3626292.3626304}
\showDOI{\tempurl}


\bibitem[Intel(2021)]%
        {TDX}
\bibfield{author}{\bibinfo{person}{Intel}.} \bibinfo{year}{2021}\natexlab{}.
\newblock \bibinfo{title}{Intel Trust Domain Extensions - White Paper}.
\newblock
\newblock
\newblock
\shownote{\url{https://cdrdv2.intel.com/v1/dl/getContent/690419}}.


\bibitem[Keidar et~al\mbox{.}(2021)]%
        {dag-rider}
\bibfield{author}{\bibinfo{person}{Idit Keidar}, \bibinfo{person}{Eleftherios
  Kokoris-Kogias}, \bibinfo{person}{Oded Naor}, {and}
  \bibinfo{person}{Alexander Spiegelman}.} \bibinfo{year}{2021}\natexlab{}.
\newblock \showarticletitle{All You Need is DAG}. In
  \bibinfo{booktitle}{\emph{Proceedings of the 2021 ACM Symposium on Principles
  of Distributed Computing}} (Virtual Event, Italy)
  \emph{(\bibinfo{series}{PODC'21})}. \bibinfo{publisher}{Association for
  Computing Machinery}, \bibinfo{address}{New York, NY, USA},
  \bibinfo{pages}{165–175}.
\newblock
\showISBNx{9781450385480}
\urldef\tempurl%
\url{https://doi.org/10.1145/3465084.3467905}
\showDOI{\tempurl}


\bibitem[Kotla et~al\mbox{.}(2010)]%
        {zyzzyva}
\bibfield{author}{\bibinfo{person}{Ramakrishna Kotla}, \bibinfo{person}{Lorenzo
  Alvisi}, \bibinfo{person}{Mike Dahlin}, \bibinfo{person}{Allen Clement},
  {and} \bibinfo{person}{Edmund Wong}.} \bibinfo{year}{2010}\natexlab{}.
\newblock \showarticletitle{Zyzzyva: Speculative {B}yzantine {F}ault
  {T}olerance}.
\newblock \bibinfo{journal}{\emph{ACM Transactions on Computer Systems (TOCS)}}
  \bibinfo{volume}{27}, \bibinfo{number}{4}, Article \bibinfo{articleno}{7}
  (\bibinfo{date}{Jan.} \bibinfo{year}{2010}), \bibinfo{numpages}{39}~pages.
\newblock
\showISSN{0734-2071}
\urldef\tempurl%
\url{https://doi.org/10.1145/1658357.1658358}
\showDOI{\tempurl}


\bibitem[Labs(2022)]%
        {PDO}
\bibfield{author}{\bibinfo{person}{Hyperledger Labs}.}
  \bibinfo{year}{2022}\natexlab{}.
\newblock \bibinfo{title}{Hyperledger Private Data Objects}.
\newblock
\newblock
\newblock
\shownote{\url{https://github.com/hyperledger-labs/private-data-objects}}.


\bibitem[Lamport(2006)]%
        {lamport2006lower}
\bibfield{author}{\bibinfo{person}{Leslie Lamport}.}
  \bibinfo{year}{2006}\natexlab{}.
\newblock \showarticletitle{Lower bounds for asynchronous consensus}.
\newblock \bibinfo{journal}{\emph{Distrib. Comput.}} \bibinfo{volume}{19},
  \bibinfo{number}{2} (\bibinfo{date}{Oct.} \bibinfo{year}{2006}),
  \bibinfo{pages}{104–125}.
\newblock
\showISSN{0178-2770}
\urldef\tempurl%
\url{https://doi.org/10.1007/s00446-006-0155-x}
\showDOI{\tempurl}


\bibitem[Len et~al\mbox{.}(2024)]%
        {optiks}
\bibfield{author}{\bibinfo{person}{Julia Len}, \bibinfo{person}{Melissa Chase},
  \bibinfo{person}{Esha Ghosh}, \bibinfo{person}{Kim Laine}, {and}
  \bibinfo{person}{Radames~Cruz Moreno}.} \bibinfo{year}{2024}\natexlab{}.
\newblock \showarticletitle{{OPTIKS}: An Optimized Key Transparency System}. In
  \bibinfo{booktitle}{\emph{33rd USENIX Security Symposium (USENIX Security
  24)}}. \bibinfo{publisher}{USENIX Association},
  \bibinfo{address}{Philadelphia, PA}, \bibinfo{pages}{4355--4372}.
\newblock
\showISBNx{978-1-939133-44-1}
\urldef\tempurl%
\url{https://www.usenix.org/conference/usenixsecurity24/presentation/len}
\showURL{%
\tempurl}


\bibitem[Levin et~al\mbox{.}(2009)]%
        {trinc}
\bibfield{author}{\bibinfo{person}{Dave Levin}, \bibinfo{person}{John~R.
  Douceur}, \bibinfo{person}{Jacob~R. Lorch}, {and} \bibinfo{person}{Thomas
  Moscibroda}.} \bibinfo{year}{2009}\natexlab{}.
\newblock \showarticletitle{TrInc: small trusted hardware for large distributed
  systems}. In \bibinfo{booktitle}{\emph{Proceedings of the 6th USENIX
  Symposium on Networked Systems Design and Implementation}} (Boston,
  Massachusetts) \emph{(\bibinfo{series}{NSDI'09})}. \bibinfo{publisher}{USENIX
  Association}, \bibinfo{address}{USA}, \bibinfo{pages}{1–14}.
\newblock


\bibitem[Li et~al\mbox{.}(2022)]%
        {li2022systematic}
\bibfield{author}{\bibinfo{person}{Mengyuan Li}, \bibinfo{person}{Luca Wilke},
  \bibinfo{person}{Jan Wichelmann}, \bibinfo{person}{Thomas Eisenbarth},
  \bibinfo{person}{Radu Teodorescu}, {and} \bibinfo{person}{Yinqian Zhang}.}
  \bibinfo{year}{2022}\natexlab{}.
\newblock \showarticletitle{A systematic look at ciphertext side channels on
  AMD SEV-SNP}. In \bibinfo{booktitle}{\emph{2022 IEEE Symposium on Security
  and Privacy (SP)}}. IEEE, \bibinfo{pages}{337--351}.
\newblock


\bibitem[Li et~al\mbox{.}(2021)]%
        {li2021cipherleaks}
\bibfield{author}{\bibinfo{person}{Mengyuan Li}, \bibinfo{person}{Yinqian
  Zhang}, \bibinfo{person}{Huibo Wang}, \bibinfo{person}{Kang Li}, {and}
  \bibinfo{person}{Yueqiang Cheng}.} \bibinfo{year}{2021}\natexlab{}.
\newblock \showarticletitle{{CIPHERLEAKS}: Breaking Constant-time Cryptography
  on {AMD SEV} via the Ciphertext Side Channel}. In
  \bibinfo{booktitle}{\emph{30th USENIX Security Symposium (USENIX Security
  21)}}. \bibinfo{publisher}{USENIX Association}, \bibinfo{pages}{717--732}.
\newblock
\urldef\tempurl%
\url{https://www.usenix.org/conference/usenixsecurity21/presentation/li-mengyuan}
\showURL{%
\tempurl}


\bibitem[Liu et~al\mbox{.}(2016)]%
        {xft}
\bibfield{author}{\bibinfo{person}{Shengyun Liu}, \bibinfo{person}{Paolo
  Viotti}, \bibinfo{person}{Christian Cachin}, \bibinfo{person}{Vivien
  Qu\'{e}ma}, {and} \bibinfo{person}{Marko Vukolic}.}
  \bibinfo{year}{2016}\natexlab{}.
\newblock \showarticletitle{XFT: Practical Fault Tolerance beyond Crashes}. In
  \bibinfo{booktitle}{\emph{Proceedings of the 12th USENIX Conference on
  Operating Systems Design and Implementation}} (Savannah, GA, USA)
  \emph{(\bibinfo{series}{OSDI'16})}. \bibinfo{publisher}{USENIX Association},
  \bibinfo{address}{USA}, \bibinfo{pages}{485–500}.
\newblock
\showISBNx{9781931971331}


\bibitem[Lund(2019)]%
        {signal-recovery}
\bibfield{author}{\bibinfo{person}{Joshua Lund}.}
  \bibinfo{year}{2019}\natexlab{}.
\newblock \bibinfo{title}{Technology Preview for secure value recovery}.
\newblock
\newblock
\newblock
\shownote{\url{https://signal.org/blog/secure-value-recovery/}}.


\bibitem[Malkhi et~al\mbox{.}(2019)]%
        {flexible-bft}
\bibfield{author}{\bibinfo{person}{Dahlia Malkhi}, \bibinfo{person}{Kartik
  Nayak}, {and} \bibinfo{person}{Ling Ren}.} \bibinfo{year}{2019}\natexlab{}.
\newblock \showarticletitle{Flexible Byzantine Fault Tolerance}. In
  \bibinfo{booktitle}{\emph{Proceedings of the 2019 ACM SIGSAC Conference on
  Computer and Communications Security}} (London, United Kingdom)
  \emph{(\bibinfo{series}{CCS '19})}. \bibinfo{publisher}{Association for
  Computing Machinery}, \bibinfo{address}{New York, NY, USA},
  \bibinfo{pages}{1041–1053}.
\newblock
\showISBNx{9781450367479}
\urldef\tempurl%
\url{https://doi.org/10.1145/3319535.3354225}
\showDOI{\tempurl}


\bibitem[Melara et~al\mbox{.}(2015)]%
        {CONIKS}
\bibfield{author}{\bibinfo{person}{Marcela~S Melara}, \bibinfo{person}{Aaron
  Blankstein}, \bibinfo{person}{Joseph Bonneau}, \bibinfo{person}{Edward~W
  Felten}, {and} \bibinfo{person}{Michael~J Freedman}.}
  \bibinfo{year}{2015}\natexlab{}.
\newblock \showarticletitle{{CONIKS}: Bringing Key Transparency to End Users}.
  In \bibinfo{booktitle}{\emph{24th USENIX Security Symposium (USENIX Security
  15)}}. \bibinfo{publisher}{USENIX Association}, \bibinfo{address}{Washington,
  D.C.}, \bibinfo{pages}{383--398}.
\newblock
\urldef\tempurl%
\url{https://www.usenix.org/conference/usenixsecurity15/technical-sessions/presentation/melara}
\showURL{%
\tempurl}


\bibitem[Messadi et~al\mbox{.}(2022)]%
        {messadi2022splitbft}
\bibfield{author}{\bibinfo{person}{Ines Messadi}, \bibinfo{person}{Markus~Horst
  Becker}, \bibinfo{person}{Kai Bleeke}, \bibinfo{person}{Leander Jehl},
  \bibinfo{person}{Sonia~Ben Mokhtar}, {and} \bibinfo{person}{R\"{u}diger
  Kapitza}.} \bibinfo{year}{2022}\natexlab{}.
\newblock \showarticletitle{SplitBFT: Improving Byzantine Fault Tolerance
  Safety Using Trusted Compartments}. In \bibinfo{booktitle}{\emph{Proceedings
  of the 23rd ACM/IFIP International Middleware Conference}} (Quebec, QC,
  Canada) \emph{(\bibinfo{series}{Middleware '22})}.
  \bibinfo{publisher}{Association for Computing Machinery},
  \bibinfo{address}{New York, NY, USA}, \bibinfo{pages}{56–68}.
\newblock
\showISBNx{9781450393409}
\urldef\tempurl%
\url{https://doi.org/10.1145/3528535.3531516}
\showDOI{\tempurl}


\bibitem[Microsoft(2022a)]%
        {scitt}
\bibfield{author}{\bibinfo{person}{Microsoft}.}
  \bibinfo{year}{2022}\natexlab{a}.
\newblock \bibinfo{title}{scitt-ccf-ledger}.
\newblock
\newblock
\newblock
\shownote{\url{https://github.com/microsoft/scitt-ccf-ledger}}.


\bibitem[Microsoft(2022b)]%
        {did-ccf}
\bibfield{author}{\bibinfo{person}{Microsoft}.}
  \bibinfo{year}{2022}\natexlab{b}.
\newblock \bibinfo{title}{{W3C} {DID} for {Confidential Consortium Framework}}.
\newblock
\newblock
\newblock
\shownote{\url{https://github.com/microsoft/did-ccf}}.


\bibitem[Microsoft(2023)]%
        {privacy-sandbox}
\bibfield{author}{\bibinfo{person}{Microsoft}.}
  \bibinfo{year}{2023}\natexlab{}.
\newblock \bibinfo{title}{Privacy Sandbox Key Management System (KMS) for
  Azure}.
\newblock
\newblock
\newblock
\shownote{\url{https://github.com/microsoft/azure-privacy-sandbox-kms}}.


\bibitem[Microsoft(2024a)]%
        {ccf-ledger}
\bibfield{author}{\bibinfo{person}{Microsoft}.}
  \bibinfo{year}{2024}\natexlab{a}.
\newblock \bibinfo{title}{{Azure Confidential Ledger}}.
\newblock
\newblock
\newblock
\shownote{\url{https://azure.microsoft.com/en-gb/products/azure-confidential-ledger},
  last accessed on 04/10/24}.


\bibitem[Microsoft(2024b)]%
        {microsoft-ccf}
\bibfield{author}{\bibinfo{person}{Microsoft}.}
  \bibinfo{year}{2024}\natexlab{b}.
\newblock \bibinfo{title}{{Confidential Consortium Framework, Microsoft}}.
\newblock
\newblock
\newblock
\shownote{\url{https://ccf.microsoft.com/}, last accessed on 04/10/24}.


\bibitem[Mishra et~al\mbox{.}(2024)]%
        {fault-in-bus}
\bibfield{author}{\bibinfo{person}{Nimish Mishra}, \bibinfo{person}{Anirban
  Chakraborty}, {and} \bibinfo{person}{Debdeep Mukhopadhyay}.}
  \bibinfo{year}{2024}\natexlab{}.
\newblock \showarticletitle{Faults in Our Bus: Novel Bus Fault Attack to Break
  {ARM} {TrustZone}}. In \bibinfo{booktitle}{\emph{Proceedings 2024 Network and
  Distributed System Security Symposium}} (San Diego, {CA}, {USA}).
  \bibinfo{publisher}{Internet Society}.
\newblock
\showISBNx{978-1-891562-93-8}
\urldef\tempurl%
\url{https://doi.org/10.14722/ndss.2024.24499}
\showDOI{\tempurl}


\bibitem[Morbitzer et~al\mbox{.}(2018)]%
        {severed}
\bibfield{author}{\bibinfo{person}{Mathias Morbitzer}, \bibinfo{person}{Manuel
  Huber}, \bibinfo{person}{Julian Horsch}, {and} \bibinfo{person}{Sascha
  Wessel}.} \bibinfo{year}{2018}\natexlab{}.
\newblock \showarticletitle{SEVered: Subverting AMD's Virtual Machine
  Encryption}. In \bibinfo{booktitle}{\emph{Proceedings of the 11th European
  Workshop on Systems Security}} (Porto, Portugal)
  \emph{(\bibinfo{series}{EuroSec'18})}. \bibinfo{publisher}{Association for
  Computing Machinery}, \bibinfo{address}{New York, NY, USA}, Article
  \bibinfo{articleno}{1}, \bibinfo{numpages}{6}~pages.
\newblock
\showISBNx{9781450356527}
\urldef\tempurl%
\url{https://doi.org/10.1145/3193111.3193112}
\showDOI{\tempurl}


\bibitem[Niu et~al\mbox{.}(2025)]%
        {Achilles}
\bibfield{author}{\bibinfo{person}{Jianyu Niu}, \bibinfo{person}{Xiaoqing Wen},
  \bibinfo{person}{Guanlong Wu}, \bibinfo{person}{Shengqi Liu},
  \bibinfo{person}{Jiangshan Yu}, {and} \bibinfo{person}{Yinqian Zhang}.}
  \bibinfo{year}{2025}\natexlab{}.
\newblock \showarticletitle{Achilles: Efficient TEE-Assisted BFT Consensus via
  Rollback Resilient Recovery}. In \bibinfo{booktitle}{\emph{Proceedings of the
  Twentieth European Conference on Computer Systems}} (Rotterdam, Netherlands)
  \emph{(\bibinfo{series}{EuroSys '25})}. \bibinfo{publisher}{Association for
  Computing Machinery}, \bibinfo{address}{New York, NY, USA},
  \bibinfo{pages}{193–210}.
\newblock
\showISBNx{9798400711961}
\urldef\tempurl%
\url{https://doi.org/10.1145/3689031.3717457}
\showDOI{\tempurl}


\bibitem[Oasis(2024)]%
        {oasis}
\bibfield{author}{\bibinfo{person}{Oasis}.} \bibinfo{year}{2024}\natexlab{}.
\newblock \bibinfo{title}{Security \& {TEE}s}.
\newblock
\newblock
\newblock
\shownote{\url{https://oasisprotocol.org/security-and-tees}}.


\bibitem[Ongaro and Ousterhout(2014)]%
        {raft}
\bibfield{author}{\bibinfo{person}{Diego Ongaro} {and} \bibinfo{person}{John
  Ousterhout}.} \bibinfo{year}{2014}\natexlab{}.
\newblock \showarticletitle{In search of an understandable consensus
  algorithm}. In \bibinfo{booktitle}{\emph{Proceedings of the 2014 USENIX
  Conference on USENIX Annual Technical Conference}} (Philadelphia, PA)
  \emph{(\bibinfo{series}{USENIX ATC'14})}. \bibinfo{publisher}{USENIX
  Association}, \bibinfo{address}{USA}, \bibinfo{pages}{305–320}.
\newblock
\showISBNx{9781931971102}


\bibitem[Porto et~al\mbox{.}(2015)]%
        {porto2015visigoth}
\bibfield{author}{\bibinfo{person}{Daniel Porto}, \bibinfo{person}{Jo\~{a}o
  Leit\~{a}o}, \bibinfo{person}{Cheng Li}, \bibinfo{person}{Allen Clement},
  \bibinfo{person}{Aniket Kate}, \bibinfo{person}{Flavio Junqueira}, {and}
  \bibinfo{person}{Rodrigo Rodrigues}.} \bibinfo{year}{2015}\natexlab{}.
\newblock \showarticletitle{Visigoth fault tolerance}. In
  \bibinfo{booktitle}{\emph{Proceedings of the Tenth European Conference on
  Computer Systems}} (Bordeaux, France) \emph{(\bibinfo{series}{EuroSys '15})}.
  \bibinfo{publisher}{Association for Computing Machinery},
  \bibinfo{address}{New York, NY, USA}, Article \bibinfo{articleno}{8},
  \bibinfo{numpages}{14}~pages.
\newblock
\showISBNx{9781450332385}
\urldef\tempurl%
\url{https://doi.org/10.1145/2741948.2741979}
\showDOI{\tempurl}


\bibitem[Qiu et~al\mbox{.}(2019)]%
        {qiu2019voltjockey}
\bibfield{author}{\bibinfo{person}{Pengfei Qiu}, \bibinfo{person}{Dongsheng
  Wang}, \bibinfo{person}{Yongqiang Lyu}, {and} \bibinfo{person}{Gang Qu}.}
  \bibinfo{year}{2019}\natexlab{}.
\newblock \showarticletitle{VoltJockey: Breaking SGX by software-controlled
  voltage-induced hardware faults}. In \bibinfo{booktitle}{\emph{2019 Asian
  Hardware Oriented Security and Trust Symposium (AsianHOST)}}. IEEE,
  \bibinfo{pages}{1--6}.
\newblock


\bibitem[Russinovich(2024)]%
        {azure-cai}
\bibfield{author}{\bibinfo{person}{Mark Russinovich}.}
  \bibinfo{year}{2024}\natexlab{}.
\newblock \bibinfo{title}{Azure AI Confidential Inferencing: Technical
  Deep-Dive}.
\newblock
\newblock
\newblock
\shownote{\url{https://techcommunity.microsoft.com/blog/azureconfidentialcomputingblog/azure-ai-confidential-inferencing-technical-deep-dive/4253150}}.


\bibitem[Russinovich(2025)]%
        {cts-announcement}
\bibfield{author}{\bibinfo{person}{Mark Russinovich}.}
  \bibinfo{year}{2025}\natexlab{}.
\newblock \bibinfo{title}{Enhancing software supply chain security with
  Microsoft’s Signing Transparency}.
\newblock
\newblock
\newblock
\shownote{\url{https://azure.microsoft.com/en-us/blog/enhancing-software-supply-chain-security-with-microsofts-signing-transparency/}}.


\bibitem[Schl{\"u}ter et~al\mbox{.}(2024)]%
        {heckler}
\bibfield{author}{\bibinfo{person}{Benedict Schl{\"u}ter},
  \bibinfo{person}{Supraja Sridhara}, \bibinfo{person}{Mark Kuhne},
  \bibinfo{person}{Andrin Bertschi}, {and} \bibinfo{person}{Shweta Shinde}.}
  \bibinfo{year}{2024}\natexlab{}.
\newblock \showarticletitle{{HECKLER}: Breaking Confidential {VMs} with
  Malicious Interrupts}. In \bibinfo{booktitle}{\emph{33rd USENIX Security
  Symposium (USENIX Security 24)}}. \bibinfo{publisher}{USENIX Association},
  \bibinfo{address}{Philadelphia, PA}, \bibinfo{pages}{3459--3476}.
\newblock
\showISBNx{978-1-939133-44-1}
\urldef\tempurl%
\url{https://www.usenix.org/conference/usenixsecurity24/presentation/schl{\"u}ter}
\showURL{%
\tempurl}


\bibitem[Schlüter et~al\mbox{.}(2024)]%
        {malicious-vc}
\bibfield{author}{\bibinfo{person}{Benedict Schlüter},
  \bibinfo{person}{Supraja Sridhara}, \bibinfo{person}{Andrin Bertschi}, {and}
  \bibinfo{person}{Shweta Shinde}.} \bibinfo{year}{2024}\natexlab{}.
\newblock \bibinfo{title}{WeSee: Using Malicious \#VC Interrupts to Break AMD
  SEV-SNP}.
\newblock
\newblock
\showeprint[arxiv]{2404.03526}~[cs.CR]
\urldef\tempurl%
\url{https://arxiv.org/abs/2404.03526}
\showURL{%
\tempurl}


\bibitem[Shamis et~al\mbox{.}(2022)]%
        {shamis2022iaaccf}
\bibfield{author}{\bibinfo{person}{Alex Shamis}, \bibinfo{person}{Peter
  Pietzuch}, \bibinfo{person}{Burcu Canakci}, \bibinfo{person}{Miguel Castro},
  \bibinfo{person}{Cedric Fournet}, \bibinfo{person}{Edward Ashton},
  \bibinfo{person}{Amaury Chamayou}, \bibinfo{person}{Sylvan Clebsch},
  \bibinfo{person}{Antoine Delignat-Lavaud}, \bibinfo{person}{Matthew Kerner},
  \bibinfo{person}{Julien Maffre}, \bibinfo{person}{Olga Vrousgou},
  \bibinfo{person}{Christoph~M. Wintersteiger}, \bibinfo{person}{Manuel Costa},
  {and} \bibinfo{person}{Mark Russinovich}.} \bibinfo{year}{2022}\natexlab{}.
\newblock \showarticletitle{{IA-CCF}: Individual Accountability for
  Permissioned Ledgers}. In \bibinfo{booktitle}{\emph{19th USENIX Symposium on
  Networked Systems Design and Implementation (NSDI 22)}}.
  \bibinfo{publisher}{USENIX Association}, \bibinfo{address}{Renton, WA},
  \bibinfo{pages}{467--491}.
\newblock
\showISBNx{978-1-939133-27-4}
\urldef\tempurl%
\url{https://www.usenix.org/conference/nsdi22/presentation/shamis}
\showURL{%
\tempurl}


\bibitem[Sheng et~al\mbox{.}(2021)]%
        {bft-forensics}
\bibfield{author}{\bibinfo{person}{Peiyao Sheng}, \bibinfo{person}{Gerui Wang},
  \bibinfo{person}{Kartik Nayak}, \bibinfo{person}{Sreeram Kannan}, {and}
  \bibinfo{person}{Pramod Viswanath}.} \bibinfo{year}{2021}\natexlab{}.
\newblock \showarticletitle{BFT Protocol Forensics}. In
  \bibinfo{booktitle}{\emph{Proceedings of the 2021 ACM SIGSAC Conference on
  Computer and Communications Security}} (Virtual Event, Republic of Korea)
  \emph{(\bibinfo{series}{CCS '21})}. \bibinfo{publisher}{Association for
  Computing Machinery}, \bibinfo{address}{New York, NY, USA},
  \bibinfo{pages}{1722–1743}.
\newblock
\showISBNx{9781450384544}
\urldef\tempurl%
\url{https://doi.org/10.1145/3460120.3484566}
\showDOI{\tempurl}


\bibitem[Spiegelman et~al\mbox{.}(2022)]%
        {Bullshark}
\bibfield{author}{\bibinfo{person}{Alexander Spiegelman}, \bibinfo{person}{Neil
  Giridharan}, \bibinfo{person}{Alberto Sonnino}, {and}
  \bibinfo{person}{Lefteris Kokoris-Kogias}.} \bibinfo{year}{2022}\natexlab{}.
\newblock \showarticletitle{Bullshark: DAG BFT Protocols Made Practical}. In
  \bibinfo{booktitle}{\emph{Proceedings of the 2022 ACM SIGSAC Conference on
  Computer and Communications Security}} (Los Angeles, CA, USA)
  \emph{(\bibinfo{series}{CCS '22})}. \bibinfo{publisher}{Association for
  Computing Machinery}, \bibinfo{address}{New York, NY, USA},
  \bibinfo{pages}{2705–2718}.
\newblock
\showISBNx{9781450394505}
\urldef\tempurl%
\url{https://doi.org/10.1145/3548606.3559361}
\showDOI{\tempurl}


\bibitem[Van~Bulck et~al\mbox{.}(2017)]%
        {sgx-step}
\bibfield{author}{\bibinfo{person}{Jo Van~Bulck}, \bibinfo{person}{Frank
  Piessens}, {and} \bibinfo{person}{Raoul Strackx}.}
  \bibinfo{year}{2017}\natexlab{}.
\newblock \showarticletitle{SGX-Step: A Practical Attack Framework for Precise
  Enclave Execution Control}. In \bibinfo{booktitle}{\emph{Proceedings of the
  2nd Workshop on System Software for Trusted Execution}} (Shanghai, China)
  \emph{(\bibinfo{series}{SysTEX'17})}. \bibinfo{publisher}{Association for
  Computing Machinery}, \bibinfo{address}{New York, NY, USA}, Article
  \bibinfo{articleno}{4}, \bibinfo{numpages}{6}~pages.
\newblock
\showISBNx{9781450350976}
\urldef\tempurl%
\url{https://doi.org/10.1145/3152701.3152706}
\showDOI{\tempurl}


\bibitem[Veronese et~al\mbox{.}(2013)]%
        {minbft}
\bibfield{author}{\bibinfo{person}{Giuliana~Santos Veronese},
  \bibinfo{person}{Miguel Correia}, \bibinfo{person}{Alysson~Neves Bessani},
  \bibinfo{person}{Lau~Cheuk Lung}, {and} \bibinfo{person}{Paulo Verissimo}.}
  \bibinfo{year}{2013}\natexlab{}.
\newblock \showarticletitle{Efficient Byzantine Fault-Tolerance}.
\newblock \bibinfo{journal}{\emph{IEEE Trans. Comput.}} \bibinfo{volume}{62},
  \bibinfo{number}{1} (\bibinfo{date}{Jan.} \bibinfo{year}{2013}),
  \bibinfo{pages}{16–30}.
\newblock
\showISSN{0018-9340}
\urldef\tempurl%
\url{https://doi.org/10.1109/TC.2011.221}
\showDOI{\tempurl}


\bibitem[Wang et~al\mbox{.}(2022)]%
        {Wang22}
\bibfield{author}{\bibinfo{person}{Weili Wang}, \bibinfo{person}{Sen Deng},
  \bibinfo{person}{Jianyu Niu}, \bibinfo{person}{Michael~K. Reiter}, {and}
  \bibinfo{person}{Yinqian Zhang}.} \bibinfo{year}{2022}\natexlab{}.
\newblock \showarticletitle{ENGRAFT: Enclave-Guarded Raft on Byzantine Faulty
  Nodes}. In \bibinfo{booktitle}{\emph{Proceedings of the 2022 ACM SIGSAC
  Conference on Computer and Communications Security}} (Los Angeles, CA, USA)
  \emph{(\bibinfo{series}{CCS '22})}. \bibinfo{publisher}{Association for
  Computing Machinery}, \bibinfo{address}{New York, NY, USA},
  \bibinfo{pages}{2841--2855}.
\newblock
\showISBNx{9781450394505}
\urldef\tempurl%
\url{https://doi.org/10.1145/3548606.3560639}
\showDOI{\tempurl}


\bibitem[Wood(2025)]%
        {ethereum}
\bibfield{author}{\bibinfo{person}{Gavin Wood}.}
  \bibinfo{year}{2025}\natexlab{}.
\newblock \bibinfo{title}{Ethereum: A Secure Decentralised Generalised
  Transaction Ledger}.
\newblock
\newblock
\urldef\tempurl%
\url{https://ethereum.github.io/yellowpaper/paper.pdf}
\showURL{%
\tempurl}
\newblock
\shownote{Shanghai Version, efc5f9a}.


\bibitem[Yin et~al\mbox{.}(2019)]%
        {hotstuff}
\bibfield{author}{\bibinfo{person}{Maofan Yin}, \bibinfo{person}{Dahlia
  Malkhi}, \bibinfo{person}{Michael~K. Reiter}, \bibinfo{person}{Guy~Golan
  Gueta}, {and} \bibinfo{person}{Ittai Abraham}.}
  \bibinfo{year}{2019}\natexlab{}.
\newblock \showarticletitle{HotStuff: BFT Consensus with Linearity and
  Responsiveness}. In \bibinfo{booktitle}{\emph{Proceedings of the 2019 ACM
  Symposium on Principles of Distributed Computing}} (Toronto ON, Canada)
  \emph{(\bibinfo{series}{PODC '19})}. \bibinfo{publisher}{Association for
  Computing Machinery}, \bibinfo{address}{New York, NY, USA},
  \bibinfo{pages}{347–356}.
\newblock
\showISBNx{9781450362177}
\urldef\tempurl%
\url{https://doi.org/10.1145/3293611.3331591}
\showDOI{\tempurl}


\end{thebibliography}

\pagebreak

\appendix

\section{Safety Proof}
\label{sec:safety}

We will prove:
\begin{itemize}
    \item Safety of slow audit (\S\ref{subsec:audit-safety-without-fast-path}),
    \item Safety of fast audit (\S\ref{subsec:audit-safety-with-fast-path}),
    \item Commit safety (\S\ref{subsec:commit-safety})
\end{itemize}

\subsection{Notation}
We (re-)introduce some shorthand notations:
\begin{itemize}[leftmargin=*]
    \item For two batches $B$ and $B'$, if $B$ is the parent of $B'$, ie, $B'.parent = B$, we write $B \ianc B'$.
    \item Two batches $B_1$ and $B_2$ are said to be \textbf{conflicting} if neither $B_1 \anc B_2$ nor $B_2 \anc B_1$.
    \item The ancestor relation is the transitive closure of the parent relation, as is notated as $B \anc B'$. We also say $B'$ \textit{extends} $B$. As a convention, the relation is reflexive, ie, $B \anc B$ for all batches $B$.
    \item An \qc formed of $N - u$ signed votes in $v$ on the batch $b$ is notated as $qc^v_b$.
    \item A fast \qc formed of $N$ signed votes is denoted by $fqc^v_b$.
    \item To explicitly state that a batch $B$ embeds the \qc $qc^v_b$ (ie, $B.auditQC_{anc} = qc^v_b$), we shall write the batch $B$ as $(B \cat qc^v_b)$.
    The view of a batch and its embedded \qc is the same ($B.v = v$).
    \item We say a batch $B'$ extends an \qc $qc^v_b$, iff $(B \cat qc^v_b) \anc B'$ where $B.auditQC_{anc} = qc^v_b$.
\end{itemize}

\par \textbf{Commit and Audit Conditions} \sys{} is a chained protocol. A batch at sequence number $n$ is cryptographically linked to all previous batches with sequence numbers $n-1$ or lower. It is associated with a unique history. Hence, committing/auditing a batch $B$ indirectly commits/audits all its ancestors.

A batch $B$ is audited by slow audit logic if the following holds:

$$ B \anc B_1 \anc (B_2 \cat qc_{B_1}^v) \anc B_3 \anc (B_4 \cat qc_{B_3}^v) $$

Note that, due to hash-chaining both $B$ and $B_1$ are audited.
As stated above, $B_1.v ~=~ B_2.v ~=~ B_3.v ~=~ B_4.v ~=~ v$, the view in which both the \qcs were proposed.

A batch $B$ is audited through the fast audit logic if the following holds:

$$ B \anc B_1 \anc (B_2 \cat fqc_{B_1}^v) $$

The two rules above constitute the most general conditions for auditing.
If the batch $B$ is signed and the auditing QC(s) form directly on $B$, $B$ is considered audited without the need for a distinct batch $B_1$. (Note that the $\anc$ relation is reflexive, so notationally $B_1 = B$ here.)
However, $B$ is indirectly audited through a distinct extending batch $B_1$ if: \one $B$ is not signed or \two $B$ is proposed in a lower view than $B_1$.

\subsection{Safety of slow audit}
\label{subsec:audit-safety-without-fast-path}

We structure the proof as a sequence of lemmas as follows:

\begin{enumerate}[label={[\hyperref[lemma:\arabic*]{L\ref*{lemma:\arabic*}}]}]
    \item In the same view, a leader cannot create two \qcs for two conflicting batches.
    \item Except for the \textsc{NewView} message at the start of a view $v$, no \qc can be formed on a batch $B$ proposed in view $v$, if $B$ does not extend $qc_{NV^v}^v$.
    \item $qc_{NV^v}^v$ can't be formed if $NV^v$ extends a branch that disobeys the branch choice rules. For all branches that qualify the branch choice rules, $qc_{NV^v}^v$ extends at most one such branch.
    \item If $v$ is the first view where a correct node considers a batch $B$ audited, branches collected in view $v + 1$ from any subset of $(N - u)$ replicas will have at least one that includes $qc_{B_1}^v$ where $B \anc B_1$.
    \item If $v$ is the first view where a correct node considers a batch $B$ audited, view $v$ must be stable.
    \item If $v$ is the first view where a correct node considers a batch $B$ audited, branches collected in any greater view from any subset of $(N - u)$ replicas will have at least one branch with the highest \qc extending $qc_B^v$, and that branch will be selected by the branch choice rules.
    \item If two batches $B_1, B_2$ are considered audited by correct nodes, either $B_1 \anc B_2$ or $B_2 \anc B_1$ (ie, they are non-conflicting).
\end{enumerate}

\begin{lemma}
\label{lemma:1}
In the same view, a leader cannot create two \qcs for two conflicting batches.
\end{lemma}
\begin{proof}
We prove this by contradiction.
Assume two \qcs are formed for conflicting batches $B_1$ and $B_2$.
The quorums that formed these \qcs intersect in at least $(f_{safe} + 1)$ nodes.
So, at least one correct node voted for both $B_1$ and $B_2$ in the same view.
But neither $B_1 \anc B_2$ nor $B_2 \anc B_1$, which means the correct node must have overwritten its local branch within the same view.
This is disallowed by the checks done before an \textsc{Append-Entry} is processed: within a view, a correct replica never overwrites its local branch.
So, it can't overwrite its local branch, thereby reaching a contradiction.
\end{proof}

\begin{lemma}
\label{lemma:2}
Except for the \textsc{NewView} message at the start of a view $v$,
no \qc can be formed on a batch $B$ proposed in view $v$, if $B$ does not extend $qc_{NV^v}^v$.
\end{lemma}
\begin{proof}
We prove this by contradiction.
Assume that there exists such a batch $B$ proposed in view $v$ that does not extend from $qc_{NV^v}^v$.
If an \qc is formed on this batch, at least $N - u$ signed votes must have been collected on batch $B$.
Out of which, at most $f_{safe}$ votes can be from compromised nodes. So, at least $N - u - f_{safe}$ votes must be from \textit{correct} nodes.
Since $N \geq 2u + f_{safe} + 1 \implies N - u - f_{safe} \geq u + 1 \geq 1$, there must have been at least $1$ correct node that voted for such a batch $B$.
However, by protocol logic, except for the \textsc{NewView} message at the start of a view $v$, a correct node never votes for a batch proposed in view $v$, if it doesn't extend $qc_{NV^v}^v$. Thereby, we arrive at a contradiction.
\end{proof}

A corollary to this lemma is that no Byzantine leader can make progress in auditing without stabilizing its view first.

\begin{lemma}
\label{lemma:3}
$qc_{NV^v}^v$ can't be formed if $NV^v$ extends a branch that disobeys the branch choice rules. For all branches that qualify the branch choice rules, $qc_{NV^v}^v$ extends at most one such branch.
\end{lemma}
\begin{proof}
A correct node doesn't vote for a $NV^v$ batch unless it extends a branch deterministically chosen by the branch Choice Rule applied to the $(N - u)$ branches embedded in $NV^v$.

If two branches qualify the branch choice rule (especially if the leader has collected more than $N-u$ \textsc{ViewChange} messages, a Byzantine leader can equivocate and create two qualifying $NV_1^v$ and $NV_2^v$, but a \qc can be formed in at most one of $NV_1^v$ and $NV_2^v$.
The new leader cannot get a \qc on any batch that extends the other equivocated $NV^v$.
\end{proof}

The rest of the lemmas therefore would prove two things:
(1) The $(N - u)$ branches embedded in an $NV^v$ always contain one branch that has all those batches that may be considered audited by a correct node, in views $\leq (v - 1)$.
(2) The branch choice rules will always choose that branch.

\begin{lemma}
\label{lemma:4}
If $v$ is the first view where a correct node considers a batch $B$ audited, branches collected in view $v + 1$ from any subset of $(N - u)$ replicas will have at least one that includes $qc_{B_1}^v$ where $B \anc B_1$.
\end{lemma}
\begin{proof}
If a correct node considers batch $B$ audited in view $v$, it must have seen:

$$ B \anc B_1 \anc (B_2 \cat qc_{B_1}^v) \anc B_3 \anc (B_4 \cat qc_{B_3}^v) $$
A quorum of $(N - u)$ replicas voted for $B_3$ that extends $qc_{B_1}^v$ in view $v$.
Any subset of $(N - u)$ replicas intersects with this quorum in at least $(f_{safe} + 1)$ nodes.
So at least one correct node's branch will have $qc_{B_1}^v$.
\end{proof}

\begin{lemma}
\label{lemma:5}
If $v$ is the first view where a correct node considers a batch $B$ audited, view $v$ must be stable.
\end{lemma}
\begin{proof}
Recall that a view $v$ is stable if $NV^v \ianc (B \cat qc_{NV^v}^v)$ is formed. A correct considers a view $v$ stable once it receives a branch containing $NV^v \ianc (B \cat qc_{NV^v}^v)$.

To prove this lemma, consider two cases: \one $B$ is proposed in view $v$ itself, \two $B$ is proposed in some view $\leq v - 1$.

For case \one, $B$ can't get a \qc if it doesn't extend $qc_{NV^v}^v$ by \Cref{lemma:2}. But the existence of $qc_{NV^v}^v$ convinces a correct node that the view is stable.

For case \two, note that $B$ must have been chosen by the branch selection rules. So, $B \anc NV^v$ by construction.
If a correct replica moved to view $v$, it is not going to vote on any proposal in lower views and it is not going to vote for any batch proposed in view $v$ if it doesn't extend $NV^v$.
Since $v$ is the first view in which $B$ was considered audited, the \qcs that commit $B$ must be in view $v$.
This implies that there are \qcs extending $NV^v$.
The first such \qc must be $qc_{NV^v}^v$ (again by \Cref{lemma:2}).

\end{proof}

\begin{lemma}
\label{lemma:6}
If $v$ is the first view in which a correct node considers a batch $B$ audited, branches collected in any greater view from any subset of $(N - u)$ replicas will have at least one with the highest \qc extending $qc_B^v$, and therefore, will be selected by the branch choice rules.
\end{lemma}
\begin{proof}
If a correct node considers a batch $B$ audited, it must have at least seen $N - u$ votes on $qc_B^v$.
So, at least $N - u - f_{safe}$ branches during the view change to the next view will extend from $qc_B^v$.
The branch choice rules always prefer a branch with the highest \qc view.
We shall prove, by induction, that all \qcs in views $> v$ extend $qc_B^v$.

Clearly, if all views $> v$ are unstable, no \qc with higher view exists and thus the branch with $qc_B^v$ is always selected. (\Cref{lemma:1} guarantees that there will not be a \qc conflicting with $qc_B^v$ in the same view $v$.)

Recall that if a \qc is formed in view $v' > v$, view $v'$ must be stable (from \Cref{lemma:2}).
For the base case of induction, assume that $v'$ is the first stable view after $v$.
During the view change to $v'$, the leader in view $v'$ will observe at least $N - u - f_{safe}$ branches extending from $qc_B^v$.
While proposing $NV^{v'}$, there isn't any \qc with a higher view than $v$, so the leader must choose to extend from $qc_B^v$ as well.
Thus, the base case holds.
Note that, as $v'$ is stable, $(B_1 \cat qc_B^v) \anc NV^{v'} \ianc (B_2 \cat qc_{NV^{v'}}^{v'})$ holds.
In other words, any \qc proposed in view $v'$ extends $qc_B^v$.
However, there is no guarantee that $qc_{NV^{v'}}^{v'}$ will be seen in \textsc{ViewChange} messages for higher views: the leader in view $v'$ may simply crash before disseminating it.

For induction, assume $\langle v, v_1, v_2, \ldots, v_m, v'\rangle$ is an increasing sequence of stable views, and our proposition holds for all views in this sequence up to $v_m$.
During the view change to $v'$, if the leader observes any branch with a \qc with view $> v$, the induction assumption compels the \qc to extend $qc_B^v$.
If no such \qc exists (i.e., all leaders in view $> v$ crashed/timed out before disseminating their stabilizing \qcs), by quorum intersection, we should still see at least $N - u - f_{safe}$ branches extending from $qc_B^v$.
In both cases, the branch with the highest \qc view extends $qc_B^v$, and therefore gets selected.

\end{proof}

\begin{lemma}
\label{lemma:7}
If two batches $B_1, B_2$ are considered audited by correct nodes, either $B_1 \anc B_2$ or $B_2 \anc B_1$ (ie, they are non-conflicting).
\end{lemma}
\begin{proof}
If both batches are considered audited first in the same view $v$, this clearly follows from \Cref{lemma:1} since the \qc in view $v$ for $B_1$ and $B_2$ cannot conflict.

Without loss of generality, assume $B_1$ was first audited in view $v$ and $B_2$ in view $v + v'$.
Since $B_1$ was audited first in view $v$ by some correct node, \Cref{lemma:6} mandates that during the view change to view $v + v'$, a branch containing $qc_{B_1}^v$ must be present and selected by the branch choice rules.
Thus, $B_1 \anc (B_1' \cat qc_{B_1}^v) \anc NV^{v+v'} \anc B_2$ holds.
Therefore, since $B_2$ is audited in some view $v + v'$ by some (other) correct node, it must extend from $qc_{B_1}^v$.
\end{proof}

This concludes the proof for the safety of slow audits.

\subsection{Safety of fast audit}
\label{subsec:audit-safety-with-fast-path}

To recap, a batch is considered to be audited by fast path if a \qc forms with all $N$ votes.
During view change, the new leader picks a branch that contains the highest batch appearing in $\geq N - u - f_{safe}$ \textsc{ViewChange} messages.
The stabilizing \qc, ie, $qc_{NV^v}^v$, is always proposed with $N - u$ votes, and is not used for auditing by fast path.

Additionally, fast path is only enabled when $N-u > 2f_{safe}$. To see why, consider the following:
Since auditing via fast path doesn't need a chain of \qcs, as necessitated by the slow path, during view change, a fast-path-audited batch can present itself without any \qc on it.
Since the view change quorum size is $N - u$, ideally, a fast-path-audited batch must appear in at least $N - u - f_{safe}$ \textsc{ViewChange} messages.
Out of these $N - u$ \textsc{ViewChange} messages, $f_{safe}$ could exhibit a conflicting batch, simply due to Byzantine nodes spoofing.
We need the fast-path-audited batch to "win" over these spoofed batches.
Hence $N - u - f_{safe} > f_{safe} \implies N-u > 2f_{safe}$.
With $N = 2u + f_{safe} + 1$, this inequality boils down to $u \geq r$.

\begin{lemma}
\label{lemma:9}
During view change, if two conflicting batches proposed in the same view have $N - u - f_{safe}$ votes each, neither could have been audited by the fast path.
\end{lemma}
\begin{proof}
We shall prove this by contradiction. Assume one of the batches audited by fast path, gathering votes from all $N$ replicas.
Out of these, $N - f_{safe}$ votes are from correct replicas.
This set of replicas intersect the other set of $N - u - f_{safe}$ nodes in $N - u - 2f_{safe}$ \textit{correct} nodes.
Since we only enable fast path when $N - u > 2f_{safe}$, at least one correct node must have voted for the conflicting batch, which is a contradiction.
\end{proof}

\begin{lemma}
\label{lemma:10}
If a branch is audited by fast path in view $v$, it is always selected by the branch choice rules for views $> v$.
\end{lemma}
\begin{proof}
Without loss of generality, assume $v$ is the first view where a batch $B$ is audited by fast path.
By \Cref{lemma:5}, view $v$ must be stable.
If $B$ was proposed before view $v$, all $N$ replicas must have voted for the branch $B \anc NV^v \anc (B_0 || qc_{NV^v}^v)$.
If $B$ was proposed in view $v$, all $N$ replicas must have voted for the branch $NV^v \anc (B_0 \cat qc_{NV^v}^v) \anc B$.

In the former case, \Cref{lemma:6} guarantees that in all future views a branch extending $qc_{NV^v}^v$ is found and selected.

In the latter case, assume that all views from $v + 1$ to $v + k$ are unstable. If $v + k + 1$ becomes stable, the branch made stable will always be chosen in all future views.
So, it suffices to prove that the branch chosen in view $v + k + 1$ extends $B$.
The fast path audited batch $B$ appears at least $N - u - f_{safe}$ times in the \textsc{ViewChange} quorum for view $v + 1$.
Hence, the branch choice rules dictate that a branch extending $B$ must be chosen.
Importantly, the number of replicas with $B$ in their local branch doesn't decrease.
For each subsequent view, it is clear to see how the branch choice rules preserve this invariant.
Hence even during the view change for view $v + k + 1$, there are at least $N - u - f_{safe}$ branches with $B$.
\end{proof}

\subsection{Proofs for commit safety}
\label{subsec:commit-safety}

Commit safety assumes no replica is malicious.
For a replica to consider a batch committed in view $v$, the view must be stabilized.

Within a view, establishing commit safety is trivial: the leader never proposes conflicting batches.
Since the leader for a given view is already known, there will be at most one proposal for any given log position in a given view.
We shall prove that a committed batch is not lost across views.

\begin{lemma}
\label{lemma:11}
If in all views $v$ and onwards, all replicas behave non-maliciously, and a correct node considers a batch $B$ committed in view $v$, $B$ appears in the chosen branch by the leader in views $> v$.
\end{lemma}
\begin{proof}
If $B$ was proposed in a view earlier than $v$, then the commit condition dictates that $B \anc NV^v \anc (B_0 \cat qc_{NV^v}^v)$ holds.
If $B$ was proposed in view $v$, $B$ must extend $qc_{NV^v}^v$, i.e., $NV^v \anc (B_0 \cat qc_{NV^v}^v) \anc B$ holds. (Here $B$ and $B_0$ could be the same batch)

Batch $B$ was committed by a correct node implies that $B$ was replicated to a quorum of \majority nodes and the leader in view $v$ collected \majority votes on $B$ or any of its descendants.
The \majority quorum intersects the $N - u$ view change quorum in more than one (correct) nodes.

In the latter case above, if all views from $v + 1$ to $v + k$ are unstable, we can show by induction that $(B_0 \cat qc_{NV^v}^v) \anc B$ is part of the longest branch with possibly only \textsc{NewView} messages from views $v + 1$ to $v + k$ in the suffix.
Since $v$ is the highest view in which a \qc was formed, the branch containing $(B_0 \cat qc_{NV^v}^v) \anc B$ is always chosen.

Therefore, in a future view, if a leader forms a stabilizing \qc, it must extend $(B_0 \cat qc_{NV^v}^v) \anc B$. So, all future \qcs extend from $B$.
Hence, the branch choice rules always ensure a branch with $B$ is chosen.

The proof for the former case directly follows from the proof above.
Note that in the former case, $B \anc NV^v \anc (B_0 \cat qc_{NV^v}^v)$ holds and to commit the majority must have been collected on a batch that extends $(B_0 \cat qc_{NV^v}^v)$.
So, $B_0$ must also have been committed, and thus, by the proof above, must be preserved across views, along with all of its ancestors, including $B$.

This concludes the proof for commit safety.
\end{proof}

Note that the proof doesn't demand that in all views before view $v$, the replicas have to be non-malicious.
The process of view stabilization ensures that any conflicting branch proposed in earlier views will never win against a branch with committed entries in view $v$.
We call this property \textbf{Influence Freedom}.
This ensures that a system with broken TEEs can be patched in-place without going through a costly reconfiguration protocol by simply triggering a view change after the patching is done.
This reduces the downtime needed for patching.

Nonetheless, we provide a reconfiguration protocol in \Cref{sec:reconfiguration} to handle membership changes.

\section{Liveness Proof}
\label{sec:liveness}
\sys{} guarantees liveness only after Global Stabilization Time (GST).
Being a partially synchronous protocol, it relies on timeouts for progress.
Let $\Delta$ be the maximum message transmission delay after GST, and $s$ be the interval at which a leader signs batches.
We set the timeout, $T_{out} > (4s + 5) \Delta$.

To recap, a correct node in view $v$ waits up to $T_{out}$ time to hear about new audited batches.
If it times out, it stops receiving messages in view $v$ and broadcasts a \textsc{ViewChange} message for view $v + 1$.
If any node receives $f_{safe} + 1$ \textsc{ViewChange} messages for some view $v + k$, it stops accepting messages in view $v$ and broadcasts its own \textsc{ViewChange} message for view $v + k$.
If any node receives $N - u$ \textsc{ViewChange} messages (including itself) for view $v + k$, or it receives the new view message $NV^{v+k}$ from the leader in view $v + k$, it increments its view to $v + k$, updates its local fork and votes for $NV^{v+k}$ (given $NV^{v+k}$ is valid).

It is important for a node to not increase its view before it has received $N - u$ \textsc{ViewChange} messages. Otherwise, nodes in different views can always keep jumping to higher and higher views without ever having enough nodes in the same view to make progress.
If a node times out waiting for $N - u$ \textsc{ViewChange} messages for view $v + 1$, it simply re-broadcasts its \textsc{ViewChange} message for view $v + 1$.

Once GST has arrived, we shall prove the following lemmas hold:
\begin{enumerate}[leftmargin=*]
    \item It takes at most $2\Delta$ time to synchronize views.
    \item An honest leader stabilizes its view in at most $2\Delta$ time.
    \item After view stabilization, it takes $3\Delta$ to propagate the knowledge of committing a batch.
    \item After view stabilization, it takes $(4s + 1)\Delta$ to propagate the knowledge of auditing a batch.
\end{enumerate}

\begin{lemma}
    Let $v$ be the first view after GST with an honest leader, and $t$ be the time at which the first correct node enters view $v$. All correct nodes enter view $v$ before $t + 2\Delta$.
\end{lemma}

\begin{proof}
    If the first correct node enters view $v$ in time $t$, it must have seen either \one $N - u$ \textsc{ViewChange} messages for view $v$, or \two a valid $NV^v$ from leader in view $v$.
    In the second case, all correct nodes with receive $NV^v$ before $t + \Delta$.

    In the first case, all other nodes must have received at least $(N - u) - u$ \textsc{ViewChange} messages by $t + \Delta$.
    Since $N \geq 2u + f_{safe} + 1 \implies N - 2u \geq f_{safe} + 1$, on receiving these \textsc{ViewChange} message, all of these correct nodes must have broadcast their own \textsc{ViewChange} message for view $v$, which causes all correct nodes to receive $N - u$ \textsc{ViewChange} within $\Delta$ more time. Hence, each correct node enters view $v$ within $t + 2\Delta$.
\end{proof}

\begin{lemma}
    Let $t$ be the time after GST at which an honest leader in view $v$ receives $N - u$ \textsc{ViewChange} message. The leader forms the stabilizing \qc, $qc_{NV^v}^v$ within $t + 2\Delta$.
\end{lemma}
\begin{proof}
    This lemma holds since the nodes on receiving the new view message $NV^v$ immediately jump to view $v$ and vote on $NV^v$. It takes $\Delta$ to broadcast $NV^v$ and another $\Delta$ to collect all the votes.
\end{proof}

\begin{lemma}
\label{lemma:commit-liveness}
    Let $t$ be the time after GST at which an honest leader in view $v$ forms $qc_{NV^v}^v$, all correct nodes consider batches proposed at $t' > t$ committed within $t' + 3\Delta$.
\end{lemma}
\begin{proof}
    Once the view $v$ is stabilized, each batch needs $\Delta$ to propagate, $\Delta$ for voting, and $\Delta$ to propagate the next batch that carries the commitQC. Therefore, each correct node considers a batch committed after $t' + 3\Delta$.
\end{proof}

\begin{lemma}
\label{lemma:audit-liveness}
    Let $t$ be the time after GST at which an honest leader in view $v$ forms $qc_{NV^v}^v$, all correct nodes consider batches proposed at $t' > t$ audited within $t' + (4s + 1)\Delta$.
\end{lemma}
\begin{proof}
    This is analogous to \Cref{lemma:commit-liveness} with the exception that \qcs are disseminated every $s$ batches.
    Even with one unresponsive node, we lose the liveness for the fast path. Hence, it takes $2s$ round-trips to form the $2$ consecutive \qcs that audit a batch and one message delay to propagate the last \qc. Hence, it takes at most $(2\times2s + 1)\Delta = (4s + 1)\Delta$ to audit a batch after view stabilization.
\end{proof}

Finally, we prove that, after GST, a client request generated at time $t$ gets audited within a finite time.
(As is standard for leader-based protocols, we make the simplifying assumption that clients are not censored.)

\begin{lemma}
    Let $t$ be the time after GST at which a client submits a request to the system. If $T_{out} > (4s + 5)\Delta$, the request gets audited by at most $t + uT_{out} + (4s + 5)\Delta$.
\end{lemma}
\begin{proof}
    Using $T_{out} > (4s + 5)\Delta$ guarantees that a node remains in view $v$ long enough for other nodes to synchronize, the leader to stabilize the view, and then audit a batch of client requests, even without the fast path.

    After GST, at most $u$ consecutive views may have Byzantine leaders and therefore may refuse to make any progress. The system, therefore, needs to time out through $u$ views in $uT_{out}$ time and then require at most $(4s + 5)\Delta$ time to audit the client's request.
\end{proof}

\section{Necessity of View Stabilization}
\label{sec:vs}
In the proofs above, we have extensively used the notion of view stabilization.
In this section, we investigate why this step is necessary for the correctness of the protocol.
Consider a non-Byzantine system before GST, where the leader in view 1 gets an \qc on batch 1 and embeds it in batch 4.
Batches 1-4 were committed properly by replicating them to a majority.
But due to asynchrony, the leader in view 1 is not able to propagate batch 4.
The cluster times out, and the leader in view 2 extends from batch 3, as it is the last committed batch of view 1.
In view 2, the leader commits two new batches and then times out while another batch is in-flight.
The network partition is lifted and leader in view 1 participates in the view change for view 3.
The branches seen by the leader in view 3 are depicted in \Cref{fig:vs1}.

\begin{figure*}
    \centering
    \includegraphics[width=0.70\linewidth]{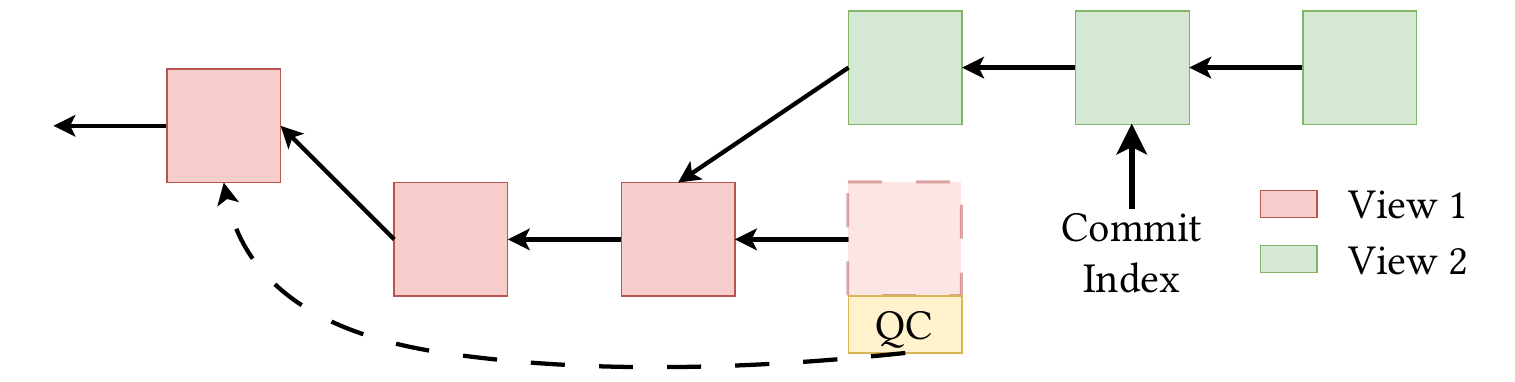}
    \caption{Necessity of View Stabilization: The dilemma in branch selection.}
    \label{fig:vs1}
\vspace{-0.3cm}
\end{figure*}

By the branch selection rules, the leader in view 3 must choose the "hidden" batch 4 from view 1, since it has the highest \qc view.
However, doing that causes committed batches in view 2 to rollback, without there being any Byzantine behavior in the system.
View stabilization would have required there be another \qc in view 2 before all the committed batches in view 2 were proposed.

Given this situation, one might consider changing the branch selection rules to prevent such mishaps.
Since, for audit safety, it is only necessary to preserve the batch for which the \qc is formed, not the \qc itself, the branch selection rules could be relaxed to select a branch containing the batch that has the \qc with the highest view, but not necessarily the branch that contains the \qc itself.
While it does solve the dilemma above, using this rule in a Byzantine setting leads to forgetting of audited batches.

\begin{figure*}
    \centering
    \includegraphics[width=0.70\linewidth]{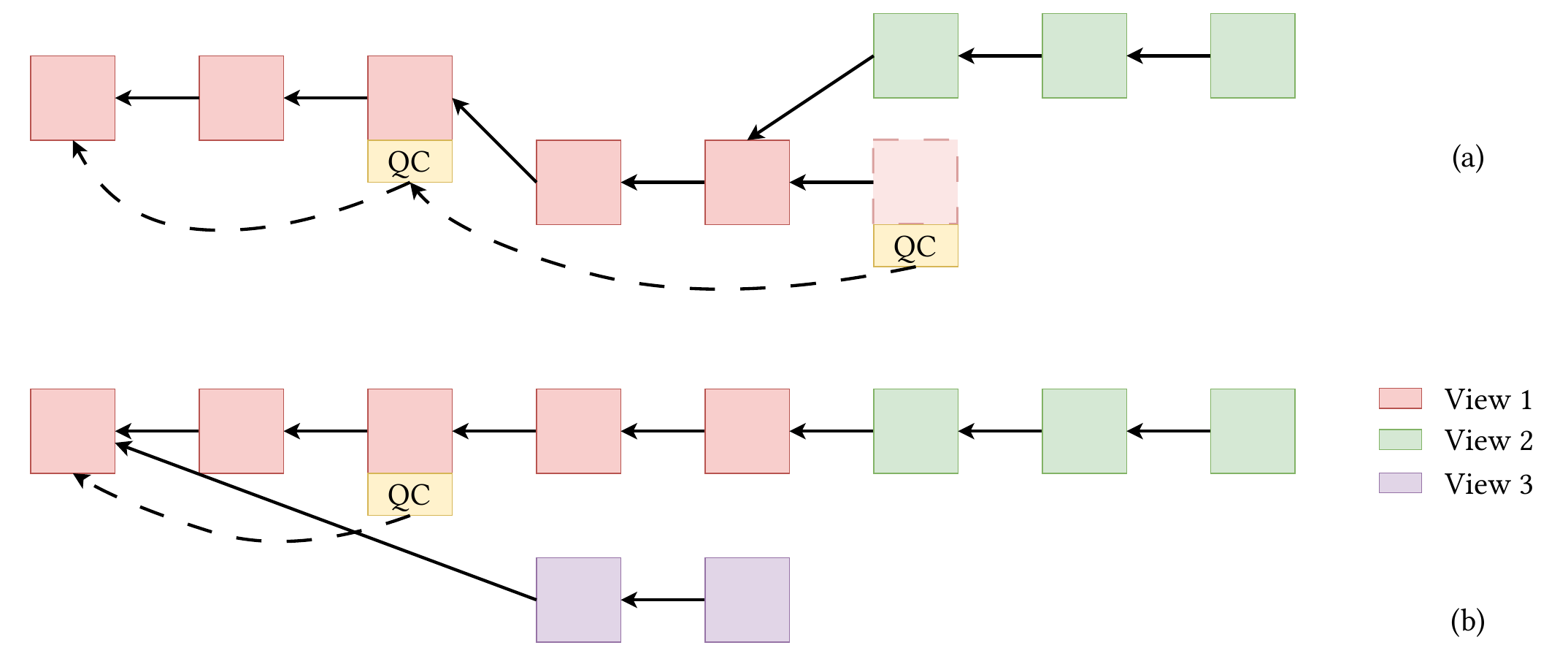}
    \caption{A Byzantine system with modified branch selection rules: (a) Branches seen for view change to view 3, (b) Branches seen for view change to view 4.}
    \label{fig:vs2}
\vspace{-0.3cm}
\end{figure*}

To understand why, consider a system with Byzantine nodes, as described in \Cref{fig:vs2}.
The leader in view forms the \qc that audits batch 1 and embeds it in batch 6.
But it times out before it could propagate the batch.
The leader in view 2 extends from batch 5 and commits 3 new batches.

During the view change to view 3, the branch from view 2's leader is chosen as it contains batch 3, the highest \qc with the highest view (but not the branch containing the \qc itself).
View 3 times out without proposing new batches.
During view change to view 4, out of $N - u$ collected branches, $f_{safe}$ show a branch that extends batch 1 directly but uses spoofed batches from view 3.
The branch selection rules accept this spoofed branch, as it contains batch 1, which has the highest \qc with the highest view.
Once the leader in view 4 enforces this branch onto all replicas, the \qc on batch 1 is forgotten.
If view 4 now times out before creating a \qc, the Byzantine replicas in view 5 can spoof a branch conflicting with view 1, and the leader in view 5 can then freely choose this conflicting branch, since there is no \qc left to "protect" batch 1.

\section{Reconfiguration}
\label{sec:reconfiguration}

Changing node membership remains an understudied problem in BFT literature (with only a few notable works, eg, ~\cite{duan2022dyno}).
Nevertheless, to make \sys{} practical, we provide a simple reconfiguration protocol that repurposes the view change logic with minor modifications.

\par \textbf{Configuration.}
A cluster, in its runtime, goes through a sequence of configurations.
Each configuration $C_i$ is defined by an ordered list of replica identities, and the values of $u$ and $f_{safe}$ corresponding to that configuration.
Each message (ie, an AppendEntries, Vote or ViewChange), along with its view, must also contain the configuration number.
A replica ignores messages from older configurations, and, like View Change, only processes messages from higher configurations if it receives enough evidence of a successful configuration change and that it is part of the newest configuration.
The list of nodes that make up the initial configuration is fixed and is public knowledge.
Any changes in the node membership thereafter need to happen through the reconfiguration protocol.

\par \textbf{Node Lifecycle.}
In its lifetime, a replica goes through different lifecycle stages, as follows:
\begin{enumerate}[leftmargin=*]
    \item \textit{Dormant:} A replica boots up in a dormant stage, before it becomes part of the cluster.
    \item \textit{Learner:} Once a dormant cluster joins the network and starts receiving messages from the leader, it becomes a learner. Learners do not vote. They are read-only replicas.
    \item \textit{Full Node:} A learner can be upgraded to a full node that takes part in the protocol.
    \item \textit{Old Full Node:} Once the cluster moves to a new configuration, all full nodes in the old configuration become Old Full Nodes. If the node is part of the new configuration, it can be both an Old Full Node and a Full Node at the same time. Otherwise, Old Full Nodes behave like learners in the new configuration and process certain messages from the old configuration.
    \item \textit{Dead:} Once an Old Full Node or a Learner is marked as Dead, it can be safely shut down and removed from the cluster.
\end{enumerate}

\par \textbf{Reconfiguration Transactions.}
\sys{} handles reconfigurations through transactions.
Since adding learners does not involve changes to the voting quorum, the transaction to add new learners is executed right after committing.
Adding learners does not change the configuration number.

Transactions to add full nodes or remove any node in the system must execute only after they are audited.
Potential equivocation in the list of voting members could be devastating.
Once removed from the cluster, a node dies, and thus if a removal transaction is rolled back, it might never recover the dead node.
Each node buffers such transactions in a queue.
The node executes a transaction once all transactions preceding it have been executed and the current view is stable.

Nodes transition through the lifecycle one stage at a time in each transaction. For example, dormant nodes become learners, learners become full nodes, and so on, but a learner cannot become an old full node, or, more importantly, old full nodes cannot become dead directly.
One transaction can only batch multiple transitions of one type.

\par \textbf{Changing set of full nodes.}
We only allow existing nodes to be upgraded to full nodes.
On executing a transaction that changes the set of full nodes to $C_{new}$, an existing full node in configuration $i$ sets the next configuration $C_{i + 1} \gets C_{new}$, assigns all nodes in $C_i$ as old full nodes and creates $C_{old} \gets C_i$, and then increments its configuration number to $i + 1$.
It then broadcasts a \textsc{View-Change} message to all nodes in $C_{new}$ and $C_{old}$.
Learners that are part of $C_{new}$ but not $C_{old}$ execute the same steps but do not trigger a View Change.

Any node, on receiving $N_{old} - u_{old}$ \textsc{View-Change} messages with configuration $i + 1$ enters a new view with configuration $i + 1$ and stops accepting any messages from lower views and configurations.
The leader in the new view, \textit{according to the new configuration}, attempts to stabilize this view by sending \textsc{New-View} messages to all nodes in $C_{old}$ and $C_{new}$ and waiting for votes from both these sets.
Unlike normal view changes, where $N_{new} - u_{new}$ votes from nodes in $C_{new}$ would have been enough, stabilizing this view requires $N_{new} - u_{new}$ votes from $C_{new}$ \textit{and} $N_{old} - u_{old}$ votes from $C_{old}$.
Once the view is stabilized, the leader can operate normally in $C_{new}$ and all nodes $C_{old}$ simply acts as learners.

\par \textbf{Removing old full nodes.}
Audited transactions to remove nodes in $C_{old}$ can be executed immediately after a view is stabilized in $C_{new}$.
It does not require changing configurations or triggering view changes.
After executing such a transaction, a node in $C_{old}$ simply goes to a dead state and removes itself from the network.

\par \textbf{Addressing liveness issues.}
The protocol for changing configurations above relies on $N_{old} - u_{old}$ nodes discovering that the reconfiguration transaction is audited.
Our BFT auditing protocol lacks a "Deliver" phase, like Hotstuff~\cite{hotstuff}, where the \qc that audits the transaction is broadcast and $N_{old} - u_{old}$ votes are collected on it.
We rely on our view change logic to re-audit the same transaction at the same log position in all future views.
However, this relies on enough nodes remaining in the old configuration for view changes in the old configuration to succeed, which is not always the case.
Concretely, as soon as $u_{old} + 1$ nodes jump to the new configuration, it is impossible for the remaining $N_{old} - u_{old} - 1$ to carry out a view change in the old configuration.
Thus these nodes may never consider the reconfiguration transaction to be audited and may never change configurations, and as a result, there may never be enough nodes to stabilize a view in the new configuration.

To this end, we introduce a new Pre-View-Change phase before invoking a view change for reconfiguration.
Once a node is about to execute a reconfiguration transaction that changes the configuration number from $i$ to $i + 1$, after auditing it in view $v$, it first broadcasts a \msg{PreViewChange}{v, i, QC} to all nodes in $C_{old}$, where $QC$ is the \qc that audits the reconfiguration transaction.
For any future view changes that happen in the old configuration, it re-sends \msg{PreViewChange}{v + k, i, QC}, where $v + k$ is the new view.
Once a node hears $N_{old} - u_{old}$ \textsc{PreViewChange} messages in any view $v + k$, it increments the configuration number to $i + 1$ and triggers the view change process as described above.
This guarantees that enough nodes in the old configuration audits the reconfiguration transaction before the view change with the new configuration change is triggered.

\par \textbf{Guarantees.}
The reconfiguration protocol guarantees the following:
\begin{enumerate}[leftmargin=*]
    \item \textbf{Safety I:} If a correct node in configuration $C_i$ considers $C_{i+k} ~(k > 0)$ to be the $(i + k)^{th}$ configuration, no other correct node in $C_i$ or any future configuration considers any other configuration $C'_{i+k}$ to be the $(i + k)^{th}$ configuration.
    \item \textbf{Safety II:} If a correct node in configuration $C_i$ considers a transaction $m$ audited at a log sequence number and another correct node in $C_{i'}$ considers another transaction $m'$ audited at the same log sequence number, then $m = m'$.
    \item \textbf{Safety III:} If a configuration $C_i$ consists of only uncompromised nodes, and if a node in configuration $C_i$ considers a transaction $m$ \textit{committed} at a log sequence number, and any other correct node in some \textit{future} configuration $C_{i+k}$ considers a transaction $m'$ committed at the same log sequence number, then $m = m'$.
    \item \textbf{Liveness:} If a correct node is part of configurations $\langle C_i, C_{i+1}, \allowbreak \ldots, C_{i+m}\rangle$, where each configuration guarantees liveness individually , and the node is currently in configuration $C_i$, eventually the node will move to configuration $C_{i+m}$.
\end{enumerate}

\section{Lower bounds on number of platforms}
\label{sec:pft-lower-bound}

In this section, we shall derive the minimum number of platforms needed to tolerate a given number of faults.

\begin{theorem}{}
    Forgoing liveness (i.e., $\pi_{live} = 0$), the minimum number of platforms, $min(\Pi)$ required to tolerate $\pi_{safe}$ faulty platforms is $\pi_{safe} + 1$.
\end{theorem}
\begin{proof}
    The minimum number of nodes to tolerate $c$ benign crashes and $f_{safe}$ commission failures is $N = 2c + f_{safe} + 1$.
    Consider the platforms $\langle 1, 2, \ldots, \pi_{safe}, \ldots, \Pi \rangle$ sorted according in descending order of their sizes (i.e., $i < j \implies N^i \geq N^j ~\forall i,j$).

    We consider the maximum possible Byzantine (commission) faulty nodes to be of the first $\pi_{safe}$ platforms.
    So $\sum_{i=1}^{\pi_{safe}} N^i = f_{safe}$ and $\sum_{i=\pi_{safe}+1}^{\Pi} N^i = N - f_{safe} = 2c + 1$.

    Let $\mu = N^1$ and $\nu = N^{\pi_{safe}}$ be the largest and smallest possible size of a faulty platform.
    We have, $\forall i > \pi_{safe}, N^i \leq \nu$. Therefore,
    $$ \sum_{i=\pi_{safe}+1}^{\Pi} N^i \leq \nu (\Pi - \pi_{safe}) $$
    $$ \implies 2c + 1 \leq \nu (\Pi - \pi_{safe}) $$
    $$ \implies \Pi \geq \frac{2c+1}{\nu} + \pi_{safe} $$

    Again, we have $\forall 1 \leq i \leq \pi_{safe}, N^i \geq \nu$. Therefore,
    $$ f_{safe} = \sum_{i=1}^{\pi_{safe}} N^i \geq \nu \pi_{safe} $$
    $$ \implies \nu \leq \frac{f_{safe}}{\pi_{safe}} $$
    $$ \implies \frac{2c+1}{\nu} \geq \frac{2c+1}{f_{safe}}\pi_{safe} $$
    $$ \implies \Pi \geq \frac{2c+1}{\nu} + \pi_{safe} \geq \left[\frac{2c+1}{f_{safe}} + 1\right]\pi_{safe} $$

    Therefore, the minimum value of $\Pi$, i.e., $min(\Pi) = \left[\frac{2c+1}{f_{safe}} + 1\right]\pi_{safe}$
    The value of $f_{safe}$ changes according to the values of $\mu$ and $\nu$.
    For example, one could envision a system where we have one node per platform ($\mu = \nu = 1$).
    It fixes $f_{safe} = \pi_{safe}$ and $min(\Pi) = \pi_{safe} + 2c + 1$, which is sound.
    However, setting $\nu = 2c + 1$, we get, $f_{safe} \geq \nu \pi_{safe} \implies min(\Pi) \leq \pi_{safe} + 1$. Since $min(\Pi)$ must be $> \pi_{safe}$, we get $min(\Pi) = \pi_{safe} + 1$.
\end{proof}

Note that this gives us a straightforward deployment strategy. Given values for $c, \pi_{safe}$, deploy $2c + 1$ nodes each in $\pi_{safe} + 1$ platforms.

To understand the case with $\pi_{live} > 0$, we first make some assumptions about the sets of platforms included in $\pi_{live}$ and $\pi_{safe}$.
Note that these sets may or may not be disjoint. Let $\pi$ denote the size of the union of these sets, i.e., the total number of Byzantine/compromised platforms, and let $r$ be the total number of nodes in these platforms.
In the worst case, the biggest $\pi$ platforms are Byzantine:
$r = \sum_{i=1}^{\pi} N^i$.
We use $N = 2c + 2f_{live} + f_{safe} + 1$.

\begin{theorem}{}
    The minimum number of platforms, $min(\Pi)$, required to tolerate $\pi_{safe}$ commission-faulty platforms and $\pi_{live}$ omission-faulty platforms is $2\pi_{live} + \pi_{safe} + 1$.
\end{theorem}
\begin{proof}
    Define $\nu = N^{\pi}$ as the minimum size of a potentially compromised platform. We have $r = \sum_{i=1}^{\pi} N^i \geq \nu \pi$. Similarly, $f_{safe} \geq \nu \pi_{safe}$ and $f_{live} \geq \nu \pi_{live}$.
    There are $N - r$ nodes in the remaining $\Pi - \pi$ platforms, each whose size can't exceed $\nu$. Thus,
    $$ N - r \leq \nu (\Pi - \pi) $$
    $$ \implies \Pi \geq \frac{N - r}{\nu} + \pi $$
    $$ \implies \Pi \geq \frac{2c + 2f_{live} + f_{safe} + 1 - r}{\nu} + \pi$$
    $$ \implies \Pi \geq \frac{2c + 1}{\nu} + \frac{2f_{live} + f_{safe} - r}{\nu} + \pi$$
    $$ \implies \Pi \geq \frac{2c + 1}{\nu} + \frac{\nu(2\pi_{live} + \pi_{safe} - \pi)}{\nu} + \pi$$
    $$ \implies \Pi \geq \frac{2c + 1}{\nu} + 2\pi_{live} + \pi_{safe} $$
    $$ \implies \Pi \geq \frac{2c + 1}{r}\pi + 2\pi_{live} + \pi_{safe} $$

    Therefore $min(\Pi) = \frac{2c + 1}{r}\pi + 2\pi_{live} + \pi_{safe}$.
    Setting $\nu = 2c + 1$, we have $r \geq \nu \pi = (2c + 1)\pi$. So, $min(\Pi) \leq 1 + 2\pi_{live} + \pi_{safe}$.
    But since $\Pi \geq \frac{2c + 1}{r}\pi + 2\pi_{live} + \pi_{safe} > 2\pi_{live} + \pi_{safe}$, the only possible value for $min(\Pi)$ is $min(\Pi) = 1 + 2\pi_{live} + \pi_{safe}$.
\end{proof}

Note that this formulation is independent of how many platforms does $\pi_{safe}$ and $\pi_{live}$ have in common.
To formulate a simple deployment strategy, note that all $\pi$ platforms must have exactly $\nu = 2c + 1$ nodes, and $f_{safe} = (2c + 1)\pi_{safe}$, $f_{live} = (2c + 1)\pi_{live}$, and $r = (2c + 1)\pi$.
In the remaining $\Pi - \pi = (2\pi_{live} + \pi_{safe} + 1 - \pi)$ platforms, the number of nodes is $N - r = (2c + 2f_{live} + f_{safe} + 1) - r = (2c + 1)(2\pi_{live} + \pi_{safe} + 1 - \pi)$.
Since each of these platforms can have at most $\nu = 2c + 1$ nodes, each of them must have exactly $2c + 1$ nodes to satisfy this equation.
So a simple deployment strategy is to deploy exactly $2c + 1$ nodes each in $2\pi_{live} + \pi_{safe} + 1$ platforms.

Minimizing the number of platforms results in more nodes than optimal.
We achieve minimum nodes by letting each platform have only one node, resulting in $N = \Pi = 2c + 2\pi_{live} + \pi_{safe} + 1$.
This gives us a Pareto frontier of number of nodes vs number of platforms.
In a real deployment scenario, one must choose the number of nodes based on the total available platforms.

\section{Details on Code Transparency Service}
\label{subsec:details-cts}

Modern software supply chains rely on automated pipelines and large dependency graphs, making a single compromised component capable of affecting thousands of downstream users.
TEEs provide attestation for the execution environment, but they do not ensure that the attested code originates from an authorized source or complies with organizational policy.
Establishing this link requires a verifiable, immutable record of software releases and their provenance.

A Code Transparency Service~\cite{scitt-ietf, Delignat23} (CTS) provides this root of trust by maintaining a ledger of signed claims describing releases, their provenance, and the policies they satisfy.
Clients and auditors can verify that a release was issued by an authorized publisher, and each accepted claim is recorded together with a receipt that can be checked offline.

Existing CTS implementations pair TEEs with CFT protocols~\cite{scitt}.
This approach inherits the limitations discussed earlier, as a single TEE compromise can undermine the ledger’s integrity, and correlated TEE failures violate the independence assumptions that CFT requires.
\syscts addresses these limitations by replacing the TEE+CFT foundation and strengthening the guarantees of the execution, providing users with audit receipts (\S\ref{sec:receipts}).
Nevertheless, CTS's policy validation required non-trivial changes to the behavior of \sys{}.

A CTS must ensure that each submitted claim satisfies the current policy.
Under a TEE+CFT design, this validation is performed exclusively by the leader, assuming its TEE is always correct.
Under PFT, this assumption does not hold, meaning a leader could propose batches containing invalid claims.
Therefore, \syscts performs validation across all replicas, following two invariants:
(1) a correct leader validates incoming claims before batching them into a block, and
(2) correct replicas independently validate every claim in a proposed batch before voting to accept it.
If correct replicas detect policy-violating claims, they simply withhold their vote, preventing the batch from gathering a quorum.

This design introduces an additional validation phase, when compared to TEE+CFT approaches, and this results in a non-trivial amount of additional cryptography-heavy computation.
To reduce this cost, since policy changes are infrequent, \syscts employs an optimistic concurrency control mechanism.
Each node, for each batch, begins by validating all claims in parallel against the most up-to-date policy available, tagging each result with the policy version applied.
Afterwards, it runs a quick sequential check to confirm that this policy version remained current throughout the validation.
In the uncommon case where a policy update occurred mid-validation, a node re-validates the batch sequentially.
Since this situation is infrequent, the cost of the deterministic fallback is minimal.

\section{TLA+ Spec for \sys{}}
\label{subsec:tlaplus}

We provide a copy of the formal specification for \sys{} in the following pages.



\section*{Supplementary material}

─────────────── MODULE *proteus* ───────────────

This is a TLA+ specification of the *Proteus* consensus protocol.
For simplicity of the specification, message delivery is assumed to be ordered and reliable.
Likewise, we also assume all transactions are signed.

EXTENDS
    TLA+ standard modules
    *Integers*,
    *Sequences*,
    *FiniteSets*,
    TLA+ community modules
    *FiniteSetsExt*,
    *SequencesExt*

CONSTANTS

Set of all replicas
CONSTANT $R$
ASSUME $R \neq \{\}$

Set of all *Byzantine* replicas
CONSTANT $BR$
ASSUME $BR \subseteq R \land Cardinality(R) \geq 3 * Cardinality(BR) + 1$

Set of all possible transactions
CONSTANT $Txs$
ASSUME $Txs \neq \{\}$

Maximum number of byzantine actions across all replicas
This parameter is completely artificial and is used to limit the state space
CONSTANT *MaxByzActions*

VARIABLES

VARIABLE
    messages in transit between any pair of replicas
    *network*,
    current view of each replica
    *view*,
    current *log* of each replica
    *log*,
    flag indicating if each replica is a primary
    *primary*,
    (primary only) the highest *log* entry on each replica replicated in this view
    *prepareQC*,

   commit index of each replica
  $commitIndex,$
   audit index of each replica
  $auditIndex,$
   total number of byzantine actions taken so far by any byzantine replica
  $byzActions,$
   (primary only) flag indicating if the primary has stabilized the view
  $viewStable$

Sequence of all variables
$vars \triangleq \langle$
  $network,$
  $view,$
  $log,$
  $primary,$
  $prepareQC,$
  $commitIndex,$
  $auditIndex,$
  $byzActions,$
  $viewStable\rangle$

---

HELPERS & VARIABLE TYPES

Number of replicas
$N \triangleq Cardinality(R)$

Set of quorums for commitment (simple majority).
$CQ \triangleq \{q \in \text{SUBSET } R : Cardinality(q) \geq N - ((N-1) \div 2)\}$

Set of quorums for auditing (super majority).
$AQ \triangleq \{q \in \text{SUBSET } R : Cardinality(q) \geq N - ((N-1) \div 3)\}$

Set of honest replicas
$HR \triangleq R \setminus BR$

Set of all views
$Views \triangleq Nat$

$ReplicaSeq \triangleq$
  $\text{CHOOSE } s \in [1 \mathrel{..} N \to R] : Range(s) = R$

The primary of view $v$
$Primary(v) \triangleq$
  $ReplicaSeq[(v\%N) + 1]$

Quorum certificates ($QCs$) are simply the index of the $log$ entry they confirm
Quorum certificates do not need views as they are always formed in the current view

Note that in the specification, we do not model signatures anywhere.
This means that signatures are omitted from the logs and messages.
When modelling byzantine faults, byz replicas will not be permitted to form
messages which would be discarded by honest replicas.

$QC \triangleq Nat$

Each *log* entry contains a view, a txn and optionally, quorum certificates for commitment and auditing

$LogEntry \triangleq [$
    $view : Views,$
    $tx : Seq(Txs),$
      For convenience, we represent a quorum certificate as a set but it can only be empty or a singleton
    $auditQC : \text{SUBSET } QC,$
    $auditQCVotes : AQ \cup \{\{\}\},$    empty set iff *auditQC* is empty.
    $commitQC : \text{SUBSET } QC]$

A *log* is a sequence of *log* entries. The index of the *log* entry is its sequence number/height
We do not explicitly model the parent relationship, the parent of *log* entry $i$ is *log* entry $i - 1$

$Log \triangleq Seq(LogEntry)$

Set of all possible *AppendEntries* messages

$AppendEntries \triangleq [$
    $type : \{\text{"AppendEntries"}\},$
    $view : Views,$
      In practice, it suffices to send only the *log* entry to append
      However, for the sake of the spec, we send the entire *log* as we need to check
      that the replica has the parent of the *log* entry to append
    $log : Log]$

Set of all possible *Vote* messages

$Votes \triangleq [$
    $type : \{\text{"Vote"}\},$
    $view : Views,$
      As with *AppendEntries*, we send the entire *log* for the sake of the spec
    $log : Log]$

Set of all possible *ViewChange* messages

$ViewChanges \triangleq [$
    $type : \{\text{"ViewChange"}\},$
    $view : Views,$
    $log : Log]$

Set of all possible *NewView* messages
Currently, we use separate messages for *NewView* and *AppendEntries*, these could be merged

$NewViews \triangleq [$
    $type : \{\text{"NewView"}\},$
    $view : Views,$
    $log : Log]$

Set of all possible messages
$Messages \triangleq$
    $AppendEntries \cup$
    $Votes \cup$
    $ViewChanges \cup$
    $NewViews$

$LogTypeOK \triangleq$
    $\wedge\ log \in [R \to Log]$
    $\wedge\ \forall\, l \in Range(log):$
        $\forall\, i \in \text{DOMAIN}\ l : l[i].auditQC = \{\} \equiv l[i].auditQCVotes = \{\}$

$NetworkTypeOK \triangleq$
    $\forall\, r, s \in R:$
        $\forall\, i \in \text{DOMAIN}\ network[r][s]:$
            $network[r][s][i] \in Messages$

Invariant to check the types of all variables
$TypeOK \triangleq$
    $\wedge$    $viewStable \in [R \to \text{BOOLEAN}]$
    $\wedge$    $view \in [R \to Views]$
    $\wedge$    $LogTypeOK$
    $\wedge$    $NetworkTypeOK$
    $\wedge$    $primary \in [R \to \text{BOOLEAN}]$
    $\wedge$    $prepareQC \in [R \to [R \to Nat]]$
    $\wedge$    $commitIndex \in [R \to Nat]$
    $\wedge$    $auditIndex \in [R \to Nat]$
    $\wedge$    $byzActions \in Nat$

---

INITIAL STATES

We begin in view 0 with a non-deterministically chosen replica as primary.
$Init \triangleq$
    Compare src/consensus/$handler.rs \neq ConsensusState$
    $\wedge\ view = [r \in R \mapsto 0]$
    $\wedge\ network = [r \in R \mapsto [s \in R \mapsto \langle\rangle]]$
    $\wedge\ log = [r \in R \mapsto \langle\rangle]$
    $\wedge\ primary \in \{f \in [R \to \text{BOOLEAN}] : Cardinality(\{r \in R : f[r]\}) = 1\}$
    $\wedge\ prepareQC = [r \in R \mapsto [s \in R \mapsto 0]]$
    $\wedge\ commitIndex = [r \in R \mapsto 0]$
    $\wedge\ auditIndex = [r \in R \mapsto 0]$
    $\wedge\ byzActions = 0$
    $\wedge\ viewStable = primary$

---

ACTIONS

$IsAuditQC(e) \triangleq$
$\quad e.auditQC \neq \{\}$

$IsCommitQC(e) \triangleq$
$\quad e.commitQC \neq \{\}$

Given a *log* $l$, returns the index of the highest *log* entry with a *commitQC*, 0 if the *log* contains no *commitQCs*
$HighestCommitQC(l) \triangleq$
$\quad$ LET $idx \triangleq SelectLastInSeq(l, IsCommitQC)$
$\quad$ IN $\quad$ IF $idx = 0$ THEN 0 ELSE $Max(l[idx].commitQC)$

Given a *log* $l$, returns the index of the highest *log* entry with a *auditQC*, 0 if the *log* contains no *auditQCs*
Compare: src/consensus/*log.rs* $\neq$ *Log*
$HighestAuditQC(l) \triangleq$
$\quad$ LET $idx \triangleq SelectLastInSeq(l, IsAuditQC)$
$\quad$ IN $\quad$ IF $idx = 0$ THEN 0 ELSE $Max(l[idx].auditQC)$

Given a *log* $l$, returns the index of the highest *log* entry with a *auditQC* over a *auditQC*
Compare: src/consensus/*log.rs* $\neq$ *Log*
$HighestQCOverQC(l) \triangleq$
$\quad$ LET $lidx \triangleq HighestAuditQC(l)$
$\quad\quad\quad idx \triangleq SelectLastInSubSeq(l, 1, lidx, IsAuditQC)$
$\quad$ IN $\quad$ IF $idx = 0$ THEN 0 ELSE $Max(l[idx].auditQC)$

Given a *log* $l$, this operator returns the highest index of a *log* entry for which $a * \text{Quorum Certificate} * (QC)$ exists. Note, the index of the *log* entry with the $QC$ corresponds to a higher *log* index than the returned index. This $QC$ is formed by unanimous $**auditQCVotes**$ from replicas.
Since a vote by a replica $r$ for some index $n$ implicitly serves as a vote for all *log* entries at index $b$ and below, the returned highest index might not have directly received a unanimous vote. Instead, replicas may have voted for this index transitively by voting for higher indices.
The pair $(idx, r)$ is a vote by replica $r$ for *log* index $idx$. While this vote may not yet be recorded in the *log*, this operator is robust against it.
$HighestUnanimity(l, idx, r) \triangleq$
$\quad$ Traverse the *log* * backwards* and record the replicas that have voted for the current $idx$ or higher indices (see $V$).
$\quad$ LET $\quad$ Include $r$'s vote in $V$ of $1..i$ if $r$ voted for index $i$.
$\quad\quad\quad V(S, i) \triangleq S \cup l[i].auditQCVotes \cup \text{IF } i \leq idx \text{ THEN } \{r\} \text{ ELSE } \{\}$
$\quad\quad\quad$ RECURSIVE $RUnanimity(\_, \_)$
$\quad\quad\quad RUnanimity(i, S) \triangleq$
$\quad\quad\quad\quad$ IF $i = 0$ THEN $\{0\}$
$\quad\quad\quad\quad$ ELSE IF $V(S, i) = R$
$\quad\quad\quad\quad\quad$ THEN $l[i].auditQC$
$\quad\quad\quad\quad\quad$ ELSE $RUnanimity(i - 1, V(S, i))$
$\quad$ IN $\quad RUnanimity(Len(l), \{\})$

$Max2(a, b) \triangleq$ IF $a > b$ THEN $a$ ELSE $b$
$Min2(a, b) \triangleq$ IF $a < b$ THEN $a$ ELSE $b$

$MaxQuorum(Q, l, m, default) \triangleq$
 LET RECURSIVE $RMaxQuorum(\_)$
  $RMaxQuorum(i) \triangleq$
   IF $i = default$ THEN $default$
   ELSE IF $\exists q \in Q : \forall n \in q : m[n] \geq i$
    THEN $i$ ELSE $RMaxQuorum(i-1)$
 IN $RMaxQuorum(Len(l))$

Checks if a *log* $l$ is well formed *e.g.* views are monotonically increasing
$WellFormedLog(l) \triangleq$
 $\forall i \in \text{DOMAIN } l :$
  check views are monotonically increasing
  $\land i > 1 \Rightarrow l[i-1].view \leq l[i].view$
  check *auditQCs* are well formed
  $\land \forall q \in l[i].auditQC :$
   *auditQCs* are always for previous entries
   $\land q < i$
   *auditQCs* are always formed in the current view
   $\land l[q].view = l[i].view$
   *auditQCs* are in increasing order
   $\land \forall j \in 1 .. i-1 :$
    $\forall qj \in l[j].auditQC : qj < q$
  check *commitQCs* are well formed
  $\land \forall q \in l[i].commitQC :$
   *commitQCs* are always for previous entries
   $\land q < i$
   *commitQCs* are in increasing order
   $\land \forall j \in 1 .. i-1 :$
    $\forall qj \in l[j].commitQC : qj < q$

Replica $r$ handling *AppendEntries* from primary $p$
Compare: src/consensus/$steady\_state$.rs $\neq do\_append\_entries$
$ReceiveEntries(r, p) \triangleq$
 there must be at least one message pending
 $\land network[r][p] \neq \langle\rangle$
 and the next message is an *AppendEntries*
 $\land Head(network[r][p]).type = \text{"AppendEntries"}$
 the replica must be in the same view
 $\land view[r] = Head(network[r][p]).view$
 and must be replicating an entry from this view
 $\land Last(Head(network[r][p]).log).view = view[r]$
 the replica only appends (one entry at a time) to its *log*
 $\land log[r] = Front(Head(network[r][p]).log)$
 received *log* must be well formed
 $\land WellFormedLog(Head(network[r][p]).log)$

for convenience, we replace the replica's *log* with the received *log* but in practice we are only appending one entry
$\land\ log' = [log\ \text{EXCEPT}\ ![r] =\ Head(network[r][p]).log]$
we remove the *AppendEntries* message and reply with a *Vote* message.
$\land\ network' = [network\ \text{EXCEPT}$
  $\ ![r][p]\ = Tail(@),$
  $\ ![p][r]\ = Append(@, [$
    $\ \ type \mapsto \text{"Vote"},$
    $\ \ view \mapsto view[r],$
    $\ \ log \mapsto log'[r]$
    $\ ])$
  $]$

replica updates its commit index provided the new commit index is greater than the current one
the only time a commit index can decrease is on the receipt of a *NewView* message if there's been a byz attack
$\land\ commitIndex' = [commitIndex\ \text{EXCEPT}\ ![r] = Max2(@,\ HighestCommitQC(log'[r]))]$
assumes that a replica can safely consider a transaction audited if there's a quorum certificate over a quorum certificate
Compare: src/consensus/*commit.rs* ≠ *maybe_byzantine_commit*
$\land\ \text{LET}\ bci\ \triangleq\ HighestQCOverQC(log'[r])$
  Compare: src/consensus/*commit.rs* ≠ *maybe_byzantine_commit_by_fast_path*
  $bciFastPath\ \triangleq\ HighestUnanimity(log'[r],\ 0,\ r)$
$\ \text{IN}\ \ auditIndex' = [auditIndex\ \text{EXCEPT}\ ![r] = Max(\{@\} \cup \{bci\} \cup bciFastPath)]$
$\land\ \text{UNCHANGED}\ \langle primary,\ view,\ prepareQC,\ byzActions,\ viewStable \rangle$

Replica *r* handling *NewView* from primary *p*

Note that unlike an *AppendEntries* message, a replica can update its view upon receiving a *NewView* message
$ReceiveNewView(r,\ p)\ \triangleq$
  there must be at least one message pending
  $\land\ network[r][p] \neq \langle\rangle$
  and the next message is a *NewView*
  $\land\ Head(network[r][p]).type = \text{"NewView"}$
  the replica must be in the same view or lower
  $\land\ view[r] \leq Head(network[r][p]).view$
  received *log* must be well formed
  $\land\ WellFormedLog(Head(network[r][p]).log)$
  update the replica's local view
  note that we do not dispatch a view change message as a primary has already been elected
  $\land\ view' = [view\ \text{EXCEPT}\ ![r] = Head(network[r][p]).view]$
  step down if replica was a primary
  $\land\ primary' = [primary\ \text{EXCEPT}\ ![r] = \text{FALSE}]$
  $\land\ viewStable' = [viewStable\ \text{EXCEPT}\ ![r] = \text{FALSE}]$
  reset *prepareQCs*, in case view was updated
  $\land\ prepareQC' = [prepareQC\ \text{EXCEPT}\ ![r] = [s \in R \mapsto 0]]$
  the replica replaces its *log* with the received *log*
  $\land\ log' = [log\ \text{EXCEPT}\ ![r] =\ Head(network[r][p]).log]$
  we remove the *NewView* message and reply with a *Vote* message.
  $\land\ network' = [network\ \text{EXCEPT}$

```
            ![r][p]  = Tail(@),
            ![p][r]  = Append(@, [
                type  ↦ "Vote",
                view  ↦ view[r],
                log   ↦ log'[r]
                ])
        ]
    ⸺ replica must update its commit index
    ⸺ Commit index may be decreased if there's been an byz attack
    ∧ commitIndex' = [commitIndex EXCEPT ![r] = Min2(@, Len(log'[r]))]
    ∧ UNCHANGED ⟨byzActions, auditIndex⟩

⸺ True iff primary p is in a stable view
⸺ A view is stable when a audit quorum have the view's first log entry
CheckViewStability(p) ≜
    LET inView(e) ≜ e.view = view[p] IN
    ∃ Q ∈ AQ :
        ∀ q ∈ Q :
            prepareQC'[p][q] ≥ SelectInSeq(log[p], inView)

⸺ Primary p receiving votes from replica r
⸺ Compare: src/consensus/steady_state.rs ≠ do_process_vote
ReceiveVote(p, r) ≜
    ⸺ p must be the primary
    ∧ primary[p]
    ⸺ and the next message is a vote from the correct view
    ∧ network[p][r] ≠ ⟨⟩
    ∧ Head(network[p][r]).type = "Vote"
    ∧ view[p] = Head(network[p][r]).view
    ∧ prepareQC' = [prepareQC EXCEPT
        ![p][r] = IF @ ≤ Len(Head(network[p][r]).log)
        THEN Len(Head(network[p][r]).log)
        ELSE @]
    ⸺ we remove the Vote message.
    ∧ network' = [network EXCEPT ![p][r] = Tail(network[p][r])]
    ⸺ if view is not already stable, check if it is now
    ∧ viewStable' = [viewStable EXCEPT ![p] =
            IF @ THEN @ ELSE CheckViewStability(p)]
    ⸺ If view is stable, then the primary can update its commit indexes
    ∧ IF viewStable'[p] THEN
        ∧ commitIndex' = [commitIndex EXCEPT ![p] =
            MaxQuorum(CQ, log[p], prepareQC'[p], @)]
        ⸺ Compare: src/consensus/commit.rs ≠ maybe_byzantine_commit
        ∧ LET bci ≜ HighestAuditQC(SubSeq(log[p], 1, MaxQuorum(AQ, log[p], prepareQC'[p], 0)))
            ⸺ Compare: src/consensus/commit.rs ≠ maybe_byzantine_commit_by_fast_path
```

$$bciFastPath \triangleq HighestUnanimity(log[p], prepareQC'[p][r], r)$$
$$\text{IN} \quad auditIndex' = [auditIndex \text{ EXCEPT } ![p] = Max(\{@\} \cup bciFastPath \cup \{bci\})]$$
$$\text{ELSE UNCHANGED } \langle commitIndex, auditIndex \rangle$$
$$\wedge \text{ UNCHANGED } \langle view, log, primary, byzActions \rangle$$

$MaxCommitQC(l, p) \triangleq$
  IF $commitIndex[p] > HighestCommitQC(l)$
    THEN $\{commitIndex[p]\}$
    ELSE $\{\}$

$MaxAuditQC(l, m) \triangleq$
  LET $idx \triangleq MaxQuorum(AQ, l, m, 0)$ IN
  IF $idx > HighestAuditQC(l)$
    THEN $[n \mapsto \{idx\}, v \mapsto \{r \in \text{DOMAIN } m : m[r] \geq idx\}]$
    ELSE $[n \mapsto \{\}, v \mapsto \{\}]$

Primary $p$ sends *AppendEntries* to all replicas
Compare: src/consensus/*steady_state.rs* $\neq$ *do_append_entries*
$SendEntries(p) \triangleq$
  $p$ must be the primary
  $\wedge\ primary[p]$
  and view must be stable
  $\wedge\ viewStable[p]$
  $\wedge\ \exists\ tx \in Txs :$
    primary will not send an *appendEntries* to itself so update *prepareQC* here
    $\wedge\ prepareQC' = [prepareQC \text{ EXCEPT } ![p][p] = Len(log[p]) + 1]$
    add the new entry to the *log*
    $\wedge$ LET $qc \triangleq MaxAuditQC(log[p], prepareQC'[p])$ IN
      $log' = [log \text{ EXCEPT } ![p] = Append(@, [$
      $view \mapsto view[p],$
        for simplicity, each txn batch includes a single txn
      $tx \mapsto \langle tx \rangle,$
      $commitQC \mapsto MaxCommitQC(log[p], p),$
      $auditQC \mapsto qc.n,$
      $auditQCVotes \mapsto qc.v])]$
    $\wedge\ network' =$
      $[r \in R \mapsto [s \in R \mapsto$
        IF $s \neq p \vee r = p$ THEN $network[r][s]$ ELSE $Append(network[r][s], [$
          $type \mapsto \text{"AppendEntries"},$
          $view \mapsto view[p],$
          $log \mapsto log'[p]])]]$
  $\wedge \text{ UNCHANGED } \langle view, primary, commitIndex, auditIndex, byzActions, viewStable \rangle$

Replica $r$ times out
Compare: src/consensus/*view_change.rs* $\neq$ *do_init_view_change*
$Timeout(r) \triangleq$

32

$\land \quad view' = [view \text{ EXCEPT } ![r] = view[r] + 1]$

*send a view change message to the new primary (even if it's itself)*
$\land \ network' = [network \text{ EXCEPT } ![Primary(view'[r])][r] = Append(@, [$
  $type \mapsto \text{"ViewChange"},$
  $view \mapsto view'[r],$
  $log \mapsto log[r]])$
$]$

*step down if replica was a primary*
$\land \ primary' = [primary \text{ EXCEPT } ![r] = \text{FALSE}]$
$\land \ viewStable' = [viewStable \text{ EXCEPT } ![r] = \text{FALSE}]$
*reset prepareQCs, these are not used until the node is elected primary*
$\land \ prepareQC' = [prepareQC \text{ EXCEPT } ![r] = [s \in R \mapsto 0]]$
$\land \text{ UNCHANGED } \langle log, \ commitIndex, \ auditIndex, \ byzActions \rangle$

*The view of the highest auditQC in log l, −1 if log contains no qcs*
$HighestQCView(l) \triangleq$
  LET $idx \triangleq HighestAuditQC(l)$ IN
  IF $idx = 0$ THEN $-1$ ELSE $l[idx].view$

*True if log l is valid log choice from the set of logs ls.*
*Assumes that $l \in ls$*
*Compare: src/consensus/view_change.rs ≠ fork_choice_rule_get*
$LogChoiceRule(l, ls) \triangleq$
  *if all logs are empty, then any l must be empty and a valid choice*
  $\lor \ \forall \ l2 \in ls : l2 = \langle \rangle$
  $\lor \ \land \ l \neq \langle \rangle$
    *l is valid if all other logs in ls are empty or l is from a higher view or*
    $\land$ LET $v1 \triangleq HighestQCView(l)$
    IN $\forall \ l2 \in ls :$
      *l is valid if all other logs in ls are empty or ...*
      $l \neq l2 \land l2 \neq \langle \rangle$
      $\Rightarrow$ LET $v2 \triangleq HighestQCView(l2)$ IN
        *l is from a higher view or ...*
        $\lor \ v1 > v2$
        *l is from the same view but at least as long*
        $\lor \ \land \ v1 = v2$
          $\land \ \lor \ Last(l).view > Last(l2).view$
            $\lor \ \land \ Last(l).view = Last(l2).view$
              $\land \ Len(l) \geq Len(l2)$

*Replica r becomes primary*
*Compare: src/consensus/view_change.rs ≠ do_init_new_leader*
$BecomePrimary(r) \triangleq$
  *replica must be assigned the new view*
  $\land \ r = Primary(view[r])$
  *a audit quorum must have voted for the replica*

33

$\land \exists q \in AQ:$
    $\land \forall n \in q:$
        $\land network[r][n] \neq \langle\rangle$
        $\land Head(network[r][n]).type =$ "ViewChange"
        $\land view[r] = Head(network[r][n]).view$
    $\land \exists l1 \in \{Head(network[r][n]).log : n \in q\}:$
        Non-deterministically pick a *log* from the set of logs in the quorum that satisfy the *log* choice rule.
        $\land LogChoiceRule(l1, \{Head(network[r][n]).log : n \in q\})$
        Primary adopts chosen *log* and adds a new entry in the new view
        $\land log' = [log \text{ EXCEPT } ![r] = Append(l1, [$
            $view \mapsto view[r],$
            $tx \mapsto \langle\rangle,$
            $commitQC \mapsto \{\},$
            $auditQC \mapsto \{\},$
            $auditQCVotes \mapsto \{\}])]$
        $\land prepareQC' = [prepareQC \text{ EXCEPT } ![r][r] = Len(log'[r])]$
        Need to update network to remove the view change message and send a *NewView* message to all replicas
        $\land network' = [r1 \in R \mapsto [r2 \in R \mapsto$
          IF $r1 = r \land r2 \in q$
          THEN $Tail(network[r1][r2])$
          ELSE IF $r1 \neq r \land r2 = r$
            THEN $Append(network[r1][r2], [$
              $type \mapsto$ "NewView",
              $view \mapsto view[r],$
              $log \mapsto log'[r]])$
            ELSE $network[r1][r2]]]$
replica becomes a primary
$\land primary' = [primary \text{ EXCEPT } ![r] = \text{TRUE}]$
primary updates its commit indexes
Commit index may be decreased if there's been an byz attack
$\land commitIndex' = [commitIndex \text{ EXCEPT}$
    $![r] = Max2(Min2(@, Len(log'[r])), HighestCommitQC(log'[r]))]$
$\land$ UNCHANGED $\langle view, byzActions, auditIndex, viewStable\rangle$

Replicas will discard messages from previous views or extra view changes messages
Note that replicas must always discard messages as the pairwise channels are ordered
so a replica may need to discard an out-of-date message to process a more recent one
$DiscardMessages \triangleq$
    $\land \exists s, r \in R:$
        $network' = [network \text{ EXCEPT } ![r][s] = SelectSeq(@,$
          LAMBDA $m : \neg(m.view < view[r] \lor$
            $(m.view = view[r] \land m.type =$ "ViewChange" $\land primary[r])))]$
    $\land$ UNCHANGED $\langle view, log, primary, prepareQC, commitIndex, auditIndex, byzActions, viewStable\rangle$

BYZANTINE ACTIONS
Byzantine actions can only be taken by byzantine replicas ($BR$) and if there are byzantine actions left to take

A byzantine replica might vote for an entry without actually appending it to its *log*.
This byzantine action currently has the same preconditions as *AppendEntries*
$ByzOmitEntries(r, p) \triangleq$
 $\land r \in BR$
 $\land byzActions < MaxByzActions$
 $\land byzActions' = byzActions + 1$
  there must be at least one message pending
 $\land network[r][p] \neq \langle\rangle$
  and the next message is an *AppendEntries*
 $\land Head(network[r][p]).type =$ "AppendEntries"
  the replica must be in the same view
 $\land view[r] = Head(network[r][p]).view$
  the replica only appends one entry to its *log*
 $\land log[r] = Front(Head(network[r][p]).log)$
  we remove the *AppendEntries* message and reply with a *Vote* message.
 $\land network' = [network$ EXCEPT
   $![r][p] = Tail(@),$
   $![p][r] = Append(@, [$
     $type \mapsto$ "Vote",
     $view \mapsto view[r],$
     $log \mapsto Head(network[r][p]).log$
     $])$
   $]$
 $\land$ UNCHANGED $\langle primary, view, prepareQC, commitIndex, auditIndex, log, viewStable \rangle$

Given an append entries message, returns the same message with the txn changed to 1
$ModifyAppendEntries(m) \triangleq [$
  $type \mapsto$ "AppendEntries",
  $view \mapsto m.view,$
  $log \mapsto SubSeq(m.log, 1, Len(m.log) - 1) \circ$
    $\langle [Last(m.log)$ EXCEPT $!.tx = \langle 1 \rangle] \rangle$
$]$

We allow a byzantine primary to equivocate by changing the txn in an *AppendEntries* message
$ByzPrimaryEquivocate(p) \triangleq$
 $\land p \in BR$
 $\land byzActions < MaxByzActions$
 $\land byzActions' = byzActions + 1$
 $\land \exists r \in R :$
   $\land network[r][p] \neq \langle\rangle$
   $\land Head(network[r][p]).type =$ "AppendEntries"
   $\land Head(network[r][p]).log \neq \langle\rangle$

35

$\quad\quad\quad\quad\wedge\ network' = [network\ \text{EXCEPT}$
$\quad\quad\quad\quad\quad\ ![r][p][1] = ModifyAppendEntries(@)]$
$\quad\quad\wedge\ \text{UNCHANGED}\ \langle view,\ log,\ primary,\ prepareQC,\ commitIndex,\ auditIndex,\ viewStable \rangle$

Next state relation

Note that the byzantine actions are included here but can be disabled by setting *MaxByzActions* to 0 or *BR* to {}.

$Next\ \triangleq$
$\quad \vee\ DiscardMessages$
$\quad \vee\ \exists\, r \in BR :$
$\quad\quad \vee\ ByzPrimaryEquivocate(r)$
$\quad\quad \vee\ \exists\, s \in R :$ Could be *CR* because we don't need byz replicas to receive messages from other byz replicas
$\quad\quad\quad ByzOmitEntries(r, s)$
$\quad \vee\ \exists\, r \in R :$
$\quad\quad \vee\ SendEntries(r)$
$\quad\quad \vee\ Timeout(r)$
$\quad\quad \vee\ BecomePrimary(r)$
$\quad\quad \vee\ \exists\, s \in R :$
$\quad\quad\quad \vee\ ReceiveEntries(r, s)$
$\quad\quad\quad \vee\ ReceiveVote(r, s)$
$\quad\quad\quad \vee\ ReceiveNewView(r, s)$

$Fairness\ \triangleq$
$\quad$ Only *Timeout* if there is no primary.
$\quad \wedge\ \text{WF}_{vars}(DiscardMessages)$
$\quad \wedge\ \forall\, r \in HR : \text{WF}_{vars}(\text{TRUE} \notin Range(primary) \wedge Timeout(r))$
$\quad \wedge\ \forall\, r \in HR : \text{WF}_{vars}(BecomePrimary(r))$
$\quad \wedge\ \forall\, r \in HR : \text{WF}_{vars}(SendEntries(r))$
$\quad \wedge\ \forall\, r, s \in HR : \text{WF}_{vars}(ReceiveEntries(r, s))$
$\quad \wedge\ \forall\, r, s \in HR : \text{WF}_{vars}(ReceiveVote(r, s))$
$\quad \wedge\ \forall\, r, s \in HR : \text{WF}_{vars}(ReceiveNewView(r, s))$
$\quad$ Omit any byzantine actions from the fairness condition.

$Spec\ \triangleq$
$\quad \wedge\ Init$
$\quad \wedge\ \Box[Next]_{vars}$
$\quad \wedge\ Fairness$

---

PROPERTIES

Correct replicas are either honest or byzantine when no byzantine actions have been taken yet

$CR\ \triangleq\ \text{IF}\ byzActions = 0\ \text{THEN}\ R\ \text{ELSE}\ HR$

$Committed(r)\ \triangleq$
$\quad SubSeq(log[r],\ 1,\ commitIndex[r])$

The view of a *log* entry is always greater than or equal to the view of the previous entry, *i.e.*,

36

the view of *log* entries is (non-strictly) monotonically increasing.
$ViewMonotonicInv \triangleq$
$\quad \forall\, r \in R :$
$\quad\quad \forall\, i \in 2\,..\,Len(log[r]) :$
$\quad\quad\quad log[r][i].view \geq log[r][i-1].view$

Every view starts with a view stabilization *log* entry. Moreover, view 0 is always stable.
Therefore, view 0 has no view stabilization *log* entry.
$ViewStabilizationInv \triangleq$
$\quad \forall\, r \in R :$
$\quad\quad \wedge\, \forall\, i \in 1\,..\,Len(log[r]) :$
$\quad\quad\quad \wedge\, log[r][i].tx = \langle\rangle \Rightarrow log[r][i].view \neq 0$
$\quad\quad\quad \wedge\, i > 1 \wedge log[r][i].view > log[r][i-1].view \Rightarrow log[r][i].tx = \langle\rangle$

Ignoring view stabilization *log* entries (modeled as empty txs), true iff the *log* $p$ is a prefix of *log* $l$.
$IsPrefixWithoutEmpty(p,\, l) \triangleq$
$\quad$ $p$ can be longer than $l$. Suppose $l$ matches $p$ as a prefix up to index $i$, but the suffix of $p$ starting at $i+1$ contains only view stabilization *log* entries. By adding the condition $Len(p) \leq Len(l)$, we ensure that such cases are not considered as $p$ being a prefix of $l$. Instead, we require that $l$ is at least as long as $p$, ensuring that $l$ has a suffix distinct from $p$.
$\quad$ Independently, this condition prevents out-of-bounds access when $p$ is longer than $l$. For example, if $l = \langle\rangle$ (an empty sequence), attempting to access $l[k]$ in the disjunct $p[k] = l[k]$ would lead to an out-of-bounds access.
$\quad \wedge\, Len(p) \leq Len(l)$
$\quad \wedge\, \forall\, k \in 1\,..\,Len(p) :$
$\quad\quad \vee\, p[k] = l[k]$
$\quad\quad \vee\, p[k].tx = \langle\rangle$

If no byzantine actions have been taken, then the committed logs of all replicas must be prefixes of each other
This, together with *CommittedLogAppendOnlyProp*, is the classic *CFT* safety property
Note that if any nodes have been byzantine, then this property is not guaranteed to hold on any node
*LogInv* implies that the audited logs of replicas are prefixes too,
as *IndexBoundsInv* ensures that the *auditIndex* is always less than or equal to the *commitIndex*.
$LogInv \triangleq$
$\quad byzActions = 0 \Rightarrow$
$\quad\quad \forall\, i, j \in R :$
$\quad\quad\quad \vee\, IsPrefixWithoutEmpty(Committed(i),\, Committed(j))$
$\quad\quad\quad \vee\, IsPrefixWithoutEmpty(Committed(j),\, Committed(i))$

$Audited(r) \triangleq$
$\quad SubSeq(log[r],\, 1,\, auditIndex[r])$

Variant of *LogInv* for the audit index and correct replicas only
We make no assertions about the state of byzantine replicas
$AuditLogInv \triangleq$
$\quad \forall\, i, j \in CR :$

$\quad\quad\quad \lor \textit{IsPrefix}(\textit{Audited}(i),\ \textit{Audited}(j))$
$\quad\quad\quad \lor \textit{IsPrefix}(\textit{Audited}(j),\ \textit{Audited}(i))$

If no byzantine actions have been taken, then each replica only appends to its committed *log*
Note that this invariant allows empty blocks (sent at the start of a view) to be rolled back
$\textit{CommittedLogAppendOnlyProp} \triangleq$
$\quad \Box[\textit{byzActions} = 0 \Rightarrow$
$\quad\quad \forall\, i \in R :$
$\quad\quad \textit{IsPrefixWithoutEmpty}(\textit{Committed}(i),\ \textit{Committed}(i)')]_{\textit{vars}}$

All correct replicas only append to their audited logs
$\textit{MonotonicAuditedIndexdProp} \triangleq$
$\quad \Box[\forall\, i \in \textit{CR} :$
$\quad\quad \textit{auditIndex}[i] \le \textit{auditIndex}'[i]]_{\textit{vars}}$

Each correct replica only appends to its audited *log*
$\textit{AuditedLogAppendOnlyProp} \triangleq$
$\quad \Box[\forall\, i \in \textit{CR} :$
$\quad\quad \textit{IsPrefix}(\textit{Audited}(i),\ \textit{Audited}(i)')]_{\textit{vars}}$

At most one correct replica is primary in a view
$\textit{OneLeaderPerTermInv} \triangleq$
$\quad \forall\, v \in 0\,..\,\textit{Max}(\textit{Range}(\textit{view})),\ r \in \textit{CR} :$
$\quad\quad \textit{view}[r] = v \land \textit{primary}[r]$
$\quad\quad \Rightarrow \forall\, s\ \in R \setminus \{r\} : \textit{view}[s] = v \Rightarrow \neg \textit{primary}[s]$

The commit and audit indexes are within bounds
The audit index is always less than or equal to the commit index
$\textit{IndexBoundsInv} \triangleq$
$\quad \forall\, r \in \textit{CR} :$
$\quad\quad \land \textit{commitIndex}[r] \le \textit{Len}(\textit{log}[r])$
$\quad\quad \land \textit{auditIndex}[r] \le \textit{commitIndex}[r]$

The *log* of each replica is well formed
$\textit{WellFormedLogInv} \triangleq$
$\quad \forall\, r \in \textit{CR} : \textit{WellFormedLog}(\textit{log}[r])$



\end{document}